\documentclass[a4paper,onecolumn,11pt,unpublished]{quantumarticle}
\pdfoutput=1
\usepackage[utf8]{inputenc}
\usepackage[english]{babel}
\usepackage[T1]{fontenc}

\usepackage{amsmath, amsthm, amsfonts, amssymb}
\usepackage{bm}
\usepackage{braket}
\usepackage{framed}
\usepackage{hyperref}
\usepackage{comment}
\usepackage{enumitem}
\usepackage{diagbox}
\usepackage{makecell}
\usepackage{pdflscape}
\usepackage{mathtools}
\usepackage{pgf}
\usepackage{float}
\usepackage{tikz}
\usepackage{circuitikz}
\usepackage{caption}
\usepackage{tablists}
\allowdisplaybreaks
\usepackage{lipsum}
\usepackage{csquotes}
\usepackage{xcolor}
\usepackage{underscore}
\usetikzlibrary{calc,arrows.meta,decorations.markings,arrows.meta,positioning}

\definecolor{ngreen}{rgb}{0.2,0.7,0.2}
\definecolor{nred}{rgb}{0.9,0.1,0}
\definecolor{npurple}{rgb}{1,0,1}
\definecolor{nblue}{rgb}{0.1,0.2,0.8}
\definecolor{lightblue}{rgb}{0.67, 0.84, 0.90}

\usepackage{wrapfig}

\newtheorem{theorem}{Theorem}[section]
\newtheorem{proposition}[theorem]{Proposition}

\newtheorem{lemma}[theorem]{Lemma}
\newtheorem*{theorem*}{Theorem}

\theoremstyle{definition}
\newtheorem{definition}[theorem]{Definition}

\theoremstyle{remark}
\newtheorem{remark}[theorem]{Remark}

\newtheorem{example}[theorem]{Example}

\newcommand{\N}{\mathbb{N}}
\newcommand{\Z}{\mathbb{Z}}

\newcommand{\R}{\mathbb{R}}
\newcommand{\T}{\mathbb{T}}

\newcommand{\M}{\mathcal{M}}
\newcommand{\XX}{\mathcal{X}}
\newcommand{\A}{\mathcal{A}}
\newcommand{\B}{\mathcal{B}}
\newcommand{\J}{\mathcal{J}}
\newcommand{\K}{\mathcal{K}}
\newcommand{\I}{\mathcal{I}}

\newcommand{\BB}{\mathrm{B}}
\newcommand{\GHZ}{\mathrm{GHZ}}
\NewDocumentCommand\Bstate{D<>{\bm{\lambda}} O{\Phi}}{\ket{\BB(#1, #2)}}
\NewDocumentCommand\BstateZero{D<>{\bm{\lambda}} O{0}}{\ket{\BB(#1, #2)}}
\NewDocumentCommand\BstateNoKet{D<>{\bm{\lambda}} O{\Phi}}{\BB(#1, #2)}
\NewDocumentCommand\Bsimple{D<>{\lambda}}{\ket{\BB(#1)}}
\newcommand{\piovertwointerval}{\left[0, \piovertwo\right)}
\newcommand{\piovertwo}{\tfrac{\pi}{2}}

\newcommand{\pioverN}{\tfrac{\pi}{N}}
\newcommand{\piovertwoN}{\tfrac{\pi}{2N}}

\newcommand{\zz}{\mathbf{z}}
\newcommand{\ww}{\mathbf{w}}
\newcommand{\yy}{\mathbf{y}}
\newcommand{\HH}{H}

\newcommand{\dd}{d}
\newcommand{\tickinds}{\mathsf{t}}

\newcommand{\Tpair}{\mathbf{T}}

\newcommand{\BPP}{\mathsf{BiCPP}}

\newcommand{\Odd}{\operatorname{Odd}}

\usepackage[numbers,sort&compress]{natbib}

\begin{document}

\title{Three-qubit nonlocality paradoxes: beyond GHZ}

\author[1]{Nadish de Silva}\email{ndesilva@sfu.ca}

\author[1]{Santanil Jana}\email{santanil_jana@sfu.ca}

\author[1]{Ming Yin}\email{ming_yin_2@sfu.ca}

\affil[1]{Department of Mathematics, Simon Fraser University, Burnaby, BC, Canada}

\maketitle

\begin{abstract}
    Quantum nonlocality paradoxes, such as that of GHZ, provide maximally sharp logical obstructions to classical probabilistic models of quantum correlations. They are key resources in a broad variety of information-theoretic tasks that exhibit unconditional quantum advantage. For example, in nonlocal games, which are communication tasks that serve as core technical tools in recent landmark results in quantum computational complexity theory. 

Their role in establishing quantum advantage motivated their study by Abramsky et al.\ who introduced an infinite family of three-qubit paradoxes exhibiting novel conditional structure. This was later extended by the present authors into a full classification program. 

In this work, we completely classify all three-qubit nonlocality paradoxes established via a \emph{biconditional parity proof}; this is a very large class of paradoxes that encompasses all earlier-known examples. We do this by introducing a suite of new structural and combinatorial techniques.  We find that the landscape of nonlocality paradoxes is far richer than previously understood, violating regularity conditions underlying all prior constructions.
\end{abstract}

\section{Introduction}
\label{sec:1}

A central organising thrust of quantum computer science is to understand and exploit uniquely quantum resources towards provable advantages over classical information processing.  A key role in this program is played by \emph{nonlocality} \cite{Bell1964, Bell1966}, a property of correlations that obstructs simulation by classical probabilistic models.  Nonlocality can be precisely certificated in terms of games, which serve as core technical tools in quantum information and complexity theory.  As structurally clarified by Abramsky--Brandenburger \cite{Abramsky2011}, nonlocality (and its generalisation, contextuality) are essentially logical in character.  This work addresses the problem of classifying logical proofs of nonlocality under minimal assumptions.

Whereas standard witnesses of nonlocality take the form of probabilistic inequalities satisfied by classical correlations but violable by quantum correlations \cite{CHSH}, the maximal form of \emph{strong nonlocality} is witnessed by a logical paradox.  Nonlocality paradoxes arise from restricted joint measurability: quantum theory allows only compatible sets of measurements, called \emph{contexts}, to be jointly measurable.  The canonical example is due to Greenberger--Horne--Zeilinger (GHZ) \cite{Greenberger_1990, Greenberger_1989}.  They gave a three-qubit state with two measurements per qubit such that: the set of constraints describing the associated possibilistic empirical data, one for each context (i.e.\ a choice of one measurement on each qubit), is unsatisfiable.

Paradoxes undergird probabilistic nonlocality in that all nonlocal correlations are dilutions of paradoxes by classical randomness. The nonlocal fraction \cite{Barrett_2006, Elitzur_1992} measures nonlocality by quantifying the degree to which correlations are paradoxical.  In many information-theoretic tasks, quantum advantage scales directly with the nonlocal fraction of an available resource state.  For example, \emph{nonlocal games} \cite{Cleve_2004} involve players---who may share a distributed resource but cannot directly communicate---giving coordinated responses to a referee's queries.  When classical correlations limit the players' success probability, deterministic winning strategies necessarily require sharing a nonlocal paradox.

Such unconditional advantages in communication tasks can be converted into rare unconditional complexity-theoretic separations.  For example, Bravyi--Gosset--K\"onig  \cite{Bravyi_2018} gave a constant-depth quantum circuit family, capable of generating and playing many nonlocal games, that cannot be simulated by classical circuits of sublogarithmic depth.  Nonlocal games also play an essential role in the proof that $\textsf{MIP}^* = \textsf{RE}$: entangled multi-interactive provers recognise exactly the semidecidable languages \cite{ji2021mip}.

Nonlocality paradoxes directly boost computation in the model of measurement-based quantum computation.  Anders--Browne \cite{Anders_2009} showed how access to measurements on GHZ states can promote a restricted linear computer to classical universality; Raussendorf \cite{Raussendorf_2013} extended this to general measurement-based computation.  Other information-theoretic tasks requiring nonlocality include device-independent quantum cryptography \cite{colbeck2009quantum,Grasselli_2023} and randomness generation \cite{acin2016certified}.

These results demonstrate that the mere presence of nonlocality is insufficient for realising advantages; the precise \emph{structure} of the paradoxes involved plays a vital role. Revealing this structure enables systematic investigation of the advantages such paradoxes confer and the mechanisms by which they do so. Moreover, understanding how abundant quantum paradoxes are, and whether simple systems can generate many independent instances of them, bears directly on how nonlocality can be amplified into provable computational separations.  These considerations naturally lead to a classification problem for nonlocality paradoxes, initiated in earlier work by the present authors~\cite{de_Silva_2025} in the three-qubit setting.

\subsection{Prior work}
A \emph{three-qubit nonlocality paradox} is a pair $(\ket{\psi}, \mathcal{M})$ where $\ket{\psi} \in (\mathbb{C}^2)^{\otimes 3}$ is a quantum state and $\mathcal{M} = (M_1, M_2, M_3)$ is a triple of local measurements for three parties (Alice, Bob, and Charlie), giving a measurement set $\XX = M_1 \sqcup M_2 \sqcup M_3$, such that for every \emph{hidden variable} (i.e.\ deterministic outcome assignment) $g: \XX \to \Z_2$, there exists a triple $(A,B,C)$ of measurements, one per party, whose joint outcome $(g(A), g(B), g(C))$ is impossible, i.e.\ has probability zero.

Thus, constructing nonlocality paradoxes requires arranging local measurements on a state in a way that precisely generates destructive interference, zeroing out the probability of many joint outcomes.  Each zero-probability event of measuring $(A,B,C)$ and obtaining the joint outcome $(a,b,c)$ excludes the set of deterministic assignments $\{g: \XX \to \Z_2 \mid (g(A),g(B),g(C)) = (a,b,c)\}$. Moreover, enough of these impossible events have to be suitably well-distributed across all measurements and outcomes to exclude \emph{all} deterministic assignments.  This restricts the states and measurement sets that can form a paradox.

Brassard et al.\ \cite{Brassard_2005} showed that nonlocality paradoxes necessarily require at least three qubits. Abramsky et al.\ \cite{Abramsky2017} began the systematic study of three-qubit nonlocality paradoxes, giving the first examples beyond that of GHZ. They showed that, up to local unitary equivalence, such paradoxes must involve \emph{equatorial measurements} of the form $\cos\left(\varphi\right) X + \sin\left(\varphi\right) Y$ with $\varphi \in [0,\pi)$. So, each measurement set $M_i \subset [0, \pi)$.  \emph{Interpolant states} $\Bsimple$, where $\lambda \in \piovertwointerval$, lie on the line connecting $\ket{\mathrm{B}(0)} = \ket{\GHZ}$ and $\lim_{\lambda \to \frac{\pi}{2}}\ket{\mathrm{B}(\lambda)} = \ket{\mathrm{Bell}} \otimes \ket{+}$ (Definition~\ref{def:interpolant}).

For scenarios involving interpolant states, the joint outcome $(a,b,c)$ for a triple $(A,B,C)$ of measurement angles in $[0,\pi)$ is impossible precisely when
\begin{equation*}
    A + B \equiv \beta(\lambda, C + c\pi) + (1 \oplus a \oplus b)\pi \mod{2\pi},
\end{equation*}
where $\beta$ is a complicated composition of trigonometric functions (Section \ref{sec:2.3}).  Taking $M_1 = M_2 = M_3 = \{0, \frac{\pi}{2}\}$ and $\lambda = 0$, we find that $\beta(0,C + c\pi) = -(C+c\pi)$ and recover the standard GHZ paradox by applying the above equation to those triples $(A,B,C)$ containing an even number of $Y$s.

Abramsky et al.\ \cite{Abramsky2017} gave an infinite family of paradoxes, indexed by $m \in \N$, with a striking novel conditional structure. Using the interpolant state $\Bsimple<\piovertwo - \frac{\pi}{N}>$, where $N = 2m$, both Alice's and Bob's $N$ measurement angles form 
the cyclic group $\frac{\pi}{N}\Z_{N}$, while Charlie's two measurements are $X,Y$ as in the standard GHZ example.  The GHZ example corresponds to $m=1$; for higher $m$, the resulting parity constraints on satisfying assignments are conditional on Charlie's outcome when he measures $Y$:
\begin{equation*}
    \Psi_0(c_0) = \begin{cases}
        a_0 \oplus b_0 = c_0, \\
        a_j \oplus b_{N-j} = 1 \oplus c_0 \ {\scriptstyle (1 \leq j \leq N-1)};
    \end{cases}
    \,
    \begin{aligned}
        \Psi_1(c_1 = 0) & =\begin{cases}
                               a_j \oplus b_{N-1-j} = 1 \ {\scriptstyle (0 \leq j \leq N-1)};
                           \end{cases} \\
        \Psi_1(c_1 = 1) & =\begin{cases}
                               a_j \oplus b_{1-j} = 0 \ {\scriptstyle (0 \leq j \leq 1)}, \\
                               a_j \oplus b_{N-j} = 1 \ {\scriptstyle (2 \leq j \leq N-1)}.
                           \end{cases}
    \end{aligned}
\end{equation*}
Here, $a_0,\ldots,a_{N-1}, b_0,\ldots,b_{N-1}, c_0,c_1 \in \Z_2$ are Boolean variables representing the measurement outcomes given by a hidden variable. For all $(c_0, c_1) \in \Z_2^2$, summing the $2N$ equations in $\Psi_0(c_0) \cup \Psi_1(c_1)$ yields $0=1$, showing paradoxicality.

This work was extended to a classification program by the present authors~\cite{de_Silva_2025}, who introduced several key new technical tools. We conjectured that all paradoxes necessarily require an interpolant state.
We exhaustively classified all paradoxes satisfying the following conditions:
\begin{enumerate}[label=(\arabic*)]
    \item an \emph{interpolant state} is used,
    \item Charlie has \emph{two} measurements,
    \item Alice and Bob have the \emph{same} number of measurements,
    \item the logical constraints are as strong as physics allows,
    \item the contradiction arises by summing \emph{all} parity constraints.
\end{enumerate}
Our earlier classification subsumed and greatly extended the family of Abramsky et al., introducing several new infinite families that exhibited numerous novel structural features.  Removing assumptions (1) through (5) fundamentally alters the possible logico-combinatorial structure of paradoxes and demands significantly stronger methods.

\subsection{Summary of main results}
The main result in this work is a complete structural classification of three-qubit nonlocality paradoxes admitting a \emph{biconditional parity proof} (Definition~\ref{def:bpp}).  In such a proof, after relabelling the parties if necessary, Charlie has exactly two measurement settings, and each Charlie conditioning---a choice of Charlie measurement and outcome---turns the Alice--Bob possibilistic constraints into parity relations: an Alice--Bob outcome pair is possible if and only if it satisfies the corresponding conditioned $\mathbb Z_2$-parity equation.  This is a vast family of three-qubit nonlocality paradoxes that includes all prior known families, including the standard GHZ example, the family of Abramsky et al.\ \cite{Abramsky2017}, and the earlier families of the present authors~\cite{de_Silva_2025}, as well as \emph{infinitely many more} paradoxes.

We show that the paradoxes admitting a biconditional parity proof are precisely those that use interpolant states when Charlie is restricted to two measurement settings (Proposition~\ref{prop:bpp-interpolant-two-charlie}); that is, our present classification \emph{drops the above assumptions (3), (4), and (5)} from the earlier classification, and \emph{is complete under assumptions (1) and (2)}.

\subsubsection{A new logical and graph-theoretic framework}
To achieve this classification, and to lay the groundwork for further study of nonlocality paradoxes, we introduce a graph-theoretic framework for studying nonlocality paradoxes in terms of $2$-CNF formulae and implication graphs (Definition~\ref{def:imp-graph}). Roughly speaking, fixing all of Charlie's outcomes $\zz \in \Z_2^{|M_3|}$ induces, for each Alice--Bob measurement pair, a set of impossible joint outcomes. Each impossible event yields a logical clause of the form
\[  \neg(A, a)\ \lor\ \neg(B, b) \]
that any deterministic assignment $g : \XX \to \Z_2$ must satisfy, with $(A, a) \longleftrightarrow g(A) = a$ and similarly for $(B, b)$. Collecting all such clauses yields a $2$-CNF formula $\Omega (\zz)$. A nonlocality paradox is then equivalent to the statement that $\Omega(\zz)$ is \emph{unsatisfiable} for every Charlie outcome assignment $\zz$.

The resulting logical data admit a natural graph-theoretic encoding. From $\Omega(\zz)$, we construct a \emph{directed implication graph} with vertices being literals of the form $(A, a)$ or $(B, b)$; directed edges represent implications. Each clause $\neg(A, a)\ \lor\ \neg(B, b)$ gives rise to two edges:
\[ (A, a) \Longrightarrow (B, b \oplus 1) \quad \text{and} \quad (B, b) \Longrightarrow (A, a \oplus 1). \] Unsatisfiability is witnessed by a path from a literal to its complement:
\[ (A_1, 0) \Longrightarrow \cdots \Longrightarrow(A_1, 1). \]
See Figure~\ref{fig:two-charlie-2N-clock-witness} for an illustrative example of such a \emph{witness path} of a contradiction. This translation turns the problem of classifying quantum paradoxes into a combination of graph-theoretic, arithmetic, and functional-analytic questions.

For biconditional parity proofs, the implication graphs acquire additional structure: they are bidirected, and their edges are equivalent to parity equations
\[ g(A)\oplus g(B)=a\oplus b\oplus1. \]
Consequently, paradoxicality is witnessed by contradictory parity cycles. This observation allows us to extract a canonical combinatorial structure from every minimal paradox and is the basis of our classification.


\subsubsection{All interpolant-state paradoxes with two Charlie measurements}
We show that every paradox with a biconditional parity proof is reducible to a canonical combinatorial description and construct the bijection of Theorem~\ref{thm:bpp-classification}:
\[
    \BPP \quad
    \cong \quad \mathsf{CanonicalTriples}.
\]
Here, $\BPP$ is the set of specifications of quantum state and minimal measurement sets for paradoxes admitting biconditional parity proofs, up to permuting parties and local changes of basis (Definition~\ref{def:bpp}).  The set $\mathsf{CanonicalTriples}$ consists of the canonical classification data, i.e.\ the lexicographically normalised triples that contain exactly one representative of each physical equivalence class (Definition~\ref{def:classification-datum}).  Each such triple consists of combinatorial and algebraic data---primarily integers---with membership decided by easily checkable conditions.

In particular, the data that classify a biconditional parity proof start with a \emph{Charlie clock} (Definition~\ref{def:charlie-clock}), consisting of a clock group $\Z_N$, a choice of three indices from $\Z_{2N}$, and a real parameter. We show that these clocks arise because witness paths force the relevant phase contributions ($\beta(\lambda, C)$) from Charlie's measurements to lie on a finite cyclic lattice so that their pairwise differences are all rational multiples of $\pi$ (Lemma~\ref{lem:rational-ticks}). A Charlie clock is paradoxical, i.e.\ can potentially admit witness paths with the right Alice/Bob measurements, if and only if a simple parity condition holds (Lemma~\ref{lem:lifting}). We classify all clocks that are simultaneously quantum-realisable and paradoxical, obtaining \emph{four infinite families} of valid clocks (Theorem~\ref{thm:valid-clocks}). For every clock in these families, we have closed-form expressions for the interpolant parameter $\lambda$ and Charlie measurements $C_0, C_1$ that realise the data.

Once the clock is fixed, Alice's and Bob's measurements are specified by an \emph{Alice--Bob completion} (Definition~\ref{def:ab-completion}): a finite choice of cosets of certain subgroups of the clock group, with the subgroups purely determined by the Charlie clock. These cosets are precisely the carriers of the contradictory witness paths (Proposition~\ref{prop:witness-cosets}). We identify an exact number-theoretic \emph{shadow test} that determines when a partial collection of carriers already forces all remaining witnesses (Proposition~\ref{prop:shadow-exactness}). This gives a canonical notion of completion, and we prove that the canonical completions are exactly those producing minimal measurement supports (Definition~\ref{def:canonical-completion} and Theorem~\ref{thm:canonicity-minimality}).
Finally, Alice's and Bob's measurements can be partitioned into layers, which are noninteracting copies of the clock group, each offset by a real-valued shift; the free choice of these shifts completes the classifying data (Definition~\ref{def:shift-tuple}).

The classification is constructive as well as structural: it produces new paradoxes that lie outside the assumptions of all previously known constructions.
To illustrate the structure of biconditional parity proofs, we use here Example~\ref{ex:second}, a genuinely new nonlocality paradox. The valid Charlie clock results in Charlie's four phase contributions being multiples of $\pi/12$, as shown in Figure~\ref{fig:charlie-ticks-discrete}. The clock belongs to the family \emph{non-GHZ with X}, with $C_0 = 0$ and $\lambda, C_1$ given by
\begin{equation*}
    \lambda = \arcsin \left(\frac{\cos(\sigma/2)}{\cos(\tau - \sigma/2)}\right),\qquad C_1 = \arg \left(\frac{1 - \tan(\lambda/2)e^{-i\tau}}{e^{-i\tau} - \tan(\lambda/2)}\right),
\end{equation*}
where $\tau = -\frac{21}{12}\pi$ and $\sigma = \frac{5}{12}\pi$.
\begin{figure}
    \centering
    \begin{tikzpicture}[scale=0.5,line cap=round,line join=round,>=latex]

        \def\R{3.35}
        \def\Rlab{3.85}
        \def\Rtlabel{4.1}
        \def\RarcT{4.5}
        \def\RarcS{2.02}

        \draw[thick] (0,0) circle (\R);
        \fill (0,0) circle (1.25pt);

        \draw[gray!55,dashed] (-\R,0)--(\R,0);

        \foreach \m in {0,...,23}{
                \pgfmathsetmacro{\ang}{15*\m}
                \fill ({\R*cos(\ang)},{\R*sin(\ang)}) circle (1.25pt);
            }

        \foreach \j in {0,...,11}{
                \pgfmathsetmacro{\ang}{15*\j}
                \node[font=\scriptsize,blue!80!black]
                at ({\Rlab*cos(\ang)},{\Rlab*sin(\ang)}) {$\j$};
            }

        \foreach \j in {0,...,11}{
                \pgfmathsetmacro{\ang}{180+15*\j}
                \node[font=\scriptsize,orange!85!black]
                at ({\Rlab*cos(\ang)},{\Rlab*sin(\ang)}) {$\j$};
            }

        \coordinate (t00) at (  0:\R);
        \coordinate (t01) at (180:\R);
        \coordinate (t10) at (-45:\R);
        \coordinate (t11) at ( 30:\R);

        \foreach \P in {t00,t01}{
                \draw[ngreen,very thick] (\P) circle (4.0pt);
                \draw[ngreen,very thick] (\P) circle (7.0pt);
                \fill[ngreen] (\P) circle (1.5pt);
            }

        \foreach \P in {t10,t11}{
                \draw[nred,very thick] (\P) circle (4.0pt);
                \draw[nred,very thick] (\P) circle (7.0pt);
                \fill[nred] (\P) circle (1.5pt);
            }


        \draw[->,very thick,ngreen]
        ({\RarcS*cos(8)},{\RarcS*sin(8)})
        arc[start angle=8,end angle=172,radius=\RarcS];

        \node[font=\small,ngreen]
        at ({1.08*cos(90)},{1.08*sin(90)}) {$s_0=12$};

        \draw[->,very thick,nred]
        ({\RarcS*1.3*cos(-37)},{\RarcS*1.3*sin(-37)})
        arc[start angle=-37,end angle=28,radius=\RarcS*1.3];

        \node[font=\small,nred]
        at ({1.15*cos(-20)},{1.15*sin(-20)}) {$s_1=5$};

        \draw[->,very thick,blue!80!black]
        ({\RarcT*cos(-5)},{\RarcT*sin(-5)})
        arc[start angle=-5,end angle=-40,radius=\RarcT];

        \node[font=\small,blue!80!black,anchor=south west]
        at ({4.7*cos(-25)},{4.7*sin(-25)}) {$t=-3$};


        \node[font=\small,ngreen,anchor=west]
        at ({\Rtlabel*cos(0)},{\Rtlabel*sin(0)}) {$t_{0,0}=0$};

        \node[font=\small,ngreen,anchor=east]
        at ({\Rtlabel*cos(180)},{\Rtlabel*sin(180)}) {$t_{0,1}=12$};

        \node[font=\small,nred,anchor=north west]
        at ({\Rtlabel*cos(-45)},{\Rtlabel*sin(-45)}) {$t_{1,0}=21$};

        \node[font=\small,nred,anchor=south west]
        at ({\Rtlabel*cos(30)},{\Rtlabel*sin(30)}) {$t_{1,1}=2$};
    \end{tikzpicture}
    \caption{The structure of Charlie's four \emph{tick indices} $t_{l,z} \in \Z_{24}$ (see Definition~\ref{def:charlie-clock}), which describe $\beta$-contributions of the form $\frac{\pi}{12}t_{l,z}$ for a novel $\BPP$ paradox.}
    \label{fig:charlie-ticks-discrete}
\end{figure}
Next, we have one example of a canonical Alice--Bob completion for this Charlie clock, in which Alice and Bob both have two layers of measurements, with arbitrary distinct shifts $\alpha(0) = 0$ and $\alpha(1) = 2/3$. Their measurement indices, illustrated in Figure~\ref{fig:alice-bob-completion}, belong to the cosets $1 + 3\Z_{12}$ and $2 + 3 \Z_{12}$ respectively for the first layer, and to the coset $2\Z_{12}$ for the second layer. The second witness path in Figure~\ref{fig:two-charlie-implication-diagrams}, or equivalently the path in Figure~\ref{fig:two-charlie-2N-clock-witness}, exists for the Charlie outcome assignments $\zz = (0,0)$ and lies in the first of the two measurement layers for Alice and Bob. As a system of linear equations, the equivalent parity constraints are
\begin{equation*}
    \Psi_0(0) = \begin{cases}
        a_1 \oplus b_{11} = 1, \\
        a_{10} \oplus b_2 = 1, \\
        a_7 \oplus b_5 = 1,    \\
        a_4 \oplus b_8 = 1;
    \end{cases}
    \qquad
    \Psi_1(0) = \begin{cases}
        a_{10} \oplus b_{11} = 0, \\
        a_7 \oplus b_2 = 1,       \\
        a_4 \oplus b_5 = 1,       \\
        a_1 \oplus b_8 = 1.
    \end{cases}
\end{equation*}
Summing these equations yields $0 = 1$. A witness path also exists in the first measurement layer for $\zz = (1,0)$, and witness paths exist in the second layer for $\zz = (0,1),\, (1,1)$. Combined, these prove the paradox. More examples of novel paradoxes admitting biconditional parity proofs can be found in Section~\ref{sec:examples}.
\begin{figure}
    \centering
    \begin{tikzpicture}[scale=0.33,line cap=round,line join=round]
        \def\R{3.35}
        \def\Rlab{3.9}
        \def\Rtlabel{4.1}
        \def\alphaone{2/3}

        \node at (0, 1.7*\R) {\textbf{Alice}};

        \draw[thick] (\R,0)
        arc[start angle=0,end angle=180,radius=\R];

        \draw[gray!55,dashed] (-\R,0)--(\R,0);

        \foreach \m in {0,...,3}{
                \pgfmathsetmacro{\ang}{15*(1 + 3*\m)}
                \pgfmathsetmacro{\j}{1 + 3*\m}
                \fill[blue!80!black] ({\R*cos(\ang)},{\R*sin(\ang)}) circle (3pt);
                \pgfkeys{/pgf/number format/.cd,fixed,precision=2}
                \node[font=\scriptsize,blue!80!black]
                at ({\Rlab*cos(\ang)},{\Rlab*sin(\ang)}) {\pgfmathprintnumber{\j}};
            }

        \node[font=\scriptsize,pink!80!black,anchor=west]
        at (\Rtlabel, 0) {$\alpha(0) = 0$};

        \node[font=\LARGE] at (0, -2) {$\sqcup$};

        \draw[thick] (\R,-4-\R)
        arc[start angle=0,end angle=180,radius=\R];

        \draw[gray!55,dashed] (-\R,-4-\R)--(\R,-4-\R);

        \foreach \m in {0,...,5}{
                \pgfmathsetmacro{\ang}{15*(2*\m) + \alphaone*15}
                \pgfmathsetmacro{\j}{2*\m}
                \fill[blue!80!black] ({\R*cos(\ang)},{-4 - \R + \R*sin(\ang)}) circle (3pt);
                \pgfkeys{/pgf/number format/.cd,fixed,precision=2}
                \node[font=\scriptsize,blue!80!black]
                at ({\Rlab*cos(\ang)},{-4 - \R + \Rlab*sin(\ang)}) {\pgfmathprintnumber{\j}};
            }

        \node[font=\scriptsize,pink!80!black,anchor=west]
        at ({\Rtlabel*cos(\alphaone*15)}, {-4-\R + \Rtlabel*sin(\alphaone*15)}) {$\alpha(1) = \frac{2}{3}$};

        \draw[-{Latex[length=3pt]},thick,pink!80!black] ({\Rtlabel+1.2}, {-4-\R})
        arc[start angle=1.5,end angle={\alphaone*15-1.5},radius=\Rtlabel];

        \def\Bshift{19}

        \node at (\Bshift, 1.7*\R) {\textbf{Bob}};

        \draw[thick] ({\R + \Bshift},0)
        arc[start angle=0,end angle=180,radius=\R];

        \draw[gray!55,dashed] ({-\R + \Bshift},0)--({\R + \Bshift},0);

        \foreach \m in {0,...,3}{
                \pgfmathsetmacro{\ang}{15*(2 + 3*\m)}
                \pgfmathsetmacro{\j}{2 + 3*\m}
                \fill[orange!80!black] ({\Bshift + \R*cos(\ang)},{\R*sin(\ang)}) circle (3pt);
                \pgfkeys{/pgf/number format/.cd,fixed,precision=2}
                \node[font=\scriptsize,orange!80!black]
                at ({\Bshift + \Rlab*cos(\ang)},{\Rlab*sin(\ang)}) {\pgfmathprintnumber{\j}};
            }

        \node[font=\scriptsize,pink!80!black,anchor=west]
        at (\Bshift + \Rtlabel, 0) {$-\alpha(0) = 0$};

        \node[font=\LARGE] at (\Bshift, -2) {$\sqcup$};

        \draw[thick] ({\R + \Bshift},-4-\R)
        arc[start angle=0,end angle=180,radius=\R];

        \draw[gray!55,dashed] ({-\R + \Bshift},-4-\R)--({\R + \Bshift},-4-\R);

        \fill[orange!80!black] ({\Bshift + \R*cos(180-\alphaone*15)},{-4 - \R + \R*sin(180-\alphaone*15)}) circle (3pt);
        \node[font=\scriptsize,orange!80!black]
        at ({\Bshift + \Rlab*cos(180-\alphaone*15)},{-4 - \R + \Rlab*sin(180-\alphaone*15)}) {0};

        \foreach \m in {1,...,5}{
                \pgfmathsetmacro{\ang}{15*(2*\m) - \alphaone*15}
                \pgfmathsetmacro{\j}{2*\m}
                \fill[orange!80!black] ({\Bshift + \R*cos(\ang)},{-4 - \R + \R*sin(\ang)}) circle (3pt);
                \pgfkeys{/pgf/number format/.cd,fixed,precision=2}
                \node[font=\scriptsize,orange!80!black]
                at ({\Bshift + \Rlab*cos(\ang)},{-4 - \R + \Rlab*sin(\ang)}) {\pgfmathprintnumber{\j}};
            }

        \node[font=\scriptsize,pink!80!black,anchor=east]
        at ({\Bshift + \Rtlabel*cos(180-\alphaone*15)}, {-4-\R + \Rtlabel*sin(180-\alphaone*15)}) {$-\alpha(1) = -\frac{2}{3}$};

        \draw[-{Latex[length=3pt]},thick,pink!80!black] ({\Bshift-\Rtlabel-2}, {-4-\R})
        arc[start angle=178.5,end angle={180-\alphaone*15+1.5},radius=\Rtlabel];
    \end{tikzpicture}
    \caption{An illustration of a canonical Alice--Bob completion consisting of two measurement layers. The cyclic, coset structure of Alice's and Bob's indices is clearly shown.}
    \label{fig:alice-bob-completion}
\end{figure}

\subsubsection{Beyond two Charlie measurements and interpolant states}

We also present families of new exotic interpolant paradoxes beyond the restriction of Charlie to two measurements, demonstrating that the restriction on Charlie is not fundamental. In Section~\ref{sec:6}, we construct new \emph{interpolant-state} paradoxes with three Charlie measurements (Theorems~\ref{thm:NN3} and~\ref{thm:4P2p3}) and explicit examples with four Charlie measurements (Example~\ref{ex:four-charlie}). Thus, the abundance of interpolant paradoxes is not confined to the setting covered by our classification theorem.

Finally, we give a qualitatively different example that lies outside the interpolant-state framework altogether. In Section~\ref{sec:counterexample}, we construct a three-qubit nonlocality paradox whose state is \emph{not} an interpolant state, thereby providing a counterexample to our earlier conjecture~\cite{de_Silva_2025} that every three-qubit nonlocality paradox arises from an interpolant state. This shows that the structural picture suggested by earlier examples is fundamentally incomplete.


Together, these results provide several striking new insights into the nature of quantum nonlocality paradoxes. Rather than isolated constructions arising from a few special examples, our systematic analysis reveals large, highly structured families. Although paradoxicality imposes strong arithmetic and symmetry constraints, substantial combinatorial freedom remains, allowing increasingly exotic constructions. In particular, all restrictions suggested by previous examples are shown to be unnecessary. Our results reveal a rich landscape of quantum nonlocality paradoxes that can be further classified, explored, and extended.


Equivalently, our classification can be read as a fine-grained study of how quantum interference gives rise to provably nonclassical behaviour.  Nonlocality paradoxes require realising many impossible (i.e.\ probability zero) events arising from exact amplitude cancellations that are sufficiently well-distributed across combinations of measurements to rule out every hidden variable; our results identify the arithmetic and combinatorial patterns by which such local cancellations assemble into a global logical contradiction.  An exciting future direction is to consider how this wealth of quantum paradoxes can be leveraged towards demonstrating unconditional quantum advantages in a variety of information processing tasks.

\paragraph*{Outline.}
Section~\ref{sec:2} reviews the measurement-scenario formalism for three-qubit nonlocality paradoxes and the amplitude equations governing impossible events.  Section~\ref{sec:3} recasts paradoxicality in terms of \(2\)-CNF formulae and implication graphs.  Section~\ref{sec:4} proves the classification of biconditional parity proofs via valid Charlie clocks, canonical Alice--Bob completions, and layer shifts.  Section~\ref{sec:6} presents new interpolant-state paradoxes beyond the restriction of Charlie to two measurements, including examples with three and four Charlie measurements.  Finally, Section~\ref{sec:counterexample} constructs a non-interpolant-state paradox lying outside the biconditional parity proof classification.

\section{Background}
\label{sec:2}
\subsection{Nonlocality}
\label{sec:2.1}

We briefly review the notion of \emph{nonlocality}, a special case of \emph{contextuality}, within the framework of Abramsky--Brandenburger \cite{Abramsky2011}. Contextuality is a property of probabilistic data relative to a given measurement scenario. In this work, we focus on data arising from quantum states and measurements, i.e.\ quantum scenarios.

Experimental setups are modelled by the abstract framework of \emph{measurement scenarios}: a triple $(\mathcal{X}, \mathcal{O}, \mathcal{C})$, consisting of a set of measurement labels $\mathcal{X}$, a set of possible outcomes $\mathcal{O}$, and a cover $\mathcal{C}$ of $\mathcal{X}$, consisting of measurement contexts, which are maximal sets of measurements that can be jointly performed. Since we are only concerned with qubit quantum scenarios, we fix $\mathcal{O} = \Z_2$. A \emph{quantum scenario} yields an \emph{empirical model} $\mathcal{E}$, consisting of probability distributions $P_C : \mathcal{O}^C \to [0,1]$ indexed by $C \in \mathcal{C}$, defined by the Born rule, which satisfy no-signalling: for any $C, C' \in \mathcal{C}$, the marginal distributions $P_C\big|_{C\cap C'}$ and $P_{C'}\big|_{C\cap C'}$ agree.

An empirical model is \emph{noncontextual} if it admits a global distribution $\xi: \mathcal{O}^\mathcal{X} \to [0, 1]$, also known as a local hidden variable model, over \emph{global assignments} that reproduces the observed marginal distributions, i.e.\ $\xi|_C = P_C$ for all $C \in \mathcal{C}$. Such models necessarily satisfy all Bell inequalities \cite{Abramsky_2012, Peres1999}, thus contextuality is witnessed by violation of a probabilistic inequality. A global assignment $g: \mathcal{X} \rightarrow \mathcal{O}$ is \emph{consistent} with an empirical model if, for each $C\in \mathcal{C}$, the joint outcome predicted by $g$ is possible, i.e.\ $P_C(g|_C) \neq 0$. An empirical model is \emph{strongly contextual} if it admits no such consistent global assignment.
Standard examples include the GHZ paradox~\cite{Greenberger_1989} and the PR box \cite{PR1994}.

Nonlocality is a special type of contextuality in which the measurements have a multipartite structure. In the quantum case, a \emph{Bell measurement scenario} $\M$ on $n$ qubits is an $n$-tuple of finite sets $(M_1, \ldots, M_n)$, where each $M_i$ is a set of measurements on the $i$-th qubit. The \emph{contexts} of $\M$ are $\mathcal{C} = \prod_{i=1}^n M_i$, i.e.\ choices of exactly one measurement per qubit. The set of all measurements is $\mathcal{X} = \bigsqcup_{i=1}^n M_i$.

A local measurement is given by
\[ E_{\theta, \varphi} := \sin\left(\theta\right)[\cos\left(\varphi\right) X + \sin\left(\varphi\right) Y] + \cos\left(\theta\right) Z , \]
where $\theta \in \piovertwointerval$ and $\varphi \in [0, 2\pi)$, and $X,\ Y,\ Z$ denote the Pauli operators. The $+1$ eigenstate of $E_{\theta, \varphi}$ is \[ \ket{\theta, \varphi} := \cos\tfrac{\theta}{2}\ket{0} + e^{i\varphi}\sin\tfrac{\theta}{2} \ket{1}\]
and $-1$ eigenstate is $\ket{\pi - \theta, \varphi + \pi}$.
Thus, a context may be specified by a tuple $(\bm{\theta}, \bm{\varphi}) := ((\theta_1, \varphi_1),\ldots,(\theta_n, \varphi_n))$, and an \emph{event} by $(\bm{\theta}, \bm{\varphi}) \rightarrow \mathbf{o}$, where $\mathbf{o} \in \mathcal{O}^n$ labels measurement outcomes (with $+1, -1$ relabelled as $0, 1$, respectively). We reformulate the definition of strong nonlocality in this setting.
\begin{definition}
    An empirical model $\mathcal{E}(\ket{\psi}, \M)$ associated with a quantum scenario $(\ket{\psi}, \M)$, where $\ket{\psi}$ is an $n$-qubit state, is \emph{strongly nonlocal} if, given any global assignment $g: \bigsqcup_{i=1}^n M_i \rightarrow \Z_2$, there exists a context $(\bm{\theta}, \bm{\varphi})$ such that the event $(\bm{\theta}, \bm{\varphi}) \to (g(\theta_1, \varphi_1), \ldots, g(\theta_n, \varphi_n))$ is impossible.
\end{definition}
Let $\ket{(\bm{\theta}, \bm{\varphi}) \rightarrow \mathbf{o}}$ denote the eigenstate associated with the event $(\bm{\theta}, \bm{\varphi}) \rightarrow \mathbf{o}$. The event is impossible if and only if $\braket{(\bm{\theta}, \bm{\varphi}) \rightarrow \mathbf{o} | \psi} = 0$. We refer to a pair $(\ket{\psi}, \M)$ as a \emph{(quantum nonlocality) paradox} when $\mathcal{E}(\ket{\psi}, \M)$ is strongly nonlocal.

\subsection{Nonlocality paradoxes of three qubits}
\label{sec:2.2}

Two-qubit quantum states do not exhibit strong nonlocality \cite{Brassard_2005}, and hence strong nonlocality requires at least three qubits. Abramsky et al.\ \cite{Abramsky2017} showed the following:
\begin{theorem}[{\cite[Theorem 6]{Abramsky2017}}]
    A tripartite quantum state admitting strong nonlocality must be in the SLOCC class of the GHZ state and, in particular, must be \emph{balanced}. Moreover, any such strongly nonlocal behaviour can be witnessed using only \emph{equatorial measurements}.
\end{theorem}
A \emph{balanced} state is of the form $\Bstate = \mathcal{N}(\ket{v_{\bm{\lambda}}} + e^{i \Phi} \ket{w_{\bm{\lambda}}})$,
where $\bm{\lambda} = (\lambda_1, \lambda_2, \lambda_3)$ with $\lambda_j \in \piovertwointerval$, $\Phi \in [0,2\pi )$, and $\ket{v_{\bm{\lambda}}} = \bigotimes_{j=1}^3 \ket{v_{\lambda_j}},\ \ket{w_{\bm{\lambda}}} = \bigotimes_{j=1}^3 \ket{w_{\lambda_j}}$ with \begin{equation*}
    \ket{v_{\lambda_j}} := \cos{\tfrac{\lambda_j}{2}} \ket{0} + \sin{\tfrac{\lambda_j}{2}} \ket{1}, \quad \ket{w_{\lambda_j}} := \sin{\tfrac{\lambda_j}{2}} \ket{0} + \cos{\tfrac{\lambda_j}{2}} \ket{1},
\end{equation*}
and $\mathcal{N}$ is a normalising constant.

\begin{figure}
    \centering
    \resizebox{0.35\textwidth}{!}{%
        \begin{tikzpicture}[scale=0.95,line cap=round,line join=round,>=latex]
            \def\R{2.65}
            \def\s{0.36}
            \def\lam{42}

            \draw[thick] (0,0) circle (\R);

            \draw[thick,dashed] (\R,0) arc[start angle=0,end angle=180,x radius=\R,y radius={\s*\R}];
            \draw[thick] (-\R,0) arc[start angle=180,end angle=360,x radius=\R,y radius={\s*\R}];

            \draw[thick] (0,-\R) -- (0,\R);

            \coordinate (O) at (0,0);
            \coordinate (N) at (0,\R);
            \coordinate (S) at (0,-\R);
            \coordinate (M) at (\R,0);
            \coordinate (V) at ({\R*sin(\lam)},{\R*cos(\lam)});
            \coordinate (W) at ({\R*sin(\lam)},{-\R*cos(\lam)});
            \coordinate (E) at (-1.82,0.70);

            \fill (O) circle (1.5pt);
            \draw[very thick,red!75!black] (O) -- (M);
            \draw[very thick,red!75!black,->] (O) -- (V);
            \draw[very thick,red!75!black,->] (O) -- (W);
            \draw[very thick,green!75!black,->] (O) -- (E);

            \draw[thick,red!75!black,->]
            ({1.34*cos(90)},{1.34*sin(90)})
            arc[start angle=90,end angle={90-\lam},radius=1.34];
            \node[text=red!75!black] at (0.78,1.46) {$\lambda$};

            \draw[thick,green!75!black,->]
            ({0.78*cos(0)},{0.50*sin(0)})
            arc[start angle=0,end angle=150,x radius=0.78,y radius=0.50];
            \node[text=green!75!black] at (-0.20,0.29) {$\varphi$};

            \node[anchor=south] at ($(N)+(0,0.10)$) {$\ket{0}$};
            \node[anchor=north] at ($(S)+(0,-0.10)$) {$\ket{1}$};
            \node[anchor=west,text=red!75!black] at ($(V)+(0.12,0.06)$) {$\ket{v_\lambda}$};
            \node[anchor=west,text=red!75!black] at ($(W)+(0.12,-0.03)$) {$\ket{w_\lambda}$};
            \node[anchor=east,text=green!75!black] at ($(E)+(0.08,0.20)$) {$E_\varphi$};
            \node[anchor=west,text=red!75!black] at ($(M)+(0.08,0.03)$) {$\ket{+}$};
        \end{tikzpicture}%
    }
    \caption{The one-qubit geometry of a balanced state. The red vectors are the states
        $\ket{v_\lambda}$ and $\ket{w_\lambda}$, symmetric about the equator, while $E_\varphi$
        is an equatorial measurement.}
    \label{fig:balanced-state-geometry}
\end{figure}
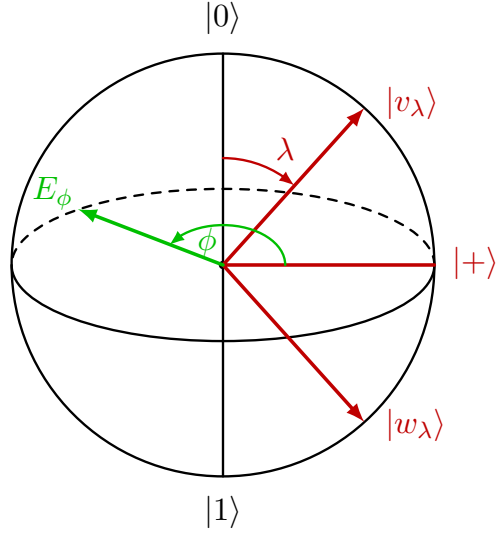
An \emph{equatorial measurement} has the form $E_\varphi := \cos\left(\varphi\right) X + \sin\left(\varphi\right) Y$, where $\varphi \in [0,2\pi)$. Since only equatorial measurements contribute to strong nonlocality, we only consider measurement scenarios $\M = (M_1, M_2, M_3)$ such that each $M_i \subset [0, 2\pi)$ is a set of angles.
We refer to $M_1,\ M_2,\ \text{and}\ M_3$ as the \emph{Alice}, \emph{Bob}, and \emph{Charlie measurements}, respectively. The following remark yields further simplifications.

\begin{remark}\label{rem:measurement}
    As shown in Abramsky et al.\ \cite[Section~2.3]{Abramsky2017}, it suffices to consider global assignments $g : \bigsqcup_{i=1}^n M_i \to \Z_2$ satisfying
    \begin{equation}\label{eqn:nice_ga}
        g(\varphi) = g(\varphi + \pi) \oplus 1
    \end{equation}
    for all equatorial measurements $E_\varphi$. Thus, to establish strong nonlocality, it is enough to rule out such assignments, and if the model is not strongly nonlocal, then an assignment exists satisfying \eqref{eqn:nice_ga}. Consequently, we may assume $M_i \subset [0,\pi)$ for all $i \in \{1,\ldots,n\}$.
\end{remark}

The following distinguished subclass of balanced states was originally introduced in \cite{de_Silva_2025}.
\begin{definition}[{\cite[Definition 2.5]{de_Silva_2025}}]\label{def:interpolant}
    An \emph{interpolant state} $\Bsimple$ is a balanced state of the form $\Bstate<(0,0,\lambda)>[0]$.
\end{definition}

All previously known examples of three-qubit nonlocality paradoxes \cite{Abramsky2017, de_Silva_2025, Brassard_2005} arise from interpolant states.
In Section~\ref{sec:counterexample}, we present a first example of a nonlocality paradox involving a non-interpolant state.

\subsubsection{Equivalences of paradox}\label{sec:equivalence}\label{sec:2.2.1}

We regard two three-qubit scenarios $(\ket{\psi},\M)$ and $(\ket{\psi'},\M')$ as \emph{equivalent} if they differ only by relabelling the qubits and by local changes of basis. More precisely, this means that there exist a local unitary $U=U_1\otimes U_2\otimes U_3$ and a permutation $\sigma$ of the three qubits such that $U\ket{\psi}=\ket{\psi'}$, and for each $i=1,2,3$, the local unitary $U_i$ carries the measurement set $M_i$ onto $M'_{\sigma(i)}$, where equatorial measurements are identified modulo $\pi$.

This equivalence relation is particularly rigid for equatorial measurement scenarios. Indeed, by \cite[Lemma~B.1]{de_Silva_2025}, if a single-qubit unitary sends a set of at least two equatorial measurements to another equatorial set, then, up to an overall phase, it must be of the form
\[ X^a P_\theta, \quad a\in\mathbb Z_2,\; \theta\in[0,2\pi); \qquad P_\theta=\operatorname{diag}(1,e^{i\theta}). \]
Thus, the allowed local changes of coordinates are rotations of the equator, possibly followed by the reflection induced by $X$:
\[ E_\varphi\longmapsto E_{\varphi+\theta} \quad \text{or} \quad E_\varphi\longmapsto E_{-\varphi-\theta}. \]
Equivalently, the only freedom on each local measurement set is to rotate the equatorial circle, and possibly reverse its orientation.

For interpolant states, this freedom is even more restricted. By \cite[Lemma~B.2]{de_Silva_2025}, if $\Bsimple$ and $\Bsimple<\lambda'>$
with $\lambda,\lambda'\neq 0$ are equivalent, then $\lambda=\lambda'$. Moreover, after fixing the interpolant form, the only remaining local unitaries are, up to phase,
\[ P_\theta\otimes P_{-\theta}\otimes I \quad \text{or} \quad XP_\theta\otimes XP_{-\theta}\otimes X. \]

Finally, if a balanced state has at least one zero component in $\bm{\lambda}$, then, after permuting the qubits, we may assume $\lambda_1=0$. In that case, the phase $\Phi$ can be removed by a phase rotation on the first qubit $P_{-\Phi}\otimes I\otimes I.$

We also require all nonlocality paradoxes to satisfy the following minimality condition: a nonlocality paradox $(\ket{\psi}, \M)$ is \emph{minimal} if $(\ket{\psi}, \M')$ is not a paradox for any $\M'$ with strictly fewer measurements than $\M$, i.e.\ $M_i' \subseteq M_i$ for all $i$, with at least one containment strict.

\subsection{Impossible events in three-qubit quantum scenarios}
\label{sec:2.3}

For $\bm{\varphi} = (\varphi_1, \varphi_2, \varphi_3)$, the event $\bm{\varphi} \to \mathbf{0}$ is impossible precisely when the amplitude $\braket{\bm{\varphi}|\BstateNoKet{}{}}$ is $0$. As shown in \cite{Abramsky2017}, this condition holds exactly when the following equation is satisfied:
\begin{equation}\label{eqn:imposs_}
    \sum_{i=1}^3 \beta(\lambda_i, \varphi_i) \equiv \pi - \Phi \mod 2\pi,
\end{equation}
where the function $\beta: \piovertwointerval \times \R/2\pi\Z \rightarrow \R/2\pi\Z$ is defined as follows:
\begin{equation}\label{eqn:beta}
    \beta(\lambda, \varphi) := \varphi - 2\arctan \left( \frac{\cos{\frac{\lambda}{2}} \sin{\varphi}}{\sin{\frac{\lambda}{2}} + \cos{\frac{\lambda}{2}} \cos{\varphi}} \right) \mod 2\pi.
\end{equation}
We note that the expression for $\beta$ differs slightly from that of \cite[Section~5.3]{Abramsky2017} and is instead due to \cite[Lemma~2.3]{de_Silva_2025}. This function provides a systematic framework for analysing impossible events in tripartite quantum scenarios.
The following two lemmas follow directly.
\begin{lemma}
    Let $(A, B, C)$ be a context in the quantum scenario $(\Bstate, \M)$ and $(a,b,c) \in \Z_2^3$. The event $(A, B, C) \to (a, b, c)$ is impossible if and only if
    \begin{equation}\label{eqn:betaABC}
        \beta(\lambda_1,A+a\pi)+\beta(\lambda_2,B+b\pi)+\beta(\lambda_3,C+c\pi) \equiv \pi - \Phi \mod 2\pi.
    \end{equation}
\end{lemma}

For convenience, in the rest of this paper,\emph{ a congruence is assumed to be modulo} ${2\pi}$ unless otherwise stated.

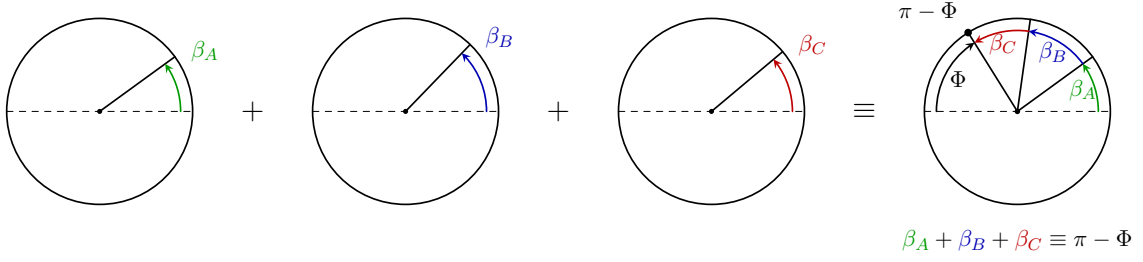
\begin{figure}
    \centering
    \resizebox{\textwidth}{!}{%
        \begin{tikzpicture}[line cap=round,line join=round,>=stealth]

            \def\R{1.55}               
            \def\g{1.00}               
            \pgfmathsetmacro{\d}{2*(\R+\g)}  

            \pgfmathsetmacro{\xA}{-1.5*\d}
            \pgfmathsetmacro{\xB}{-0.5*\d}
            \pgfmathsetmacro{\xC}{ 0.5*\d}
            \pgfmathsetmacro{\xS}{ 1.5*\d}

            \pgfmathsetmacro{\xplusAB}{(\xA+\xB)/2}
            \pgfmathsetmacro{\xplusBC}{(\xB+\xC)/2}
            \pgfmathsetmacro{\xequiv}{(\xC+\xS)/2}

            \pgfmathsetmacro{\aA}{36}
            \pgfmathsetmacro{\aB}{46}
            \pgfmathsetmacro{\aC}{40}
            \pgfmathsetmacro{\aAB}{\aA+\aB}
            \pgfmathsetmacro{\aABC}{\aA+\aB+\aC}   

            \pgfmathsetmacro{\Rarc}{\R-0.20}

            \pgfmathsetmacro{\Rlabout}{\R+0.17}     
            \pgfmathsetmacro{\Rlabin}{\R-0.40}      
            \pgfmathsetmacro{\RlabPhi}{\R-0.42}     

            \pgfmathsetmacro{\mA}{\aA/2}
            \pgfmathsetmacro{\mB}{\aB/2}
            \pgfmathsetmacro{\mC}{\aC/2}
            \pgfmathsetmacro{\mAB}{\aA+\aB/2}
            \pgfmathsetmacro{\mBC}{\aAB+\aC/2}
            \pgfmathsetmacro{\mPhi}{(180+\aABC)/2}

            \begin{scope}[shift={(\xA,0)}]
                \draw[thick] (0,0) circle (\R);
                \draw[dashed] (-\R,0)--(\R,0);
                \fill (0,0) circle (1.2pt);

                \draw[thick] (0,0)--({\R*cos(\aA)},{\R*sin(\aA)});
                \draw[->,thick,green!60!black]
                (\Rarc,0) arc[start angle=0,end angle=\aA,radius=\Rarc];

                \node[green!60!black,anchor=west]
                at ({\Rlabout*cos(\aA)},{\Rlabout*sin(\aA)})
                {$\beta_A$};
            \end{scope}

            \node at (\xplusAB,0) {\Large $+$};

            \begin{scope}[shift={(\xB,0)}]
                \draw[thick] (0,0) circle (\R);
                \draw[dashed] (-\R,0)--(\R,0);
                \fill (0,0) circle (1.2pt);

                \draw[thick] (0,0)--({\R*cos(\aB)},{\R*sin(\aB)});
                \draw[->,thick,blue!70!black]
                (\Rarc,0) arc[start angle=0,end angle=\aB,radius=\Rarc];

                \node[blue!70!black,anchor=west]
                at ({\Rlabout*cos(\aB)},{\Rlabout*sin(\aB)})
                {$\beta_B$};
            \end{scope}

            \node at (\xplusBC,0) {\Large $+$};

            \begin{scope}[shift={(\xC,0)}]
                \draw[thick] (0,0) circle (\R);
                \draw[dashed] (-\R,0)--(\R,0);
                \fill (0,0) circle (1.2pt);

                \draw[thick] (0,0)--({\R*cos(\aC)},{\R*sin(\aC)});
                \draw[->,thick,red!75!black]
                (\Rarc,0) arc[start angle=0,end angle=\aC,radius=\Rarc];

                \node[red!75!black,anchor=west]
                at ({\Rlabout*cos(\aC)},{\Rlabout*sin(\aC)})
                {$\beta_C$};
            \end{scope}

            \node at (\xequiv,0) {\Large $\equiv$};

            \begin{scope}[shift={(\xS,0)}]
                \draw[thick] (0,0) circle (\R);
                \draw[dashed] (-\R,0)--(\R,0);
                \fill (0,0) circle (1.3pt);

                \draw[thick] (0,0)--({\R*cos(\aA)},{\R*sin(\aA)});
                \draw[thick] (0,0)--({\R*cos(\aAB)},{\R*sin(\aAB)});
                \draw[thick] (0,0)--({\R*cos(\aABC)},{\R*sin(\aABC)});

                \draw[->,thick,green!60!black]
                (\Rarc,0) arc[start angle=0,end angle=\aA,radius=\Rarc];

                \draw[->,thick,blue!70!black]
                ({\Rarc*cos(\aA)},{\Rarc*sin(\aA)})
                arc[start angle=\aA,end angle=\aAB,radius=\Rarc];

                \draw[->,thick,red!75!black]
                ({\Rarc*cos(\aAB)},{\Rarc*sin(\aAB)})
                arc[start angle=\aAB,end angle=\aABC,radius=\Rarc];

                \draw[->,thick]
                ({\Rarc*cos(180)},{\Rarc*sin(180)})
                arc[start angle=180,end angle=\aABC,radius=\Rarc];

                \node[green!60!black]
                at ({\Rlabin*cos(\mA)},{\Rlabin*sin(\mA)})
                {$\beta_A$};

                \node[blue!70!black]
                at ({\Rlabin*cos(\mAB)},{\Rlabin*sin(\mAB)})
                {$\beta_B$};

                \node[red!75!black]
                at ({\Rlabin*cos(\mBC)},{\Rlabin*sin(\mBC)})
                {$\beta_C$};

                \node
                at ({\RlabPhi*cos(\mPhi)},{\RlabPhi*sin(\mPhi)})
                {$\Phi$};

                \fill ({\R*cos(\aABC)},{\R*sin(\aABC)}) circle (1.8pt);
                \node[anchor=south east]
                at ({(\R+0.08)*cos(\aABC)},{(\R+0.08)*sin(\aABC)})
                {$\pi-\Phi$};

                \node at (0,-2.15)
                {$\textcolor{green!60!black}{\beta_A}
                        +\textcolor{blue!70!black}{\beta_B}
                        +\textcolor{red!75!black}{\beta_C}
                        \equiv \pi-\Phi$};
            \end{scope}

        \end{tikzpicture}%
    }
    \caption{Schematic form of the \(\beta\)-equation. Here
        \(\beta_A:=\beta(\lambda_1,A+a\pi)\), \(\beta_B:=\beta(\lambda_2,B+b\pi)\), and
        \(\beta_C:=\beta(\lambda_3,C+c\pi)\). The event \((A,B,C)\to(a,b,c)\) is
        impossible exactly when the three angle contributions sum to the target
        direction \(\pi-\Phi\).}
    \label{fig:beta-equation}
\end{figure}

\begin{lemma}[{\cite[Lemma~2.4]{de_Silva_2025}}]\label{lem:beta-facts}
    The following properties of $\beta$ can be easily verified.
    \begin{enumerate}
        \item Modulo $2\pi$, for all fixed $\lambda \in \piovertwointerval$, $\beta(\lambda, \varphi) \in [0, 2\pi)$ is strictly decreasing as a function of $\varphi$ on $(0, 2\pi)$ and is thus bijective on $[0, 2\pi)$.
        \item $\beta(0, \varphi) \equiv -\varphi $.
        \item For all $\lambda \in \piovertwointerval$, $\beta(\lambda, \varphi) \equiv 0$ if and only if $\varphi \equiv 0 $, and $\beta(\lambda, \varphi) \equiv \pi $ if and only if $\varphi \equiv \pi $.
        \item $\beta(\lambda, \piovertwo) \equiv \lambda - \piovertwo $.
    \end{enumerate}
\end{lemma}

To better analyse \eqref{eqn:betaABC}, the auxiliary function
\begin{equation*}
    \delta(\lambda, \varphi) := \beta(\lambda, \varphi + \pi) - \beta(\lambda, \varphi) \equiv \pi - 2\arctan(\sin\varphi\tan\lambda) \mod 2\pi,
\end{equation*}
was introduced in \cite{de_Silva_2025}, which captures the change in $\beta$ under a flip of the measurement outcome. The properties of $\delta$ determine which events in a given context are impossible and were analysed in \cite{de_Silva_2025}.

\begin{lemma}\label{lem:delta-facts}
    Let $\lambda\in\piovertwointerval$ and $\varphi\in[0,\pi)$. Then $\delta(\lambda,\varphi)\in(0,\pi]$, with $\delta(\lambda,\varphi)\equiv\pi$ if and only if $\lambda=0$ or $\varphi=0$, and
    \[
        \delta(\lambda_1,\varphi_1)\pm\delta(\lambda_2,\varphi_2)\equiv 0
        \iff
        \sin\varphi_1\tan\lambda_1=\mp\sin\varphi_2\tan\lambda_2.
    \]
    In particular, $\delta(\lambda_1,\varphi_1)+\delta(\lambda_2,\varphi_2)\equiv 0$ if and only if $\delta(\lambda_1,\varphi_1)=\delta(\lambda_2,\varphi_2)\equiv\pi$.
\end{lemma}

\begin{remark}
    Since $\delta(\lambda, \varphi) \not\equiv 0$ for all $\lambda$ and $\varphi$, impossible events in a given context must differ by at least $2$. Thus, the number of impossible events in any context is at most $4$.
\end{remark}

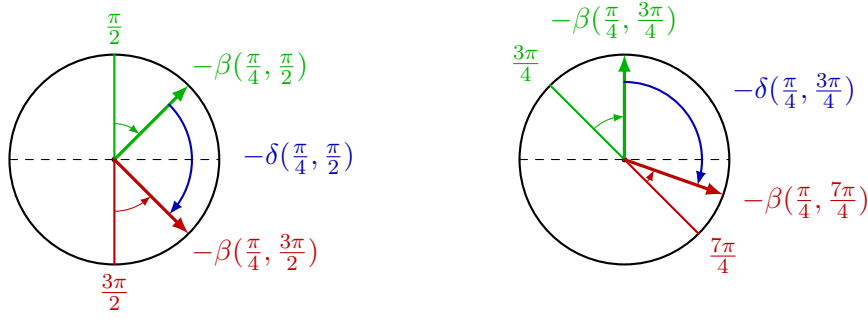
\begin{figure}
    \centering
    \begin{tikzpicture}[scale=0.95,line cap=round,line join=round,>=latex]
        \def\R{1.46}
        \def\Rdelta{1.08} 
        \pgfmathsetmacro{\nBetaSeven}{360-atan(sqrt(2)/4)}
        \pgfmathsetmacro{\midDeltaRight}{(90+(\nBetaSeven-360))/2}

        \begin{scope}[shift={(-3.55,0)}]
            \draw[thick] (0,0) circle (\R);
            \draw[dashed] (-\R,0)--(\R,0);
            \fill (0,0) circle (1.1pt);

            \draw[thick,green!70!black] (0,0)--(0,\R);
            \draw[thick,red!75!black]  (0,0)--(0,-\R);

            \node[green!70!black,anchor=south] at (0,\R) {$\frac{\pi}{2}$};
            \node[red!75!black,anchor=north] at (0,-\R) {$\frac{3\pi}{2}$};

            \draw[very thick,green!70!black,->] (0,0)--({\R*cos(45)},{\R*sin(45)});
            \draw[very thick,red!75!black,->]  (0,0)--({\R*cos(315)},{\R*sin(315)});

            \draw[->,green!70!black]
            ({0.50*cos(90)},{0.50*sin(90)})
            arc[start angle=90,end angle=45,radius=0.50];

            \draw[->,red!75!black]
            ({0.72*cos(270)},{0.72*sin(270)})
            arc[start angle=270,end angle=315,radius=0.72];

            \node[green!70!black,anchor=west]
            at ({\R*cos(45)-0.10},{\R*sin(45)+0.24})
            {$-\beta(\frac{\pi}{4},\frac{\pi}{2})$};

            \node[red!75!black,anchor=west]
            at ({\R*cos(315)-0.10},{\R*sin(315)-0.26})
            {$-\beta(\frac{\pi}{4},\frac{3\pi}{2})$};

            \draw[->,thick,blue!70!black]
            ({\Rdelta*cos(45)},{\Rdelta*sin(45)})
            arc[start angle=45,end angle=-45,radius=\Rdelta];

            \node[blue!70!black,anchor=west]
            at ({(\R+0.18)*cos(0)},{(\R+0.18)*sin(0)+0.02})
            {$-\delta(\frac{\pi}{4},\frac{\pi}{2})$};
        \end{scope}

        \begin{scope}[shift={(3.55,0)}]
            \draw[thick] (0,0) circle (\R);
            \draw[dashed] (-\R,0)--(\R,0);
            \fill (0,0) circle (1.1pt);

            \draw[thick,green!70!black] (0,0)--({\R*cos(135)},{\R*sin(135)});
            \draw[thick,red!75!black]  (0,0)--({\R*cos(315)},{\R*sin(315)});

            \node[green!70!black,anchor=south east]
            at ({\R*cos(135)+0.02},{\R*sin(135)-0.08})
            {$\frac{3\pi}{4}$};

            \node[red!75!black,anchor=north west]
            at ({\R*cos(315)-0.02},{\R*sin(315)+0.08})
            {$\frac{7\pi}{4}$};

            \draw[very thick,green!70!black,->] (0,0)--({\R*cos(90)},{\R*sin(90)});
            \draw[very thick,red!75!black,->]  (0,0)--({\R*cos(\nBetaSeven)},{\R*sin(\nBetaSeven)});

            \draw[->,green!70!black]
            ({0.60*cos(135)},{0.60*sin(135)})
            arc[start angle=135,end angle=90,radius=0.60];

            \draw[->,red!75!black]
            ({0.44*cos(315)},{0.44*sin(315)})
            arc[start angle=315,end angle=\nBetaSeven,radius=0.44];

            \node[green!70!black,anchor=west]
            at (-1.18,{\R+0.48})
            {$-\beta(\frac{\pi}{4},\frac{3\pi}{4})$};

            \node[red!75!black,anchor=west]
            at ({\R*cos(\nBetaSeven)+0.12},{\R*sin(\nBetaSeven)-0.08})
            {$-\beta(\frac{\pi}{4},\frac{7\pi}{4})$};

            \draw[->,thick,blue!70!black]
            ({\Rdelta*cos(90)},{\Rdelta*sin(90)})
            arc[start angle=90,end angle={\nBetaSeven-360},radius=\Rdelta];

            \node[blue!70!black,anchor=west]
            at ({(\R+0.20)*cos(\midDeltaRight)},{(\R+0.20)*sin(\midDeltaRight)})
            {$-\delta(\frac{\pi}{4},\frac{3\pi}{4})$};
        \end{scope}

    \end{tikzpicture}
    \caption{Two examples illustrating the angles $-\beta(\lambda,\phi)$ and the corresponding separation
    $-\delta(\lambda,\phi)$ between two outcomes' $\beta$-contributions. The green and red diameters indicate the two
    outcomes ($0$ and $1$ respectively) of the measurement at angle $\phi \in [0,\pi)$.  The left circle illustrates $\lambda = \pi/4,\, \phi = \pi/2$; the right circle illustrates $\lambda = \pi/4,\, \phi = 3\pi/4$.   We see that the effect of $-\beta$ is to ``pull'' its input angle to the right.  The higher $\lambda$ is, the stronger this pulling action is.  Moreover, as we see in the right circle, the closer an angle is to the leftmost point of a circle, the more it gets pulled by $-\beta$.}
    \label{fig:beta-delta-examples}
\end{figure}

\section{A graph-theoretic formalism for paradoxes}
\label{sec:3}
\label{sec:graph}

A graph-theoretic formalism encodes the logical constraints of nonlocality paradoxes in a precise combinatorial structure, making their properties accessible to rigorous analysis using standard tools from graph theory.
In this section and those that follow, we draw on classical results from graph theory \cite{Diestel2017} and \textsf{2-SAT} \cite{Papadimitriou1994, Biere2021}. In particular, we rely on the correspondence between $2$-CNF formulae and implication graphs \cite{Aspvall1979ALA}. For completeness and clarity, we restate these results in a unified graph-theoretic language tailored to nonlocality paradoxes.

\subsection{Logical structure of nonlocality paradoxes}
\label{sec:graph.logical}

We review the logical structure of the three-qubit interpolant-state paradoxes developed in \cite{de_Silva_2025} and extend it to general balanced-state scenarios. The GHZ state is itself an interpolant state (with $\bm\lambda=\mathbf 0$), and the corresponding paradox is characterised by a single unsatisfiable system of $\Z_2$-linear equations. By contrast, the family of paradoxes in \cite{Abramsky2017, de_Silva_2025} involves multiple such systems, conditioned on the outcome of Charlie’s measurements.

We fix the quantum scenario $(\Bstate,\M)$. By \eqref{eqn:imposs_}, an event $(A,B,C)\to(a,b,c)$ is impossible if and only if
\begin{equation} \label{eqn:imposs_abc}
    \beta(\lambda_1, A) + a\,\delta(\lambda_1,A) + \beta(\lambda_2, B) + b\,\delta(\lambda_2 , B) \equiv \pi - T_{C,c} -\Phi ,
\end{equation}
where
\begin{equation}\label{eqn:charlie-tick-def}
    T_{C,c} := \beta(\lambda_3, C) + c\, \delta(\lambda_3 , C) \in \R/2\pi\Z,
\end{equation}
which we refer to as a Charlie \emph{tick value}.

We refer to the pair $(C,z)$, consisting of a measurement $C\in M_3$ and an outcome $c=z$, as a \emph{Charlie conditioning}. We define
\begin{equation*}
    E_{C,z} := \bigl\{ ((A,a),(B,b)) \mid (A,B,C)\to(a,b,z) \text{ is impossible} \bigr\},
\end{equation*}
the set of outcome-labelled measurement pairs on Alice and Bob that are impossible with $(C,z)$. We will refer to the atomic propositions $(A,a)$, $(B,b)$ as \emph{literals}, defining $\neg(A,a)$ as $(A,a\oplus 1)$.

Now write $M_3 = \{C_0,\dots,C_{n-1}\}$ and fix a bit string $\mathbf{z} = (z_0,\ldots,z_{n-1}) \in \Z_2^n$. We define the set
\begin{equation*}
    E_{\zz} := \bigcup_{l=0}^{n-1} E_{C_l,z_l},
\end{equation*}
which collects all impossible Alice--Bob literal pairs with the chosen Charlie outcomes $\zz$. We refer to $\zz$ as a \emph{total Charlie assignment}.  Each element $((A,a),(B,b))\in E_{\zz}$ imposes the constraint that Alice and Bob cannot output $a$ on $A$ and $b$ on $B$, i.e.\ it gives rise to the clause
\begin{equation}\label{eq:clause}
    \neg((A, a)\ \land\ (B, b)) \iff (A, a\oplus 1) \vee (B, b\oplus 1).
\end{equation}
Let $\Omega_l(z_l)$ denote the conjunction of all clauses arising from $E_{C_l,z_l}$, and define
\begin{equation}\label{eq:2CNF}
    \Omega(\mathbf z) := \bigwedge_{l=0}^{n-1} \Omega_l(z_l).
\end{equation}
So, each choice of $\mathbf z$ determines a $2$-CNF formula.

\begin{lemma}\label{lem:2CNF-paradox}
    The quantum scenario $(\Bstate,\M)$, with $|M_3| =n$, is a nonlocality paradox if and only if, for every $\mathbf{z}\in \Z_2^n$, $\Omega(\mathbf{z})$ is unsatisfiable.
\end{lemma}
For interpolant-state scenarios, $((A,a),(B,b))\in E_{C,z}$ if and only if $ ((A, a\oplus 1),(B,b\oplus 1))\in E_{C,z}$.
The two clauses \eqref{eq:clause} arising from the pairs $((A,a),(B,b))$ and $ ((A,a\oplus 1),(B,b\oplus 1))$ together yield the constraint on compatible global assignments $g:\bigsqcup_{i=1}^3 M_i  \rightarrow \Z_2$
\begin{equation}\label{eqn:z2-constraint}
    g(A)\oplus g(B)=a\oplus b \oplus 1.
\end{equation}
Thus, in the interpolant case, the formula $\Omega(\zz)$ reduces to a system of $\mathbb Z_2$-linear equations, which we denote by $\Psi(\zz)$. The paradox condition is then precisely that this system be inconsistent for every total Charlie assignment. Equivalently, Lemma~\ref{lem:2CNF-paradox} says that \begin{equation}
    (\Bsimple,\M) \text{ is a paradox} \quad\Longleftrightarrow\quad
    \Psi(\zz)\text{ is inconsistent for all }\zz \in \Z_2^n.
\end{equation}

\subsection{Graph-theoretic structure of nonlocality paradoxes}
\label{sec:graph.graph}

The logical structure of three-qubit paradoxes can be encoded in a graph-theoretic framework. As before, we fix $(\Bstate,\M)$ with $M_3=\{C_0,\dots,C_{n-1}\}$. For a given Charlie conditioning $(C_l,z)$, we define $G_l(z)$ to be a simple bipartite graph with the vertex set
\begin{equation*} V({G}_l (z_l)) := \{(A,a)\mid A\in M_1,\ a\in\Z_2\} \sqcup \{(B,b)\mid B\in M_2,\ b\in\Z_2\}, \end{equation*}
and edge set $E({G}_l (z_l)) := E_{C_l, z_l}$.
\begin{definition}
    Given a total Charlie assignment $\zz$, we define the \emph{literal graph} $G (\zz)$ as the union \[ G (\zz) := \bigcup_{l=0}^{n-1} G_{C_l, z_l}, \]
    where the vertex set is $V(G_l (z_l))$ and the edge set is $\bigcup_l E_{C_l, z_l}$.
\end{definition}
Since $\beta$ is injective on $[0,2\pi)$, if an Alice--Bob literal pair $(A, a)$ and $(B, b)$ forms an edge (i.e.\ gives an impossible event) for one Charlie conditioning $(C_l, z_l)$, it cannot form an edge for any other Charlie conditioning $(C_{l'}, z_{l'})$. Therefore, the graph $G (\zz)$ is simple. Furthermore, for any given Charlie conditioning $(C_l, z_l)$, each literal $(A,a)$ or $(B,b)$ has degree at most $1$; hence, each edge set $E_{C_l, z_l}$ is a partial matching.

Each edge of $G (\zz)$ induces a clause of the form \eqref{eq:clause}, and the conjunction of all such clauses yields the $2$-CNF formula $\Omega(\zz)$ defined in \eqref{eq:2CNF}. The satisfiability of $\Omega(\zz)$ admits a natural graph-theoretic characterisation.

\begin{definition}\label{def:imp-graph}
    Given a $2$-CNF formula $\Omega(\zz)$ associated with a total Charlie assignment $\zz$, we define its \emph{implication graph} $I_\Omega (\zz)$ as follows. The vertices are the literals appearing in $\Omega(\zz)$, and for each forbidden pair $((A,a),(B,b)) \in E(G(\zz))$, we include the two directed implications \begin{equation}\label{eq:implications}
        (A,a)\Longrightarrow (B,b\oplus 1),
        \qquad
        (B,b)\Longrightarrow (A,a\oplus 1).
    \end{equation}
    Equivalently, the incompatibility of the literals $(A,a)$ and $(B,b)$ is encoded by the clauses
    \[ (A,a\oplus 1)\lor (B,b\oplus 1), \]
    or, in implication form, by \eqref{eq:implications}.
\end{definition}
The implication graph $I_\Omega (\zz)$ is a directed bipartite graph, with bipartition given by the literals over $M_1$ and those over $M_2$. The satisfiability of $\Omega(\zz)$ can be characterised in terms of the strongly connected components of this graph.
\begin{lemma}[{\cite[Theorem~1]{Aspvall1979ALA}}]\label{lem:impl-graph}
    The $2$-CNF formula $\Omega(\mathbf z)$ is unsatisfiable if and only if there exists some $X \in M_1\sqcup M_2$ such that the literal $(X,x)$ and its complement $(X,x\oplus 1)$ belong to the same \emph{strongly connected component} of $I_\Omega (\zz)$. Equivalently, there are directed paths
    \begin{equation}\label{eq:cycle-witness}
        (X,x)\Longrightarrow \cdots \Longrightarrow (X,x\oplus 1) \Longrightarrow \cdots \Longrightarrow (X,x).
    \end{equation}
\end{lemma}
For interpolant-state scenarios, the forbidden pairs arise in complementary pairs. Consequently, the implication graph has an additional reversibility property.



\begin{lemma}
    \label{lem:interpolant-impl}
    For an interpolant-state scenario, every directed edge of an implication graph $I_{\Psi}(\zz)$ occurs together with its reverse. Hence $I_{\Psi}(\zz)$ is bidirected. Consequently, $\Psi(\zz)$ is unsatisfiable if and only if $I_{\Psi}(\zz)$ contains a path from some literal to its complement.
\end{lemma}
\begin{proof}
    Suppose that $(A,B,C_l)\to (a,b,z_l)$ is impossible. By the interpolant impossibility condition, replacing $(a,b)$ by $(a\oplus 1,b\oplus 1)$ does not change the relevant parity term. Hence, the complementary event $(A,B,C_l)\to (a\oplus 1,b\oplus 1,z_l)$ is also impossible.
    The first forbidden pair gives the implication
    \[ (A,a)\Longrightarrow (B,b\oplus 1), \]
    while the complementary forbidden pair gives the reverse implication
    \[ (B,b\oplus 1)\Longrightarrow (A,a).\]
    The same argument applies to the other implication arising from the forbidden pair. Thus every directed edge in $I_{\Psi}(\zz)$ occurs with its reverse, so $I_{\Psi}(\zz)$ is bidirected.

    By Lemma~\ref{lem:impl-graph}, $\Psi(\zz)$ is unsatisfiable if and only if some literal and its complement lie in the same strongly connected component. Since $I_{\Psi}(\zz)$ is bidirected, this is equivalent to the existence of a path from a literal to its complement.
\end{proof}

\begin{remark}\label{rem:witness-paths}
    Since the implication graph is bipartite, with parts given by Alice and Bob literals, every path alternates between Alice and Bob. It therefore suffices to study paths that start and end at Alice literals. We refer to a path of the form
    \[ (X,x)\Longrightarrow \cdots \Longrightarrow (X,x\oplus 1) \]
    as a \emph{witness path}.
\end{remark}

\begin{remark}\label{rem:min-witness-path}
    In the arguments below, we consider witness paths that are \emph{minimal} with respect to deleting redundant subpaths. Edges of the implication graph are naturally labelled by the Charlie conditioning from which they arise. Along a \emph{minimal witness path}, two consecutive edges from the same conditioning traverse a matching edge and immediately return, so they are a removable backtrack. 
    Thus, consecutive edges in a minimal witness must arise from distinct Charlie conditionings.

    In particular, when $|M_3|=2$, a minimal witness path must alternate between the two Charlie conditionings. This is the structure that will be encoded in Sections \ref{sec:two-charlie} and \ref{sec:6} by the corresponding return maps.
\end{remark}

\subsection{Biconditional parity proofs}
\label{subsec:bpp}

\begin{definition}
    \label{def:bpp}
    Let $(\ket{\psi},\M)$ be a three-qubit quantum scenario.  We say that it
    admits a \emph{biconditional parity proof} if, after possibly permuting the
    parties, $M_3=\{C_0,C_1\}$ and, for each Charlie conditioning $(C_l,z)$, there
    are data
    \[
        F_{l,z}\subseteq M_1\times M_2,
        \qquad
        p_{l,z}:F_{l,z}\to\Z_2,
    \]
    such that:
    \begin{enumerate}
        \item if $(A,B)\notin F_{l,z}$, then no event
              $(A,B,C_l)\to(a,b,z)$ is impossible;
        \item if $(A,B)\in F_{l,z}$, then
              \[
                  (A,B,C_l)\to(a,b,z)\text{ is possible}
                  \quad\Longleftrightarrow\quad
                  a\oplus b=p_{l,z}(A,B);
              \]
        \item for every total Charlie assignment $\zz=(z_0,z_1)\in\Z_2^2$, the
              system
              \[
                  g(A)\oplus g(B)=p_{l,z_l}(A,B),
                  \qquad
                  l\in\Z_2,\ (A,B)\in F_{l,z_l},
              \]
              has no solution $g:M_1\sqcup M_2\to\Z_2$.
    \end{enumerate}
    We write $\BPP$ for the set of minimal paradoxes admitting such a proof,
    modulo the equivalence of Section~\ref{sec:equivalence}.
\end{definition}

\begin{proposition}
    \label{prop:bpp-interpolant-two-charlie}
    A minimal three-qubit paradox admits a biconditional parity proof if and only if it is equivalent to an interpolant-state paradox with two Charlie measurements, that is, to a scenario of the form
    \[
        (\Bsimple,\M),
        \qquad |M_3|=2.
    \]
\end{proposition}

\begin{proof}
    Suppose first that \((\ket{\psi},\M)\) admits a biconditional parity proof.  By the
    three-qubit reduction theorem recalled in Section~\ref{sec:2.2}, we may pass to an
    equivalent balanced equatorial representative.  After the party permutation allowed in
    Definition~\ref{def:bpp}, write Charlie for the conditioning party, so
    \(M_3=\{C_0,C_1\}\), and write the balanced-state parameters as
    \((\lambda_A,\lambda_B,\lambda_C)\).

    Each active biconditional constraint gives, for fixed \((C_l,z)\), a pair of complementary impossible events
    \[
        (A,B,C_l)\to(a,b,z),
        \qquad
        (A,B,C_l)\to(a\oplus1,b\oplus1,z).
    \]
    Subtracting their \(\beta\)-equations gives
    \[
        \begin{cases}
            \delta(\lambda_A,A)+\delta(\lambda_B,B)\equiv 0 \quad & \text{if forbidden outcomes are } 00,\, 11; \\
            \delta(\lambda_A,A) \equiv \delta(\lambda_B,B)        & \text{if forbidden outcomes are } 01,\, 10.
        \end{cases}
    \]
    The first equation forces \[ \delta(\lambda_A,A)\equiv\delta(\lambda_B,B)\equiv\pi.\]
    Thus, a parity-\(0\) edge implies equality of the \(\delta\)-values at its endpoints, whereas a parity-\(1\) edge forces both endpoint \(\delta\)-values to equal \(\pi\). For each total Charlie assignment, the inconsistent parity graph contains a cycle whose edge parities sum to \(1\). Such a cycle necessarily contains a parity-\(1\) edge, and hence every \(\delta\)-value associated with a vertex on the cycle is equal to \(\pi\). This forces \(\lambda_A=\lambda_B=0\). Thus the state is equivalent to one of the form
    \(\Bstate<(0,0,\lambda)>[\Phi]\).  Since a zero-parameter qubit allows the residual
    phase \(\Phi\) to be removed by a local phase rotation, the state is equivalent to
    \(\Bsimple\).  Hence the scenario is equivalent to one of the form
    \[
        (\Bsimple,\M),
        \qquad |M_3|=2.
    \]
    Conversely, suppose \((\Bsimple,\M)\) is a paradox with \(M_3=\{C_0,C_1\}\).  Put
    \[
        T_{l,z}:=\beta(\lambda,C_l+z\pi).
    \]
    Since \(\beta(0,\varphi)\equiv-\varphi\), the impossibility condition becomes
    \[
        (A,B,C_l)\to(a,b,z)\text{ is impossible}
        \quad\Longleftrightarrow\quad
        A+B\equiv T_{l,z}+(1\oplus a\oplus b)\pi .
    \]
    Thus, for fixed \(A,B,l,z\), either \(A+B\not\equiv T_{l,z}\mod\pi\), in which case there
    is no Alice--Bob constraint, or there is a unique \(\epsilon_{l,z}(A,B)\in\Z_2\) such that
    \[
        A+B\equiv T_{l,z}+\epsilon_{l,z}(A,B)\pi .
    \]
    In the latter case, the impossible outcomes are exactly those with
    \(1\oplus a\oplus b=\epsilon_{l,z}(A,B)\), so the possible outcomes are exactly the single parity class
    \[
        a\oplus b=\epsilon_{l,z}(A,B).
    \]
    Define
    \[
        F_{l,z}:=\{(A,B)\in M_1\times M_2\mid A+B\equiv T_{l,z}\mod\pi\},
        \qquad
        p_{l,z}(A,B):=1\oplus\epsilon_{l,z}(A,B).
    \]
    Then pairs outside \(F_{l,z}\) impose no constraint, while pairs inside \(F_{l,z}\) satisfy
    \[
        (A,B,C_l)\to(a,b,z)\text{ is possible}
        \quad\Longleftrightarrow\quad
        a\oplus b=p_{l,z}(A,B).
    \]
    Since \((\Bsimple,\M)\) is a paradox, for every total Charlie assignment
    \(\zz=(z_0,z_1)\) the resulting Alice--Bob parity system has no global solution.  Hence, the scenario admits a biconditional parity proof.
\end{proof}



\section{Classification of biconditional parity proofs}
\label{sec:4}
\label{sec:two-charlie}

In this section, we give a complete classification of minimal biconditional parity proofs of three-qubit strong nonlocality. Recall that, up to equivalence, these are precisely those paradoxes involving an interpolant state $\Bsimple = \Bstate<(0,0,\lambda)>[0]$ and for which Charlie is restricted to two measurements: $M_3 = \{C_0, C_1\}\subset [0,\pi)$.  We will show that, despite the seemingly continuous nature of the parameters involved, a discrete structure emerges: we will demonstrate a bijection between these paradoxes and tuples of numbers that satisfy simple, easily checkable conditions.

We reserve the scalar subscripts \(l,z\in\Z_2\) for a Charlie measurement \(C_l\) and a corresponding Charlie outcome \(z\).  The four pairs \((l,z)\) index Charlie's four tick values $T_{C_l,z}$, i.e.\ the possible $\beta$-values he can contribute, as defined in \eqref{eqn:charlie-tick-def}. For convenience, we will write $T_{l,z} := T_{C_l,z}$.  A total Charlie assignment (as in Section \ref{sec:graph.logical}) is a pair of outcomes for $C_0, C_1$, i.e.\ a vector $\zz$ in \[
    Q:=\Z_2^2.
\]

Our exposition consists of three key stages:
\begin{enumerate}
    \item First, in Section~\ref{subsec:charlie-clocks}, we show that the data associated with Charlie's qubit and measurements ($\lambda \text{ and } M_3$ or, equivalently, his set of four tick values), are equivalent to a choice of discrete clock structure: a cyclic group $\Z_N$, three nonzero elements of $\Z_{2N}$, and a real-valued offset $\mu \in [0, 2N)$.  We show that this induced clock leads to natural normal forms for Charlie's ticks as well as Alice and Bob's measurements.  Alice and Bob's measurements are best expressed as a pair of an element of the clock group $\Z_N$ and a discrete \emph{layer} label that indexes a real-valued shift.  Finally, we precisely characterise the discrete clock structures, with simple and easily checkable criteria, that are valid in the sense of being both quantum-realisable and completable by a choice of Alice and Bob measurement sets to a minimal paradox.
    \item Next, in Section~\ref{subsec:ab-completions}, we determine the choices of Alice and Bob measurement sets that complete a valid Charlie clock to yield a minimal biconditional parity proof.  We apply Lemma \ref{lem:interpolant-impl} and find that the existence of witness paths in the implication graphs of the paradox yields a choice of coset of $\Z_N$ for each total Charlie assignment and a layer index for each coset.  Together, these layer-labelled cosets yield a cover of Alice's and Bob's measurements.  We eliminate the redundancy in these data by reducing this to a choice of coset only for certain \emph{selected} total Charlie assignments.  We define such a partial choice to be \emph{canonical} if it uniquely defines a full choice of cosets and characterise canonicity with elementary number-theoretic criteria.
    \item Finally, in Section~\ref{subsec:bpp-classification}, we show that, upon fixing a valid Charlie clock and canonical Alice--Bob completion of that clock, the only remaining freedom is a choice of real-valued shift for each of the layer indices in the Alice--Bob completion.  We are then ready to assemble all the above results to give the desired classification theorem.
\end{enumerate}

The main theorem is a bijection
\[
    \BPP \quad
    \cong \quad \mathsf{CanonicalTriples}  \quad\subset \quad
    \bigsqcup_{\Gamma\in\mathsf{Clock}}\;\;
    \bigsqcup_{P\in\mathsf{Comp}(\Gamma)}\;
    \mathsf{Shift}(P).
\]
Here \(\BPP\) denotes minimal biconditional parity proofs, up to the physical equivalence of Section~\ref{sec:equivalence}; \(\mathsf{Clock}\) denotes valid Charlie clocks; \(\mathsf{Comp}(\Gamma)\) denotes canonical Alice--Bob completions over \(\Gamma\); and \(\mathsf{Shift}(P)\) records the remaining real layer shifts.  The set $\mathsf{CanonicalTriples}$ is a subset of all possible triples of a Charlie clock, an Alice--Bob completion of that clock, and a free choice of layer shifts that contains precisely one representative triple for each physical equivalence class of minimal biconditional parity proofs.

\subsection{Charlie clocks}
\label{subsec:charlie-clocks}

The purpose of this subsection is to isolate the part of the classification controlled entirely by the state and Charlie's two measurements.  We begin with abstract definitions of the Charlie clock structure and what it means for such a structure to be \emph{paradoxical} and \mbox{(quantum-)\emph{realisable}}; in the next subsection, we show the definition of paradoxical allows us to construct witness paths for a paradox.  We then show how a biconditional parity proof yields a clock structure, and express all measurements relative to this clock, with offsets.  Finally, we classify the clocks that are both paradoxical and realisable.

\begin{definition}
    \label{def:charlie-clock}
    A \emph{Charlie clock} is a tuple
    \[
        \Gamma=(N,t,s_0,s_1,\mu)
    \]
    with
    \[
        N\in\N_{>0},
        \qquad
        t,s_0,s_1\in\Z_{2N},
        \qquad
        \mu\in[0,2N)\subset\R,
    \]
    and
    \[
        \gcd(N,t,s_0,s_1)=1.
    \]
    The \emph{tick indices} of \(\Gamma\) are
    \[
        t_{0,0}:=0,
        \qquad
        t_{0,1}:=s_0,
        \qquad
        t_{1,0}:=t,
        \qquad
        t_{1,1}:=t+s_1
    \]
    in \(\Z_{2N}\).  The corresponding Charlie \emph{tick values} are
    \begin{equation}\label{eqn:charlie-tick-decomp}
        T_{l,z}:=\frac{\pi}{N}(t_{l,z}+\mu)
        \in\R/2\pi\Z,
        \qquad
        l,z\in\Z_2.
    \end{equation}
    For a total Charlie assignment \(\zz=(z_0,z_1)\in Q\), set
    \[
        H_{\zz}:=t_{1,z_1}-t_{0,z_0}=t+z_1s_1-z_0s_0\in\Z_{2N},
    \]
    \[
        d_{\zz}:=\gcd(N,H_{\zz}),
        \qquad
        D_{\zz}:=d_{\zz}\Z_N\leq\Z_N,
    \]
    and
    \[
        u_{\zz}:=t_{0,z_0}=z_0s_0\in\Z_N.
    \]
\end{definition}

Geometrically, a Charlie clock places Charlie's four possible outcome-dependent $\beta$-contributions on a common \(2N\)-clock, with the real offset \(\mu\) fixing the absolute rotation and the indices \(t_{l,z}\) recording the discrete positions of the ticks.  The parameters \(s_0\) and \(s_1\) are the spacings between opposite outcomes of $C_0$ and $C_1$, while \(t\) is the offset between the first tick of \(C_0\) and the first tick of \(C_1\).  For a total Charlie assignment \(\zz=(z_0,z_1)\), the number \(H_{\zz}=t_{1,z_1}-t_{0,z_0}\) will be shown to be how many ticks the clock advances between consecutive Alice literals in a witness path; such a step is obtained by following two consecutive implication edges connected to the intermediate Bob literal, one from each Charlie conditioning. Thus \(d_{\zz}=\gcd(N, H_{\zz})\) records the period of iterating this advancement on the measurement clock \(\mathbb Z_N\), and \(D_{\zz}=d_{\zz}\mathbb Z_N\) is exactly the return orbit subgroup: its cosets are the possible Alice index sets of \(\zz\)-witness paths.  Finally, \(u_{\zz}=t_{0,z_0}\) is the reflection offset for the first Charlie conditioning, so any Alice measurement $A$ involved in an impossible event in that conditioning forces Bob to use the reflected measurement \(u_{\zz}-B\).

Next, we define a class of Charlie clocks that we will, in the next subsection, show are those that allow us to construct the witness paths needed for a paradox.
\begin{definition}
    \label{def:paradoxical-clock}
    A Charlie clock \(\Gamma\) is \emph{paradoxical} if
    \[
        \frac{H_{\zz}}{d_{\zz}}
        \quad\text{is odd for every }\zz\in Q.
    \]
    Equivalently, if
    \[
        \operatorname{Odd}_N(a)
        \quad\Longleftrightarrow\quad
        \frac{a}{\gcd(N,a)}\text{ is odd},
    \]
    then \(\Gamma\) is paradoxical precisely when \(\operatorname{Odd}_N(H_{\zz})\) holds for every \(\zz\in Q\).
\end{definition}

We now define the class of abstract Charlie clock structures that are actually realised by a choice of interpolant state and a pair of measurements.
\begin{definition}
    \label{def:realisable-clock}
    A Charlie clock \(\Gamma=(N,t,s_0,s_1,\mu)\) is \emph{realisable} if there exist
    \[
        \lambda\in \piovertwointerval,
        \qquad
        C_0,C_1\in[0,\pi),
        \qquad
        C_0 < C_1,
    \]
    such that
    \[
        T_{l,z}
        \equiv
        \beta(\lambda,C_l+z\pi)
        \qquad
        \forall\, l,z\in\Z_2.
    \]
    Equivalently, the two Charlie measurements \(C_0,C_1\) and the interpolant state \(\Bsimple\) realise the four tick values of \(\Gamma\).
\end{definition}

Finally, we now define the set of valid clocks to be those that are both quantum-realisable and completable by a choice of Alice--Bob measurements to a paradox.  To avoid this set containing two clocks that represent the same paradox up to physical equivalence, we normalise the clock data.

\begin{definition}
    \label{def:valid-clock}
    A Charlie clock is \emph{valid} if it is paradoxical, realisable, and normalised: the realising measurements are labelled by
    their representatives \(0\le C_0<C_1<\pi\); in the GHZ case we use the additional
    phase freedom to impose \(C_0=0\), equivalently \(\mu=0\); and, among the two
    residual equatorial orientations, we take the lexicographically least clock.
\end{definition}
We denote the set of valid Charlie clocks by $\mathsf{Clock}$.

\subsubsection{Clock normal form}
\label{subsubsec:clock-normal-form}

We first explain how to extract the Charlie clock from a biconditional parity proof.  Lemma \ref{lem:interpolant-impl} says that in an interpolant-state scenario the relevant implication graphs are bidirected, so a contradiction is equivalent to the existence of a witness path from a literal to its complement.  In the biconditional case, a minimal witness path alternates between the edges of two Charlie measurements.  Moving two edges along such a path translates Alice's measurement angle by a difference of two Charlie ticks.  Since the path returns to the same Alice measurement with opposite outcome, that tick difference must be a rational multiple of \(\pi\).  Repeating this for all total Charlie assignments produces a common finite clock.

Recall~\cite[(8)]{de_Silva_2025} that for scenarios involving interpolant states, the impossibility equation \eqref{eqn:imposs_abc} has a much simpler form: an event $(A,B,C_l) \to (a,b,z_l)$ is impossible if and only if
\begin{equation}\label{eqn:imposs_interpolant}
    A + B \equiv T_{l,z} + (1 \oplus a \oplus b) \pi.
\end{equation}
More generally, a context $(A,B,C)$ contains impossible events if and only if
\begin{equation}\label{eqn:imposs_interpolant_mod_pi}
    A + B \equiv T_{l,z} \mod\pi.
\end{equation}

\begin{lemma}
    \label{lem:rational-ticks}
    Let \((\Bsimple,\mathcal M)\) be an interpolant-state paradox with $|M_3| = 2$.  Let
    \[
        T_{l,z}\in\R/2\pi\Z,
        \qquad
        l,z\in\Z_2,
    \]
    be the four Charlie tick values.  Then all tick differences are rational multiples of \(\pi\).  Hence there exist
    \[
        N\in\N_{>0},
        \qquad
        \mu\in\R/2N\Z,
        \qquad
        t_{l,z}\in\Z_{2N}
    \]
    such that
    \[
        T_{l,z}
        \equiv
        \frac{\pi}{N}(t_{l,z}+\mu)
        \qquad
        \forall\, l,z\in\Z_2.
    \]
\end{lemma}

\begin{proof}
    Fix a total Charlie assignment \(\zz \in Q\).  By Lemma~\ref{lem:interpolant-impl}, the associated implication graph contains a witness path.  Delete redundant subpaths and write a minimal witness path as
    \begin{equation}\label{eqn:dir-path}
        (A_{j_0},a_0)
        \Longrightarrow
        (B_{k_0},b_0)
        \Longrightarrow
        (A_{j_1},a_1)
        \Longrightarrow\cdots\Longrightarrow
        (B_{k_{L-1}},b_{L-1})
        \Longrightarrow
        (A_{j_0},a_0\oplus1),
    \end{equation}
    where the edges traversed alternately correspond to the Charlie conditionings $(C_0,z_0)$ and $(C_1,z_1)$. The integer $L$ is the number of Alice-to-Alice steps in this one-way path (without returning to the original literal).

    Considering the first two edges, by \eqref{eqn:imposs_interpolant_mod_pi} we have
    \[
        A_{j_0}+B_{k_0}
        \equiv
        T_{0,z_0}
        \mod\pi
        \qquad \text{and} \qquad
        A_{j_1}+B_{k_0}
        \equiv
        T_{1,z_1}
        \mod\pi.
    \]
    Subtracting gives
    \[
        A_{j_1}
        \equiv
        A_{j_0}+T_{1,z_1}-T_{0,z_0}
        \mod\pi.
    \]
    Repeating the same two-edge calculation along the path yields
    \[
        A_{j_q}
        \equiv
        A_{j_0}+q(T_{1,z_1}-T_{0,z_0})
        \mod\pi
    \]
    for \(q=0,\ldots,L-1\).  Then, traversing the final two edges in the witness path, we return to the Alice measurement $A_{j_0}$ and deduce that
    \[
        L(T_{1,z_1}-T_{0,z_0})
        \equiv0\mod\pi.
    \]
    Thus \(T_{1,z_1}-T_{0,z_0}\) is a rational multiple of \(\pi\).

    As \(\zz\) varies, this gives the four cross-differences
    \[
        T_{1,0}-T_{0,0},
        \quad
        T_{1,1}-T_{0,0},
        \quad
        T_{1,0}-T_{0,1},
        \quad
        T_{1,1}-T_{0,1}.
    \]
    These are the relevant tick differences directly seen by witness paths.  The remaining differences among the four ticks are obtained from these by addition and subtraction.  For example,
    \[
        T_{0,1}-T_{0,0}
        =(T_{1,0}-T_{0,0})-(T_{1,0}-T_{0,1}).
    \]
    Thus all tick differences are rational multiples of \(\pi\).  Choosing a common denominator over all tick differences gives the claimed clock representation.
\end{proof}

Given a clock, we can naturally write Alice and Bob's measurements relative to it by foliating $[0,\pi)$ into a disjoint union of translations of $\frac{\pi}{N}\Z_N$.

\begin{definition}
    \label{def:clock-layer-data}
    Let \(N\in\N_{>0}\) and let \(\mu\in[0,2N)\).  A \emph{clock layer decomposition} of Alice and Bob's measurement sets $M_1, M_2$ consists of a finite index set \(\I\), real numbers \(\alpha_i\in\R\), and subsets \(\J_i,\K_i\subseteq\Z_N\), for \(i\in \I\), such that Alice's and Bob's measurement sets have the form
    \begin{align*}
        M_1 & =\bigcup_{i\in \I}\frac{\pi}{N}(\J_i+\alpha_i),     \\
        M_2 & =\bigcup_{i\in \I}\frac{\pi}{N}(\K_i+\mu-\alpha_i).
    \end{align*}
    Two layers are required to have shifts distinct modulo \(\Z\).
\end{definition}

\begin{proposition}[Clock normal form]
    \label{prop:clock-layer-normal-form}
    Every minimal biconditional parity paradox is equivalent to one whose Charlie ticks form a Charlie clock and whose Alice and Bob measurements admit a clock layer decomposition.  Moreover, if an impossible Alice--Bob event occurs between a measurement in the \(i\)-th Alice layer and a measurement in the \(i'\)-th Bob layer, then \(i=i'\).
\end{proposition}

\begin{proof}
    Lemma~\ref{lem:rational-ticks} gives a common denominator \(N\) and tick indices \(t_{l,z}\).  Subtracting the common tick \(t_{0,0}\) from all tick indices gives the normalisation
    \[
        t_{0,0}=0,
        \qquad
        t_{0,1}=s_0,
        \qquad
        t_{1,0}=t,
        \qquad
        t_{1,1}=t+s_1.
    \]
    Dividing the quantities $N, t, s_0, s_1, \mu$ by any common divisor of \(N,t,s_0,s_1\), if necessary, gives \(\gcd(N,t,s_0,s_1)=1\).  Thus the Charlie data are a Charlie clock.

    We next construct the layers.  Let \(A\in M_1\).  Since the paradox is minimal, \(A\) appears in some impossible event.  Hence there are \(B\in M_2\), \(l\in\Z_2\), and \(z\in\Z_2\) such that
    \[
        A+B\equiv T_{l,z}\mod\pi.
    \]
    Write
    \[
        A=\frac{\pi}{N}(j+\alpha)
    \]
    with \(j\in\Z_N\) and \(\alpha\in\R\).  Then
    \[
        B\equiv
        \frac{\pi}{N}(t_{l,z}-j+\mu-\alpha)
        \mod\pi.
    \]
    Thus all Bob measurements which pair with Alice measurements in the coset \(\frac{\pi}{N}(\Z_N+\alpha)\) lie in the corresponding Bob coset \(\frac{\pi}{N}(\Z_N+\mu-\alpha)\).  Since there are finitely many measurements, only finitely many values of \(\alpha\) occur modulo \(\Z\).  These values define the layer set \(\I\), and the subsets \(\J_i,\K_i\subseteq\Z_N\) are the discrete indices occurring in each layer.

    Finally, suppose
    \begin{equation}\label{eqn:general_A_B_labelling}
        A_{i,j}:=\frac{\pi}{N}(j+\alpha_i),
        \qquad
        B_{i,j}:=\frac{\pi}{N}(k+\mu-\alpha_{i'})
    \end{equation}
    participate in an impossible event with tick \(T_{l,z}\).  Multiplying the congruence \(A_{i,j}+B_{i,j}\equiv T_{l,z}\mod\pi\) by \(N/\pi\) gives
    \[
        j+k+\alpha_i-
        \alpha_{i'}+
        \mu
        \equiv
        t_{l,z}+\mu
        \mod N.
    \]
    Hence \(\alpha_i-\alpha_{i'}\in\Z\).  Distinct layers have shifts distinct modulo \(\Z\), so \(i=i'\).
\end{proof}

As a consequence of the above proposition, each of the measurements involved in a witness path, required to establish a paradox, lives entirely within one layer. Having chosen a layer index $i \in \I$, we may simplify the notation and write $A_j := A_{i,j},\, B_k := B_{i,k}$. Moreover, upon fixing a total Charlie assignment $\zz = (z_0, z_1)$, we write $t_l := t_{l, z_l}$.

\begin{lemma}\label{lem:edges-vs-eqns}
    Let $i \in \I$ be a layer index and $\zz \in Q$ be a total Charlie assignment. Then the following are equivalent:
    \begin{enumerate}
        \item The implication graph $I_\Psi(\zz)$ contains a directed edge $(A_{j}, a) \Rightarrow (B_k, b \oplus 1)$.

        \item There exists $l \in \Z_2$ such that the quantities $j \in \J_i,\; k \in \K_i,\; a,b \in \Z_2$ satisfy
              \begin{equation}\label{eqn:interpolant-imposs-2N}
                  j + k \equiv t_l + (1 \oplus a \oplus b) N \mod{2N}.
              \end{equation}

        \item There exists $l \in \Z_2$ such that the quantities $j \in \J_i,\; k \in \K_i$ satisfy
              \begin{equation}\label{eqn:interpolant-imposs-N}
                  j + k \equiv t_l \mod{N}.
              \end{equation}
    \end{enumerate}
\end{lemma}
\begin{proof}
    The equivalence of statements 2 and 3 is immediate. We show that 1 and 2 are equivalent.

    Suppose that the implication graph contains a directed edge $(A_{j}, a) \Rightarrow (B_k, b \oplus 1)$. Thus we have a forbidden pair $((A_j, a), (B_k, b)) \in E(G(\zz))$, so that there exists $l \in \Z_2$ such that the event $(A_j, B_k, C_l) \to (a, b, z_l)$ is impossible. Take \eqref{eqn:imposs_interpolant}, substitute the expressions for $T_l$ \eqref{eqn:charlie-tick-decomp} and $A_j, B_k$ \eqref{eqn:general_A_B_labelling}, and multiply by $N/\pi$, to obtain \eqref{eqn:interpolant-imposs-2N}. Conversely, assuming that \eqref{eqn:interpolant-imposs-2N} holds for some $l \in \Z_2$, we may reverse the steps above to deduce that the event $(A_j, B_k, C_l) \to (a, b, z_l)$ is impossible, and hence that the implication graph $I_\Psi(\zz)$ contains a directed edge $(A_{j}, a) \Rightarrow (B_k, b \oplus 1)$.
\end{proof}

\begin{remark}\label{rem:edges-vs-parity-eqn}
    In the previous works on three-qubit nonlocality paradoxes~\cite{Abramsky2017, de_Silva_2025}, paradoxicality was primarily expressed via the inconsistency of $\Z_2$-linear systems. Given a total Charlie assignment $\zz$, we note that directed edges in the implication graph $I_\Psi(\zz)$ equivalently encode the same information as the linear system $\Psi(\zz)$, and that there is an explicit translation from one formalism to the other, recalling the derivation of~\eqref{eqn:z2-constraint} and Definition~\ref{def:imp-graph}.

    Explicitly, a pair of bidirected edges
    \[(A_j,a) \Longleftrightarrow (B_k, b\oplus 1),\qquad (B_k,b) \Longleftrightarrow (A_j, a \oplus 1)\]
    corresponds to the linear constraint
    \[g(A_j) \oplus g(B_k) = a \oplus b \oplus 1\]
    on compatible global assignments $g: \bigsqcup_{i=1}^3 M_i \to \Z_2$.
\end{remark}

\subsubsection{Valid Charlie clocks}
\label{subsubsec:valid-clocks}

We now completely classify the valid Charlie clocks.  This is the only point in the proof where quantum realisability, via the analytic functions \(\beta\) and \(\delta\), enter.  After this subsection, all remaining arguments are combinatorial and number-theoretic.

First, we define the pairs of possible $\beta$-contributions Charlie can make with his two outcomes for a single fixed measurement, assuming the state parameter $\lambda$ is also fixed.
\begin{definition}
    \label{def:tick-pair}
    Let \(\lambda\in \piovertwointerval\) and \(C\in[0,\pi)\).  The \emph{tick pair} associated to \((\lambda,C)\) is
    \[
        \Tpair(C):=\bigl(\beta(\lambda,C),\beta(\lambda,C+\pi)\bigr)
        \in(\R/2\pi\Z)^2.
    \]
    If \(\lambda>0\) and \(C>0\), this pair is written uniquely as
    \[
        \Tpair(C)=(-\tau,\sigma-\tau),
        \qquad
        0<\tau<\sigma<\pi.
    \]
    The special pair \((0,\pi)\) is the tick pair of \(C=0\).  We call $\sigma$ the \emph{spread} of the tick pair.
\end{definition}

Next, we determine which tick pairs arise from actual choices of $\lambda$ and Charlie measurement $C \in [0, \pi)$ and show how to extract $\lambda, C$ from the tick pair.

\begin{lemma}\label{lem:tick-pair-sin-lambda}
    Let $(\Bsimple, \M)$ be an interpolant-state scenario and let $C \in M_3$, $C > 0$ with associated tick pair $\Tpair(C) = (-\tau, \sigma - \tau)$. Then
    \begin{equation}\label{eqn:sin-lambda}
        \sin \lambda = \frac{\cos(\sigma/2)}{\cos(\tau - \sigma/2)}.
    \end{equation}
\end{lemma}
\begin{proof}
    See Appendix~\ref{sec:proof-tick-pair-sin-lambda}.
\end{proof}

\begin{proposition}
    \label{prop:one-tick-pair}
    A tick pair is realisable if and only if it has one of the following forms.

    \begin{enumerate}
        \item If \(\lambda=0\), then
              \[
                  \Tpair(C)=(-C,\pi-C)
              \]
              for some \(C\in[0,\pi)\).

        \item If \(\lambda>0\) and \(C=0\), then
              \[
                  \Tpair(C)=(0,\pi).
              \]

        \item If \(\lambda>0\) and \(C\in(0,\pi)\), then
              \[
                  \Tpair(C)=(-\tau,\sigma-\tau)
              \]
              for unique \(0<\tau<\sigma<\pi\). Conversely, every such pair is realised by unique \(C\in(0,\pi)\) and unique \(\lambda\), namely
              \begin{equation}\label{eqn:lambda-closed-form}
                  \lambda = \Lambda(\tau,\sigma) := \arcsin\left(\frac{\cos(\sigma/2)}{\cos(\tau - \sigma/2)}\right).
              \end{equation}
    \end{enumerate}
\end{proposition}

\begin{proof}
    For \(\lambda=0\), the identities \(\beta(0,C)=-C\) and \(\delta(0,C)\equiv\pi\) give the first case.  For \(\lambda>0\) and \(C=0\), the identities \(\beta(\lambda,0)=0\) and \(\delta(\lambda,0)\equiv\pi\) give the second case.  For \(\lambda>0\) and \(C\in(0,\pi)\), the inequalities \(0<\tau<\sigma<\pi\) follow from the range properties of \(\beta\) and \(\delta\).  See Appendix~\ref{sec:one-tick-pair} for the proof of the converse.  
\end{proof}

Next we characterise when two tick pairs can be realised by two Charlie measurements on a single state. Two distinct tick pairs $\Tpair_0 = (-\tau_0, \sigma_0 - \tau_0)$ and $\Tpair_1 = (-\tau_1, \sigma_1 - \tau_1)$ are quantum realisable if and only if they can be realised by the same $\lambda$. If one of the tick pairs is $(0, \pi)$, then $\lambda$ is uniquely determined by the other tick pair via \eqref{eqn:lambda-closed-form}. Otherwise, it is necessary and sufficient that
\begin{gather}
    \Lambda(\tau_0, \sigma_0) = \Lambda(\tau_1, \sigma_1) \nonumber\\
    \implies\quad \frac{\cos(\sigma_0/2)}{\cos(\tau_0 - \sigma_0/2)} = \frac{\cos(\sigma_1/2)}{\cos(\tau_1 - \sigma_1/2)}.\label{eqn:equal-lambdas}
\end{gather}

\begin{theorem}
    \label{thm:two-tick-pairs}
    Let
    \[
        \Tpair_0=(-\tau_0,\sigma_0-
        \tau_0),
        \qquad
        \Tpair_1=(-\tau_1,\sigma_1-
        \tau_1)
    \]
    be two distinct tick pairs, with \(0\leq\tau_l<\sigma_l\leq\pi\) for \(l\in\Z_2\), ordered so that \(\tau_0<\tau_1\).  They are realised by two distinct Charlie measurements on the same interpolant state if and only if exactly one of the following cases holds.

    \begin{enumerate}
        \item[(a)] \textbf{GHZ.}  \(\sigma_0=\sigma_1=\pi\).

        \item[(b)] \textbf{Non-GHZ with \(X\).}  \(\Tpair_0=(0,\pi)\).

        \item[(c)] \textbf{Non-GHZ with equal spread.}  \(\sigma_0=\sigma_1=\sigma<\pi\), and
              \[
                  \tau_0+\tau_1=\sigma.
              \]

        \item[(d)] \textbf{Non-GHZ with unequal spread.}  \(\sigma_0\neq\sigma_1\), and
              \begin{equation}\label{eqn:insane-case-d-condition}
                  \begin{gathered}
                      0 < \Theta < 1,\quad
                      \text{where}\quad \Theta := \frac{\kappa_0^2 + \kappa_1^2 - 2\kappa_0\kappa_1\cos\omega}{\sin^2\omega},\\[0.7em]
                      \text{with}\quad \kappa_l := \cos\frac{\sigma_l}{2},\quad \omega := \tau_0 - \tau_1 + \frac{\sigma_1 - \sigma_0}{2}.
                  \end{gathered}
              \end{equation}
    \end{enumerate}
    In all cases, the realising triple \((\lambda,C_0,C_1)\) is unique.
\end{theorem}

\begin{proof}
    First, note that whenever the tick pairs fall under one of the cases (a)--(c), they are quantum-realisable by Proposition \ref{prop:one-tick-pair}. For (d), we claim that, whenever $\sigma_0 \neq \sigma_1$ (so that $\lambda > 0$), \eqref{eqn:equal-lambdas} holds if and only if \eqref{eqn:insane-case-d-condition} holds.

    Let $\theta_l := -\tau_l + \sigma_l/2$, so that $\omega = \theta_1 - \theta_0$. By \eqref{eqn:sin-lambda} and \eqref{eqn:equal-lambdas}, we have quantum-realisable tick pairs if and only if
    \begin{equation*}
        \sin\lambda = \frac{\cos(\sigma_0/2)}{\cos\theta_0} = \frac{\cos(\sigma_1/2)}{\cos\theta_1}.
    \end{equation*}
    Then we have
    \begin{equation}\label{eqn:clever_sub}
        \cos\theta_l = \frac{\kappa_l}{\sin\lambda}\quad \forall l \in\Z_2.
    \end{equation}
    Now, starting with
    \begin{equation*}
        \cos\theta_1 =  \cos(\theta_0 + \omega),
    \end{equation*}
    expand the right-hand side and substitute \eqref{eqn:clever_sub} to eliminate the $\theta_l$. Rearranging gives
    \begin{equation}\label{eqn:lambda_explicit}
        \sin^2 \lambda = \Theta,
    \end{equation}
    which is satisfied by a unique $\lambda \in \left(0, \piovertwo\right)$ if and only if $0 < \Theta < 1$.

    Now, we show the converse by claiming that (a)--(d) exhaust all possible cases of two quantum-realisable tick pairs.

    If $\lambda = 0$, then for all $C \in [0, \pi)$ we have $\beta(\lambda, C) \equiv -C$ and $\delta(\lambda, C) \equiv \pi$. Hence we recover case (a).

    Now we consider all cases when $\lambda > 0$. By the discussion above this theorem, if one Charlie measurement is an $X$-measurement, $C_0 = 0$, then we recover case (b). If instead $C_0, C_1 > 0$, then $\sigma_0,\sigma_1 < \pi$. We separate into the two remaining cases: when $\sigma_0 = \sigma_1$ and when $\sigma_0 \neq \sigma_1$.

    If $\sigma_0 = \sigma_1 = \sigma$, then by \eqref{eqn:lambda-closed-form} and \eqref{eqn:equal-lambdas} the tick pairs are quantum-realisable if and only if
    \begin{equation*}
        \cos(\tau_0 - \sigma/2) = \cos(\tau_1 - \sigma/2).
    \end{equation*}
    Since $-\sigma/2 < \tau_l - \sigma/2 < \sigma/2$, this equation is satisfied if and only if $\tau_0 - \sigma/2 = -(\tau_1 - \sigma/2)$, i.e.\ $\tau_0 + \tau_1 = \sigma$. Hence we recover case (c). 

    Finally, case (d) then generally covers all instances of two quantum-realisable tick pairs with $\sigma_0 \neq \sigma_1$.

    In all four cases, uniqueness of $(\lambda,C_0,C_1)$ follows from Proposition~\ref{prop:one-tick-pair}.
\end{proof}

Finally, we adjoin the conditions of a Charlie clock to be paradoxical to it being realisable, and characterise the valid clocks.

\begin{theorem}
    \label{thm:valid-clocks}
    A Charlie clock \(\Gamma=(N,t,s_0,s_1,\mu)\) is valid if and only if it is the normalised representative (in the sense of Definition~\ref{def:valid-clock}) of a clock in one of the following four families.

    \begin{enumerate}
        \item[(a)] \textbf{GHZ clocks.}
              Realisability conditions:
              \begin{equation*}
                  \begin{gathered}
                      s_0 = s_1 = N,\quad \mu = 0,\quad t \equiv -v \mod{2N},\\
                      \text{where}\quad 1 \leq v < N,\quad \gcd(v,N) = 1.
                  \end{gathered}
              \end{equation*}
              Paradoxicality conditions:
              \begin{equation*}
                  \begin{gathered}
                      v \text{ odd},\quad N \text{ even}.
                  \end{gathered}
              \end{equation*}

        \item[(b)] \textbf{Non-GHZ clocks with \(X\).}
              Realisability conditions:
              \begin{equation*}
                  \begin{gathered}
                      s_0 = N,\quad s_1 = s < N,\quad \mu = 0,\quad t \equiv -v \mod{2N},\\
                      \text{where}\quad 1 \leq v \leq s-1.
                  \end{gathered}
              \end{equation*}
              Paradoxicality conditions:
              \begin{equation*}
                  \Odd_N(v),\quad \Odd_N(s-v),\quad \Odd_N(N+v),\quad \Odd_N(N+v-s).
              \end{equation*}

        \item[(c)] \textbf{Non-GHZ clocks with equal spread.}
              Realisability conditions:
              \begin{equation*}
                  s_0 = s_1 = s < N,\quad 1 \leq t < s,\quad \mu \equiv -\frac{t+s}{2} \mod{2N}.
              \end{equation*}
              Paradoxicality conditions:
              \begin{equation*}
                  \Odd_N(t),\quad \Odd_N(t+s), \quad \Odd_N(t-s).
              \end{equation*}

        \item[(d)] \textbf{Non-GHZ clocks with unequal spread.}
              Realisability conditions:
              \begin{equation*}
                  \begin{gathered}
                      s_0 \neq s_1,\quad t \neq 0,\quad 0 < \Theta < 1\\
                      \text{where}\quad \Theta := \frac{\kappa_0^2 + \kappa_1^2 - 2\kappa_0\kappa_1\cos\omega}{\sin^2\omega},\\
                      \text{with}\quad \kappa_l := \cos\left(\piovertwoN s_l\right),\quad \omega := \piovertwoN(2t + s_1 - s_0).
                  \end{gathered}
              \end{equation*}
              Whenever these quantum realisability conditions hold, we have $\mu \equiv -\frac{N}{\pi}\tau_0 \mod{2N}$, where $\tau_0$ is the unique solution in $(0, \pi)$ to
              \begin{equation*}
                  \frac{\kappa_0}{\cos(\tau_0 - \piovertwoN s_0)} = \frac{\kappa_1}{\cos(\tau_0 - \pioverN t - \piovertwoN s_0)}.
              \end{equation*}
              Paradoxicality conditions:
              \begin{equation*}
                  \Odd_N(t),\quad \Odd_N(t+s_1),\quad \Odd_N(t-s_0),\quad \Odd_N(t+s_1-s_0).
              \end{equation*}
    \end{enumerate}
\end{theorem}

\begin{proof}
    The four realisability families are exactly the four cases of Theorem~\ref{thm:two-tick-pairs}, translated into clock parameters.  The additional paradoxicality requirement is
    \[
        \operatorname{Odd}_N(H_{\zz})
        \qquad
        \forall\zz\in Q.
    \]
    Since the four values of \(H_{\zz}\) are
    \[
        t,
        \qquad
        t+s_1,
        \qquad
        t-s_0,
        \qquad
        t+s_1-s_0,
    \]
    this is precisely the stated collection of oddness conditions, with the evident simplifications in cases (a)--(c).  Hence the displayed list is necessary and sufficient.

    Below, we explicitly list the tick pairs in each case, written in terms of the updated parameters. For case (d), the equation determining $\tau_0$ and hence $\mu$ is obtained from \eqref{eqn:equal-lambdas}, noting that $\sigma_l = \pioverN s_l$ and $\tau_1 \equiv \tau_0 - \pioverN t$.

    \begin{enumerate}[label=(\alph*)]
        \item $\Tpair_0 = (0, \pi),\quad \Tpair_1 = \left(-\pioverN v,\, \pioverN (N-v)\right)$;
        \item $\Tpair_0 = (0, \pi),\quad \Tpair_1 = \left(-\pioverN v,\, \pioverN(s-v)\right)$;
        \item $\Tpair_0 = \left(-\piovertwoN(t + s),\, \piovertwoN(s - t)\right),\quad \Tpair_1 = \left(\piovertwoN(t-s),\, \piovertwoN(t+s)\right)$;
        \item $\Tpair_0 = \left(\pioverN \mu,\, \pioverN(\mu + s_0)\right),\quad \Tpair_1 = \left(\pioverN(\mu + t),\, \pioverN(\mu + t + s_1)\right)$.
    \end{enumerate}
\end{proof}

\begin{proposition}
    \label{prop:clock-separation}
    Two biconditional parity paradoxes arising from distinct normalised valid Charlie clocks are inequivalent.
\end{proposition}

\begin{proof}
    For interpolant states, the equivalence relation of Section~\ref{sec:equivalence} preserves the interpolant parameter and the two Charlie measurements, up to the stated reflection symmetry. Hence it preserves the two tick pairs. By Theorem~\ref{thm:two-tick-pairs}, those tick pairs determine the normalised clock. Distinct normalised valid clocks therefore give inequivalent proofs.
\end{proof}

We have now completed the state-and-Charlie part of the classification.  The next subsections fix a valid clock and determine exactly which Alice--Bob measurements complete it to a minimal biconditional parity proof.

\subsection{Alice--Bob completions}
\label{subsec:ab-completions}

Fix a valid Charlie clock $\Gamma$. For each total Charlie assignment $\zz$, Definition \ref{def:charlie-clock} associates a subgroup \(D_{\zz} \leq \Z_N\). This subsection shows how the Alice and Bob measurements constituting a witness path in the implication graph $I_\Psi(\zz)$ are precisely described by cosets of $D_\zz$. More specifically, if Alice’s measurements lie in the coset \(y+D_{\zz}\), then Bob’s measurements are forced to lie in the reflected coset \(u_{\zz}-y+D_{\zz}\).

While a paradox determines Alice (and therefore Bob) coset choices for each $\zz$, we do not necessarily need a coset for \emph{every} $\zz$ to reconstruct the paradox.  This is because after choosing cosets only for some selected $S \subset Q$, their union may already contain cosets for each $\zz \in Q \setminus S$. This is decided by the \emph{shadow test}, a simple congruence-divisibility criterion.

Much of the technical work in this subsection is devoted to reducing the problem from specifying all coset witnesses to specifying only a distinguished subset of them. The payoff is the notion of a \emph{canonical Alice--Bob completion}, which identifies the minimal data needed to reconstruct the Alice and Bob measurement sets uniquely, up to arbitrary real-valued labels assigned to the layers.

\subsubsection{Completions and measurement supports}
\label{subsubsec:completions}

Before the identification of witness paths with cosets in Section~\ref{subsubsec:witness-cosets}, we first introduce a set-theoretic construction that records the measurements that must appear in Alice’s and Bob’s measurement sets in order to contain a coset witness path. We then define an Alice--Bob completion as the data of a selected collection of total Charlie assignments, a partition of these assignments into layers, and a choice of coset for each selected assignment. In the remainder of this subsection, we show that these cosets indeed give witness paths, and we develop a simple shadow test for determining when a partial selection of layer-indexed cosets already contains all witness paths needed to certify a paradox.

\begin{definition}
    \label{def:witness-carrier}
    Let $\Gamma$ be a Charlie clock, $\I$ be a finite set, $i\in \I$, $\zz\in Q$, and let
    \[
        y\in\Z_N/D_{\zz}.
    \]
    The \emph{witness carrier} of type \(\zz\), layer \(i\), and representative \(y\) is the subset
    \[
        \mathsf W_{\zz}(i,y)
        \subseteq \I\times\Z_N\times\{\mathsf A,\mathsf B\}
    \]
    defined by
    \[
        \begin{aligned}
            \mathsf W_{\zz}(i,y)
            := \{(i,j,\mathsf A) \mid j\in y+D_{\zz}\} \cup
            \{(i,k,\mathsf B) \mid k\in u_{\zz}-y+D_{\zz}\}.
        \end{aligned}
    \]
\end{definition}

\begin{definition}
    \label{def:ab-completion}
    Let $\Gamma$ be a valid Charlie clock. An \emph{Alice--Bob completion} over $\Gamma$ is a tuple
    \[
        P=(S,\I,\iota,\yy)
    \]
    where
    \[
        S\subseteq Q,
        \qquad
        \I=\{0,\ldots,g\}\text{ for some } g\geq 0,
    \]
    \[
        \iota:S\twoheadrightarrow \I
    \]
    is a surjective function, and $\yy$ is an $S$-indexed vector of coset representatives:
    \[
        y_{\zz}\in\Z_N/D_{\zz}
        \qquad
        \forall\zz\in S.
    \]
    The \emph{measurement support} of \(P\) is
    \[
        \mathsf{MSupp}(P)
        :=
        \bigcup_{\zz\in S}
        \mathsf W_{\zz}(\iota(\zz),y_{\zz})
        \subseteq \I\times\Z_N\times\{\mathsf A,\mathsf B\}.
    \]
    For \(i\in \I\), its Alice and Bob index sets are
    \[
        \J_i(P):=
        \{j\in\Z_N\mid(i,j,\mathsf A)\in\mathsf{MSupp}(P)\},
    \]
    and
    \[
        \K_i(P):=
        \{k\in\Z_N\mid(i,k,\mathsf B)\in\mathsf{MSupp}(P)\}.
    \]
    Equivalently,
    \[
        \J_i(P)=
        \bigcup_{\zz \in \iota^{-1}(i)}(y_{\zz}+D_{\zz}),
        \qquad
        \K_i(P)=
        \bigcup_{\zz \in \iota^{-1}(i)}(u_{\zz}-y_{\zz}+D_{\zz}).
    \]
\end{definition}

Throughout, every coset representative is taken to be the least non-negative integer representative of its coset.
Every valid clock admits Alice--Bob completions: take \(S=Q\), one layer, and any representatives \(y_{\zz}\in\Z_N/D_{\zz}\).  We can remove one measurement at a time until the paradox is minimal.

\subsubsection{Witness paths as paired cosets}
\label{subsubsec:witness-cosets}

We now prove that witness carriers are exactly the witness paths that demonstrate $\Psi(\zz)$ is inconsistent once they have been expressed in terms of the Charlie clock.

\begin{lemma}
    \label{lem:paths-layerwise}
    Let $\Gamma$ be a valid clock, let \(P\) be an Alice--Bob completion over $\Gamma$, and choose any real shifts \(\alpha_i\) for the layers, distinct modulo \(\Z\).  In the physical scenario determined by $\Gamma$, \(P\), and these shifts, every implication edge joins Alice and Bob measurements in the same layer.  Consequently, every witness path lies entirely in one layer.
\end{lemma}

\begin{proof}
    The proof is the same calculation as in Proposition~\ref{prop:clock-layer-normal-form}.  If an Alice measurement in layer \(i\) and a Bob measurement in layer \(i'\) participate in an impossible event, then their angles have the form
    \[
        \frac{\pi}{N}(j+\alpha_i),
        \qquad
        \frac{\pi}{N}(k+\mu-\alpha_{i'}).
    \]
    The impossibility congruence forces \(\alpha_i-\alpha_{i'}\in\Z\).  Since distinct layers have shifts distinct modulo \(\Z\), one has \(i=i'\).  Thus every edge lies inside one layer.  A path is a sequence of edges sharing vertices, and a vertex has a unique layer, so the entire path remains in one layer.  Notice that the conclusion depends only on the fact that the layer shifts are distinct; the particular values of the shifts do not enter the finite congruence calculations below.
\end{proof}

The following maps will be shown to capture the effect of taking two nontrivial steps in a witness path from an Alice literal, to a Bob literal, to the next Alice literal, as expressed in terms of the clock.  We require the effect on both the measurement label and the literal (measurement with outcome).

\begin{definition}
    \label{def:return-maps}
    Let $\Gamma$ be a Charlie clock and let \(\zz\in Q\).  The \emph{measurement return map} for \(\zz\) is
    \[
        R_{\zz}:\Z_N\longrightarrow\Z_N,
        \qquad
        R_{\zz}(j)=j+H_{\zz}.
    \]
    The \emph{literal return map} for \(\zz\) is
    \[
        \widetilde R_{\zz}:\Z_{2N}\longrightarrow\Z_{2N},
        \qquad
        \widetilde R_{\zz}(j)=j+H_{\zz}.
    \]
    For \(y\in\Z_N\), the orbit of \(y\) under \(R_{\zz}\) is the coset
    \[
        y+D_{\zz}\subseteq\Z_N.
    \]
\end{definition}

Observe that, within a fixed layer, a literal can be equivalently written in a clock form as follows.

\begin{remark}\label{rem:clocks}
    Within each layer of Alice and Bob measurements, we have the following bijective correspondence between literals and literal clock points in $\Z_{2N}$:
    \begin{equation*}
        (A_{j}, a) \longleftrightarrow j + aN,\qquad (B_{k}, b) \longleftrightarrow k + bN.
    \end{equation*}
    Thus we may relabel a minimal directed path of the form of \eqref{eqn:dir-path} as
    \begin{equation}\label{eqn:dir-path-relabelled}
        j_0 + a_0N \longrightarrow k_0 + b_0N \longrightarrow \cdots \longrightarrow j_q + a_q N \longrightarrow \cdots \longrightarrow j_0 + a_L N.
    \end{equation}
\end{remark}

The following lemma establishes that, assuming that a minimal directed path involves Alice and Bob measurements that have a coset structure consistent with the measurement return map, there is a well-defined \emph{lifting} from measurement indices in $\Z_N$ to literals in $\Z_{2N}$, hence justifying the well-definedness of the literal return map.

\begin{lemma}\label{lem:lifting}
    Let \(\Gamma\) be a Charlie clock, \(\zz\in Q\), and $i \in \I$.

    Suppose that, in layer $i$, the implication graph $I_\Psi(\zz)$ contains a minimal directed path
    \begin{equation*}
        j_0 + a_0N \longrightarrow k_0 + b_0N \longrightarrow \cdots \longrightarrow j_q + a_q N \longrightarrow \cdots \longrightarrow j_0 + a_L N,
    \end{equation*}
    where
    \begin{gather*}
        j_{q+1} \equiv R_\zz(j_q) \mod{N},\quad \text{i.e.}\quad j_q \equiv j_0 + q\HH_\zz \mod{N},\\
        \text{and} \quad k_q \equiv u_\zz - j_q \mod{N}
    \end{gather*}
    for $q = 0,\ldots,L-1$, so that $L = N/\dd_\zz$. Identify $j_L = j_0$.

    Then, for all $q = 0, \ldots, L$, we have
    \begin{align}
        j_q + a_qN & \equiv j_0 + a_0N + q\HH_\zz \mod{2N},\label{eqn:2N-shift} \\
        k_q + b_qN & \equiv t_0 - (j_q + a_qN) \mod{2N},\nonumber
    \end{align}
    so that we can equivalently write the directed path as
    \begin{equation}\label{eqn:dir-path-with-liftings}
        \tilde{j}_0 \longrightarrow t_0 - \tilde{j}_0 \longrightarrow \tilde{j}_0 + \HH_\zz \longrightarrow \cdots \longrightarrow \tilde{j}_0 + q\HH_\zz \longrightarrow \cdots \longrightarrow \tilde{j}_0 + \frac{N}{\dd_\zz}\HH_\zz,
    \end{equation}
    where $\tilde{j}_0 := j_0 + a_0 N$.

    Thus, the directed path is a witness path if and only if $\HH_\zz/\dd_\zz$ is odd.
\end{lemma}
\begin{proof}
    We prove the lemma for Alice's measurement indices by induction; the argument for Bob's is symmetric. Clearly \eqref{eqn:2N-shift} holds for $q = 0$. Now assume that \eqref{eqn:2N-shift} holds for some $q = 0,\ldots,L-1$. Then by Lemma~\ref{lem:edges-vs-eqns} we have
    \begin{align}
                        & j_q + k_q \equiv t_0 + (1 \oplus a_q \oplus b_q)N \mod 2N,\nonumber               \\
        \text{and}\quad & j_{q+1} + k_q \equiv t_1 + (1 \oplus a_{q+1} \oplus b_q)N \mod 2N\nonumber        \\
        \implies\quad   & j_{q+1} + a_{q+1}N \equiv j_q + a_q N + \HH_\zz \mod 2N,\label{eqn:2N-difference}
    \end{align}
    thus showing the inductive step.

    The path is a witness path if and only if $j_0 + a_0 \not\equiv j_0 + a_0 + N\frac{\HH_\zz}{\dd_\zz} \mod 2N$, i.e.\ $\HH_\zz/\dd_\zz$ is odd.
\end{proof}

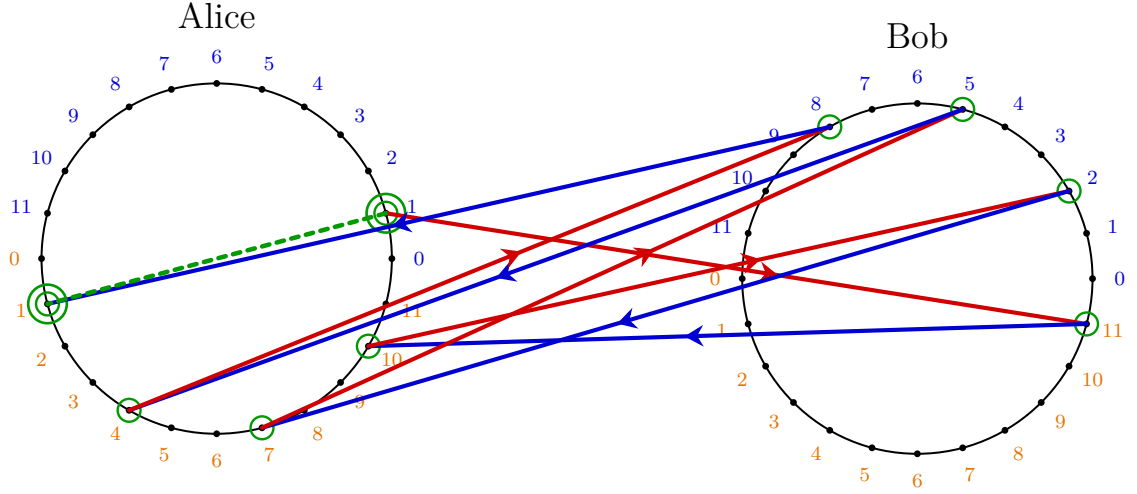
\begin{figure}
    \centering
    \resizebox{\textwidth}{!}{%
        \begin{tikzpicture}[
                line cap=round,
                line join=round,
                font=\small,
                midarrow/.style={
                        postaction={decorate},
                        decoration={
                                markings,
                                mark=at position 0.56 with {
                                        \arrow{Stealth[length=7.5pt,width=8.5pt]}
                                    }
                            }
                    },
                redpath/.style={red!82!black,line width=1.65pt,midarrow},
                bluepath/.style={blue!82!black,line width=1.65pt,midarrow},
                hig/.style={draw=green!60!black,line width=0.95pt}
            ]
            \def\R{2.45}
            \def\LabR{2.83}
            \coordinate (Acent) at (-4.9,0.18);
            \coordinate (Bcent) at ( 4.9,-0.10);

            \draw[thick] (Acent) circle (\R);
            \draw[thick] (Bcent) circle (\R);
            \node[font=\Large] at ($(Acent)+(0,3.4)$) {Alice};
            \node[font=\Large] at ($(Bcent)+(0,3.4)$) {Bob};

            \foreach \k in {0,...,23} {
                    \pgfmathsetmacro{\ang}{15*\k}
                    \coordinate (A\k)  at ($(Acent)+({\R*cos(\ang)},{\R*sin(\ang)})$);
                    \coordinate (B\k)  at ($(Bcent)+({\R*cos(\ang)},{\R*sin(\ang)})$);
                    \coordinate (Al\k) at ($(Acent)+({\LabR*cos(\ang)},{\LabR*sin(\ang)})$);
                    \coordinate (Bl\k) at ($(Bcent)+({\LabR*cos(\ang)},{\LabR*sin(\ang)})$);
                    \fill[black] (A\k) circle (1.35pt);
                    \fill[black] (B\k) circle (1.35pt);
                }

            \foreach \k in {0,...,11} {
                    \node[font=\scriptsize,text=blue!85!black] at (Al\k) {$\k$};
                    \node[font=\scriptsize,text=blue!85!black] at (Bl\k) {$\k$};

                    \pgfmathtruncatemacro{\kk}{\k+12}
                    \node[font=\scriptsize,text=orange!90!black] at (Al\kk) {$\k$};
                    \node[font=\scriptsize,text=orange!90!black] at (Bl\kk) {$\k$};
                }

            \foreach \p in {1,13,16,19,22} {
                    \draw[hig] (A\p) circle (0.16);
                }
            \foreach \p in {2,5,8,23} {
                    \draw[hig] (B\p) circle (0.16);
                }

            \draw[hig,line width=1.05pt] (A1)  circle (0.27);
            \draw[hig,line width=1.05pt] (A13) circle (0.27);

            \draw[redpath]  (A1)  -- (B23);
            \draw[bluepath] (B23) -- (A22);
            \draw[redpath]  (A22) -- (B2);
            \draw[bluepath] (B2)  -- (A19);
            \draw[redpath]  (A19) -- (B5);
            \draw[bluepath] (B5)  -- (A16);
            \draw[redpath]  (A16) -- (B8);
            \draw[bluepath] (B8)  -- (A13);

            \draw[dashed,green!60!black,line width=1.8pt] (A1) -- (A13);

        \end{tikzpicture}%
    }
    \caption{A concrete witness path in the bipartite implication graph, drawn on the literal \(2N\)-clock for \(N=12\). Each circle is a copy of \(\mathbb Z_{24}\): the blue and orange labels are respectively the outcome-\(0\) and outcome-\(1\) copies of the measurement clock \(\mathbb Z_{12}\), so antipodal points represent complementary literals of the same measurement. The green-ringed vertices are the literals traversed by the path; the red and blue arrows are implication edges arising from the two Charlie conditionings, and the path alternates between these two partial matchings. Its Alice measurement indices form the return-map coset \(1+3\mathbb Z_{12}=\{1,4,7,10\}\), while the corresponding Bob indices form the paired reflected coset \(2+3\mathbb Z_{12}=\{2,5,8,11\}\). For convenience, we write $A_j^a$ for the literal $(A_j,a)$ and so on. On the literal clock the highlighted path \(A_1^0\Rightarrow B_{11}^1\Rightarrow A_{10}^1\Rightarrow B_2^0\Rightarrow A_7^1\Rightarrow B_5^0\Rightarrow A_4^1\Rightarrow B_8^0\Rightarrow A_1^1\) starts at \(A_1^0\) and ends at its complement \(A_1^1\), marked by the prominent dashed green chord. Thus the picture displays the lift from a coset orbit in \(\mathbb Z_{12}\) to a genuine witness path between antipodal literals in \(\mathbb Z_{24}\).}
    \label{fig:two-charlie-2N-clock-witness}
\end{figure}

Now we show that the structure of a minimal directed path considered in Lemma~\ref{lem:lifting} in fact applies to \emph{every} possible minimal directed path, allowing us to deduce an equivalence between Alice and Bob measurement cosets and witness paths.

\begin{proposition}
    \label{prop:witness-cosets}
    Let \(\Gamma\) be a valid Charlie clock, let \(P\) be an Alice--Bob completion over \(\Gamma\), and choose any distinct real shifts for the layers.  Fix \(i\in \I\) and \(\zz\in Q\).  In the resulting physical scenario, the implication graph for the total Charlie assignment \(\zz\) contains a \(\zz\)-witness path contained in layer \(i\) if and only if there exists
    \[
        y\in\Z_N/D_{\zz}
    \]
    such that
    \[
        y+D_{\zz}
        \subseteq
        \J_i(P),
        \qquad
        u_{\zz}-y+D_{\zz}
        \subseteq
        \K_i(P).
    \]
    Equivalently, the measurement support contains the witness carrier \(\mathsf W_{\zz}(i,y)\).
\end{proposition}

\begin{proof}
    By Lemma~\ref{lem:paths-layerwise}, any witness path lies in a single layer.  Fix such a layer \(i\).  Along a minimal path for the total assignment \(\zz=(z_0,z_1)\), the two Charlie conditionings alternate.  Two consecutive edges therefore move Alice's index by
    \[
        H_{\zz}=t_{1,z_1}-t_{0,z_0}
    \]
    modulo \(N\).  If the Alice indices encountered are \(j_0,j_1,j_2,\ldots\), then
    \[
        j_{q+1}
        \equiv
        j_q+H_{\zz}
        \mod N.
    \]
    Thus the Alice indices in the path form the orbit
    \[
        j_0+D_{\zz}.
    \]
    The Bob index paired with an Alice index \(j\) by the first Charlie conditioning satisfies
    \[
        k\equiv u_{\zz}-j\mod N,
    \]
    so the Bob indices form
    \[
        u_{\zz}-j_0+D_{\zz}.
    \]
    Thus any witness path carries the displayed paired cosets.

    Conversely, suppose layer \(i\) contains the two cosets
    \[
        y+D_{\zz}\subseteq \J_i(P),
        \qquad
        u_{\zz}-y+D_{\zz}\subseteq \K_i(P).
    \]
    Starting from \(y\), repeatedly apply the measurement return map \(R_{\zz}\).  This runs through the coset \(y+D_{\zz}\).  The corresponding Bob indices all lie in the displayed Bob coset, so the implication graph contains the alternating path. Upon lifting measurement indices to literals, Lemma~\ref{lem:lifting} implies that the transformation between successive Alice literals is governed by \(\widetilde R_{\zz}\).  After \(N/d_{\zz}\) return steps the literal displacement is
    \[
        \frac{N}{d_{\zz}}H_{\zz}
        =
        N\frac{H_{\zz}}{d_{\zz}}.
    \]
    Since \(\Gamma\) is valid, it is paradoxical, so \(H_{\zz}/d_{\zz}\) is odd.  Hence the lift ends at the complementary literal, and the alternating path is a witness path.
\end{proof}

We have now established that, given a valid Charlie clock, a paradox is the same as a choice of representative for a coset of $D_\zz \leq \Z_N$, and layer index, for each $\zz \in Q$.  Next, we reduce this information to the minimum necessary and identify those Alice--Bob completions that do not contain redundant information.

\subsubsection{The shadow test}
\label{subsubsec:shadows}

An Alice--Bob completion may select witness carriers only for some total Charlie assignments.  Nevertheless, the measurement support generated by those selected carriers may contain further witness carriers.  The purpose of the shadow test is to determine exactly which carriers are already present.  Since all sets involved are cosets of cyclic groups, this test ultimately reduces to elementary congruence checks.

Fix a valid Charlie clock \(\Gamma\), and let $P=(S,\I,\iota,\yy)$ be an Alice--Bob completion.  For a target assignment \(\ww\in Q\), a layer \(i\in \I\), and a representative $r\in\Z_N/D_{\ww}$, we want to decide whether the full carrier $\mathsf W_{\ww}(i,r)$ is contained in \(\mathsf{MSupp}(P)\).  This means checking both
\[
    r+D_{\ww}
    \subseteq
    \J_i(P)
\]
on Alice's side, and
\[
    u_{\ww}-r+D_{\ww}
    \subseteq
    \K_i(P)
\]
on Bob's side.

It is useful to compare Alice and Bob in the same coordinates.  Reflect Bob's clock through \(u_{\ww}\), i.e.\ apply
\[
    k\longmapsto u_{\ww}-k.
\]
Then the target Bob coset \(u_{\ww}-r+D_{\ww}\) becomes \(r+D_{\ww}\).  A selected Bob source coset
\[
    u_{\zz}-y_{\zz}+D_{\zz}
\]
is correspondingly reflected to
\[
    y_{\zz}+u_{\ww}-u_{\zz}+D_{\zz}.
\]
Thus both Alice and reflected Bob ask the same question: which parts of the target coset \(r+D_{\ww}\) are covered by selected source cosets?

\begin{definition}
    \label{def:parity-quotient}
    Let \(\ww\in Q\).  The \emph{parity quotient} of the \(\ww\)-return subgroup is
    \[
        \Pi_{\ww}:=D_{\ww}/2D_{\ww}.
    \]
    Equivalently,
    \[
        \Pi_{\ww}\cong
        \begin{cases}
            0 = \{0\},      & N/d_{\ww}\text{ odd},  \\
            \Z_2 = \{0,1\}, & N/d_{\ww}\text{ even}.
        \end{cases}
    \]

    For \(r\in\Z_N/D_{\ww}\), the coset \(r+D_{\ww}\) is a translation of the group \(D_{\ww}\).  Its \emph{parity fibres} are the fibres of the quotient map
    \[
        r+D_{\ww}\longrightarrow D_{\ww}/2D_{\ww},
        \qquad
        r+h\longmapsto h+2D_{\ww}.
    \]
    Thus, writing the elements of \(\Pi_{\ww}\) as \(\varepsilon=0\) when \(\Pi_{\ww}\) is a singleton, and as \(\varepsilon=0,1\) when \(\Pi_{\ww}\cong\Z_2\), the parity fibre corresponding to $\varepsilon$ is
    \[
        r+\varepsilon d_{\ww}+2D_{\ww}
        \subseteq
        r+D_{\ww}.
    \]
    If \(N/d_{\ww}\) is odd, this is the whole coset.  If \(N/d_{\ww}\) is even, the two parity fibres are the two alternating halves of the coset.
\end{definition}

\begin{definition}
    \label{def:shadow}
    Let \(P=(S,\I,\iota,\yy)\) be an Alice--Bob completion.  Fix
    \[
        \ww\in Q,
        \qquad
        i\in \I,
        \qquad
        r\in\Z_N/D_{\ww},
        \qquad
        p\in\{0,1\}.
    \]
    For \(\zz\in S\), set
    \[
        c^p_{\ww\zz}:=
        y_{\zz}+p(u_{\ww}-u_{\zz})
        \in\Z_N/D_{\zz}.
    \]
    Here \(p=0\) denotes Alice, while \(p=1\) denotes Bob after reflection through \(u_{\ww}\).

    The \emph{\(p\)-shadow} of layer \(i\) on the candidate \((\ww,r)\) is the subset $\operatorname{Sh}^{p}_{i,\ww}(r) \subseteq \Pi_{\ww}$ defined by
    \[
        \operatorname{Sh}^{p}_{i,\ww}(r)
        :=
        \left\{
        \varepsilon\in\Pi_\ww \;\middle|\;
        \exists \ \zz\in \iota^{-1}(i)
        \text{ such that }
        d_\zz\mid 2d_\ww
        \text{ and }
        c^{p}_{\ww\zz}\equiv r+\varepsilon d_\ww \mod{d_\zz}
        \right\}.
    \]

    The \emph{solution set} of \(\ww\) in \(P\) is
    \[
        \operatorname{Sol}_P(\ww)
        :=
        \left\{
        (i,r) \Bigm\vert
        \operatorname{Sh}^{0}_{i,\ww}(r)=    \operatorname{Sh}^{1}_{i,\ww}(r)=\Pi_{\ww}
        \right\}.
    \]
\end{definition}

The divisibility condition in Definition~\ref{def:shadow} says that a selected source coset is large enough to contain a parity fibre of the target coset.  The congruence says which fibre it contains.

\begin{lemma}
    \label{lem:parity-fibre-containment}
    Let \(\ww,\zz\in Q\), let \(r\in\Z_N/D_{\ww}\), let \(c\in\Z_N/D_{\zz}\), and let \(\varepsilon\in\Pi_{\ww}\).  Then
    \[
        r+\varepsilon d_{\ww}+2D_{\ww}
        \subseteq
        c+D_{\zz}
    \]
    if and only if
    \[
        d_{\zz}\mid 2d_{\ww}
        \qquad\text{and}\qquad
        c\equiv r+\varepsilon d_{\ww}\mod{d_{\zz}}.
    \]
\end{lemma}

\begin{proof}
    The fibre \(r+\varepsilon d_{\ww}+2D_{\ww}\) is a coset of the subgroup \(2D_{\ww}\), and \(c+D_{\zz}\) is a coset of \(D_{\zz}\).  A coset \(a+H\) is contained in a coset \(b+K\) precisely when \(H\leq K\) and \(a\in b+K\).  Here the subgroup containment \(2D_{\ww}\leq D_{\zz}\) is equivalent to \(d_{\zz}\mid 2d_{\ww}\), and the point containment is the stated congruence.
\end{proof}

It remains to justify that no finer pieces are needed.  A priori, a target coset might be covered by thirds, quarters, or still smaller intersections.  The following two lemmas show that, in the case where $|M_3| = 2$ and hence $|Q| = 4$, any such cover is detected already on the quotient \(D_{\ww}/2D_{\ww}\).

\begin{lemma}[Small cyclic covers]
    \label{lem:small-cyclic-covers}
    Let \(C\) be a finite cyclic group. Suppose \(C\) is irredundantly covered by at most
    three proper cosets of subgroups. Then the ordered relative index pattern is one of
    \[
        (2,2),\qquad (3,3,3),\qquad (2,4,4).
    \]
\end{lemma}

\begin{proof}
    A coset of index \(m\) has size \(|C|/m\). If two proper cosets cover \(C\), their
    indices \(m_1,m_2>1\) satisfy
    \[
        1\le \frac1{m_1}+\frac1{m_2},
    \]
    hence \(m_1=m_2=2\).

    Now suppose three proper cosets irredundantly cover \(C\), with ordered indices
    \(2\le m_1\le m_2\le m_3\). Then
    \[
        1\le \frac1{m_1}+\frac1{m_2}+\frac1{m_3}.
    \]
    If \(m_1\ge 3\), this forces \(m_1=m_2=m_3=3\).

    It remains to consider \(m_1=2\). Let \(C_0\le C\) be the unique index-two
    subgroup, let \(P\) be the index-two coset appearing in the cover, and let \(P'\)
    be the other coset of \(C_0\). The two remaining cosets must cover \(P'\).
    Consider one of them, say \(E=c+H\), of index \(m=[C:H]\).

    If \(m\) is odd, then \(H\not\subseteq C_0\), so the quotient map
    \(C\to C/C_0\cong \mathbb Z_2\) maps \(H\) onto \(\mathbb Z_2\). Hence every
    coset of \(H\) meets \(P\) and \(P'\) equally, so
    \[
        |E\cap P'|=\frac{|E|}{2}=\frac{|C|}{2m}\le \frac{|C|}{6}.
    \]
    If \(m\) is even, then \(H\subseteq C_0\), so every coset of \(H\) lies entirely
    inside one parity class. Since the cover is irredundant, any remaining coset that
    helps cover \(P'\) must lie inside \(P'\). It cannot have index \(2\), since then it
    would equal \(P'\) and the third coset would be redundant. Thus its index is at
    least \(4\), and its size is at most \(|C|/4\).

    Therefore if either of the two remaining cosets had odd index, their total
    contribution to \(P'\) would be at most
    \[
        \frac{|C|}{6}+\frac{|C|}{4}<\frac{|C|}{2}=|P'|,
    \]
    so they could not cover $P'$. Hence both remaining cosets have even index, lie inside \(P'\), and
    each has size at most \(|C|/4\). Since together they cover \(P'\), both must have
    size exactly \(|C|/4\). Thus their indices are both \(4\), giving the pattern
    \((2,4,4)\).
\end{proof}

\begin{lemma}[Dyadic rigidity]
    \label{lem:dyadic-rigidity}
    Fix \(\ww\in Q\) and a target coset \(C=r+D_\ww\). Suppose \(C\) is covered by source
    cosets of types \(\zz\ne \ww\), with at most one source coset for each such \(\zz\). Then
    \(C\) is covered either by one whole source coset or by its two parity fibres.
\end{lemma}

\begin{proof}
    Intersect each source coset with \(C\), translate \(C\) to identify it with the cyclic
    group \(D_\ww\), and discard redundant intersections. A nonempty intersection is a
    coset of \(D_\ww\cap D_\zz\) inside \(D_\ww\), with relative index
    \[
        [D_\ww:D_\ww\cap D_\zz]=\frac{d_\zz}{\gcd(d_\ww,d_\zz)}.
    \]
    If some relative index is \(1\), then one source coset contains all of \(C\). Otherwise
    Lemma~\ref{lem:small-cyclic-covers} leaves only the patterns
    \[
        (2,2),\qquad (3,3,3),\qquad (2,4,4).
    \]
    The pattern \((2,2)\) is exactly the cover by the two cosets of the unique
    index-two subgroup \(2D_\ww\), namely the two parity fibres.

    We rule out the other two patterns using
    \[
        H_{00}+H_{11}=H_{01}+H_{10}.
    \]
    Equivalently, for any target \(\ww\), the value \(H_\ww\) is a signed sum of the three
    non-target \(H_\zz\)'s.

    Suppose first that \((3,3,3)\) occurs. Let \(a:=\nu_3(d_\ww)\), the 3-adic valuation of $d_\ww$. For every non-target
    \(\zz\), divisibility of the relative index
    \[
        \frac{d_\zz}{\gcd(d_\ww,d_\zz)}
    \]
    by \(3\) implies \(\nu_3(d_\zz)\ge a+1\). Since \(d_\zz=\gcd(N,H_\zz)\), we have
    \(3^{a+1}\mid N\) and \(3^{a+1}\mid H_\zz\) for all \(\zz\ne \ww\). The parallelogram
    identity then gives \(3^{a+1}\mid H_\ww\), hence
    \[
        3^{a+1}\mid \gcd(N,H_\ww)=d_\ww,
    \]
    contradicting the definition of \(a\).

    The pattern \((2,4,4)\) is excluded in the same way with the prime \(2\): all three
    non-target relative indices are divisible by \(2\), so all three non-target \(d_\zz\)'s
    have strictly larger \(2\)-adic valuation than \(d_\ww\). The parallelogram identity
    then forces the same extra factor of \(2\) into \(H_\ww\), and hence into \(d_\ww\), a
    contradiction.

    Thus only the whole-coset case and the two-parity-fibre case remain.
\end{proof}

\begin{proposition}
    \label{prop:shadow-exactness}
    Let \(P\) be an Alice--Bob completion over a valid clock \(\Gamma\).  For every \(\ww\in Q\), the set \(\operatorname{Sol}_P(\ww)\) is exactly the set of pairs \((i,r)\) such that
    \[
        \mathsf W_{\ww}(i,r)
        \subseteq
        \mathsf{MSupp}(P).
    \]
\end{proposition}

\begin{proof}
    Suppose first that \((i,r)\in\operatorname{Sol}_P(\ww)\).  The equality
    \[
        \operatorname{Sh}^{0}_{i,\ww}(r)=\Pi_{\ww}
    \]
    says that every parity fibre of \(r+D_{\ww}\) is contained in Alice's layer-\(i\) support.  Hence the Alice coset \(r+D_{\ww}\) is contained in \(\J_i(P)\).  Similarly,
    \[
        \operatorname{Sh}^{1}_{i,\ww}(r)=\Pi_{\ww}
    \]
    says that, after reflecting Bob through \(u_{\ww}\), every parity fibre of the reflected Bob target is contained in the reflected Bob support.  Reflecting back gives
    \[
        u_{\ww}-r+D_{\ww}\subseteq \K_i(P).
    \]
    Thus \(\mathsf W_{\ww}(i,r)\subseteq\mathsf{MSupp}(P)\).

    Conversely, suppose \(\mathsf W_{\ww}(i,r)\subseteq\mathsf{MSupp}(P)\).  Then \(r+D_{\ww}\) is covered by selected Alice source cosets in layer \(i\), and after reflecting Bob through \(u_{\ww}\), the same target coset is covered by reflected Bob source cosets in the same layer.  On each side, if the selected \(\ww\)-carrier itself is the target, then all parity fibres are covered.  Otherwise the cover uses only the three non-target assignments, so Lemma~\ref{lem:dyadic-rigidity} applies: the cover is whole or by the two parity fibres.  Hence the corresponding shadow is all of \(\Pi_{\ww}\) on both Alice and reflected Bob, and \((i,r)\in\operatorname{Sol}_P(\ww)\).
\end{proof}

Thus the shadow test computes the witness carriers already present in an Alice--Bob completion using only the elementary checks
\[
    d_{\zz}\mid 2d_{\ww}
    \qquad\text{and}\qquad
    c^p_{\ww\zz}\equiv r+\varepsilon d_{\ww}\mod{d_{\zz}}.
\]

\subsubsection{Canonical Alice--Bob completions}
\label{subsubsec:canonical-completions}

We now show that every minimal biconditional parity proof, which we have shown to determine a valid Charlie clock, also determines a unique canonical Alice--Bob completion.

\begin{definition}
    \label{def:canonical-completion}
    An Alice--Bob completion \(P=(S,\I,\iota,\yy)\) over a valid clock \(\Gamma\) is \emph{canonical} if
    \[
        \operatorname{Sol}_P(\zz)=\{(\iota(\zz),y_{\zz})\}
        \qquad
        \forall\zz\in S,
    \]
    and
    \[
        |{\operatorname{Sol}_P(\ww)}|\geq2
        \qquad
        \forall\ww\in Q\setminus S.
    \]
    We denote by \(\mathsf{Comp}(\Gamma)\) the set of canonical Alice--Bob completions over \(\Gamma\).
\end{definition}

Equivalently, \(P\) is canonical precisely when every total Charlie assignment has at least one solution and
\[
    S=
    \{\ww\in Q \mid
    \operatorname{Sol}_P(\ww)
    \text{ is a singleton}\}.
\]
This equivalence is perhaps the most useful way to remember the definition: the selected assignments are exactly the uniquely witnessed assignments.

\begin{theorem}
    \label{thm:canonicity-minimality}
    Let \(\Gamma\) be a valid Charlie clock and let \(P\) be an Alice--Bob completion over \(\Gamma\).  Then \(P\) is canonical if and only if \(\mathsf{MSupp}(P)\) is the measurement support of a minimal biconditional parity proof with clock \(\Gamma\).
\end{theorem}

\begin{proof}
    Assume first that \(P\) is canonical.  Every selected \(\zz\in S\) has its selected witness carrier, and every unselected \(\ww\notin S\) has at least two solutions.  Hence every total Charlie assignment has a witness carrier in the measurement support.  Since \(\Gamma\) is valid, these carriers lift to witness paths.  Therefore the generated scenario is paradoxical.

    Now remove any point of \(\mathsf{MSupp}(P)\).  By definition of measurement support, the point lies in at least one selected carrier \(\mathsf W_{\zz}(\iota(\zz),y_{\zz})\).  Since \(P\) is canonical, this carrier is the unique solution of type \(\zz\).  Removing the point breaks that carrier, and deleting points cannot create a new carrier.  Hence the reduced measurement support has no \(\zz\)-witness, so it is not paradoxical.  Thus the proof is minimal.

    Conversely, let \(X\subseteq \I\times\Z_N\times\{\mathsf A,\mathsf B\}\) be the measurement support of a minimal proof with clock \(\Gamma\).  For each \(\ww\in Q\), let \(\operatorname{Sol}_X(\ww)\) be the set of all pairs \((i,r)\) such that \(\mathsf W_{\ww}(i,r)\subseteq X\).  Define
    \[
        S_X:=
        \{\ww\in Q \mid
        \operatorname{Sol}_X(\ww)\text{ is a singleton}\}.
    \]
    For \(\zz\in S_X\), write the unique solution as \((i_{\zz},y_{\zz})\).  These data define an Alice--Bob completion \(P_X\).

    We claim that the singleton carriers cover all of \(X\).  Let \(m\in X\).  Since the proof is minimal, removing \(m\) destroys paradoxicality.  Therefore, for some \(\ww\in Q\), every \(\ww\)-witness in \(X\) contains \(m\).  But two distinct carriers of the same type \(\ww\) are disjoint, because they either lie in different layers or are distinct cosets of the same subgroup in a fixed layer.  Hence there is exactly one \(\ww\)-carrier in \(X\), and \(m\) lies on a singleton carrier.  This proves that the completion \(P_X\) reconstructs \(X\), and by construction it is canonical.
\end{proof}

Thus, a canonical completion is precisely the finite coset data that minimally and uniquely specifies the Alice--Bob measurement support of the proof.

\subsection{Assembling the classification}
\label{subsec:bpp-classification}

We now assemble the classification.  The valid clock gives the state and Charlie measurements.  The canonical Alice--Bob completion gives the finite measurement support.  The only remaining freedom is the assignment of a real number to each layer.  These numbers have the effect of translating each layer to a different disjoint copy of the clock and do not affect the essential structure of the paradox.  This freedom is recorded by a shift tuple.

\subsubsection{Layer shifts and measurement sets}
\label{subsubsec:layer-shifts}

\begin{definition}
    \label{def:shift-tuple}
    Let \(P=(S,\I,\iota,y)\) be an Alice--Bob completion.  A \emph{shift tuple} for \(P\) is an injective function $\alpha:\I\rightarrow[0,1)$ such that $\alpha(0)=0$.
    The layers are ordered so that
    \[
        i<i'
        \quad\Longleftrightarrow\quad
        \min\iota^{-1}(i)<_{\mathrm{lex}}
        \min\iota^{-1}(i').
    \]
    We denote the set of shift tuples for \(P\) by \(\mathsf{Shift}(P)\).
\end{definition}

\begin{definition}
    \label{def:physical-measurement-sets}
    Let \(\Gamma=(N,t,s_0,s_1,\mu)\) be a valid Charlie clock, let \(P\) be an Alice--Bob completion over \(\Gamma\), and let \(\alpha\in\mathsf{Shift}(P)\).  The associated \emph{physical Alice and Bob measurement sets} are
    \[
        M_1(\Gamma,P,\alpha)
        :=
        \bigcup_{i\in \I}
        \frac{\pi}{N}\bigl(\J_i(P)+\alpha(i)\bigr),
    \]
    and
    \[
        M_2(\Gamma,P,\alpha)
        :=
        \bigcup_{i\in \I}
        \frac{\pi}{N}\bigl(\K_i(P)+\mu-
        \alpha(i)\bigr),
    \]
    with angles taken modulo \(\pi\).
\end{definition}

\begin{proposition}
    \label{prop:shift-normalisation}
    After fixing a normalised valid clock and a canonical Alice--Bob completion, the normalisation in Definition~\ref{def:shift-tuple} removes exactly the global rotation freedom and the relabelling of layers.
\end{proposition}

\begin{proof}
    A local phase rotation of Alice together with the opposite phase rotation of Bob adds the same constant to all layer shifts.  The condition \(\alpha(0)=0\) fixes this global freedom.  The layers themselves are otherwise only names, and the lexicographic convention orders them canonically by the first selected total Charlie assignment in the corresponding fibre of \(\iota\).  Once the clock has been normalised, no further layer ambiguity remains.
\end{proof}

Thus, after the clock and canonical completion are fixed, the shift tuple is exactly the remaining continuous parameter in the proof.

\subsubsection{The bijection}
\label{subsubsec:classification-bijection}

We can now define the combinatorial data that classify a biconditional parity proof up to the physical equivalence of Section \ref{sec:equivalence}, eliminating the possible redundancy created by any freedom of how to label the three parties.

\begin{definition}
    \label{def:classification-datum}
    A \emph{classification datum} is a triple
    \[
        (\Gamma,P,\alpha)
    \]
    where
    \[
        \Gamma\in\mathsf{Clock},
        \qquad
        P\in\mathsf{Comp}(\Gamma),
        \qquad
        \alpha\in\mathsf{Shift}(P).
    \]
    The \emph{canonical classification datum} of a physical equivalence class of a biconditional parity proof is the lexicographically least such triple obtained from the finitely many admissible choices of which party is Charlie,
    of the Alice--Bob order, and, when the valid-clock normalisation does not
    already distinguish them, of the residual equatorial orientation.  We denote the set of canonical classification data triples as \[
        \mathsf{CanonicalTriples}.
    \]
\end{definition}

\begin{theorem}[Classification of biconditional parity proofs]
    \label{thm:bpp-classification}
    The construction
    \[
        (\Gamma,P,\alpha)
        \longmapsto
        \Big(\Bsimple,\, \big(M_1(\Gamma,P,\alpha), M_2(\Gamma,P,\alpha),\{C_0,C_1\}\big)\Big)
    \]
    induces a bijection
    \[
        \BPP \quad
        \cong \quad \mathsf{CanonicalTriples}  \quad\subset \quad
        \bigsqcup_{\Gamma\in\mathsf{Clock}}\;\;
        \bigsqcup_{P\in\mathsf{Comp}(\Gamma)}\;
        \mathsf{Shift}(P).
    \]
\end{theorem}

\begin{proof}
    We first prove surjectivity of the construction.  Let \((\Gamma,P,\alpha)\) be a classification datum.  Since \(\Gamma\) is valid, it is realised by an interpolant state and two Charlie measurements, and all paired cosets of its return subgroups lift to witness paths.  Since \(P\) is canonical, Theorem~\ref{thm:canonicity-minimality} shows that its measurement support is minimal and paradoxical.  The shift tuple \(\alpha\) places this finite support on the equator.  Hence the datum produces a minimal biconditional parity proof.

    For injectivity, start with a minimal biconditional parity proof.  Proposition~\ref{prop:clock-layer-normal-form} recovers its normalised Charlie clock and layer decomposition.  The clock is valid by the existence of witness paths and the quantum realisability of the original scenario.  Theorem~\ref{thm:canonicity-minimality} recovers the unique canonical Alice--Bob completion by selecting precisely the uniquely witnessed total Charlie assignments.  Proposition~\ref{prop:shift-normalisation} recovers the unique normalised shift tuple.  Thus every minimal proof determines exactly one classification datum, and the two constructions are inverse to one another.
\end{proof}

This completes the classification.  Every minimal biconditional parity proof is encoded by a valid finite Charlie clock, a canonical Alice--Bob completion in terms of coset representatives, and a normalised tuple of real-valued layer shifts.

\subsubsection{Examples}\label{sec:examples}
Below, we give concrete examples of biconditional parity proofs. We exhibit a valid Charlie clock from each of the four cases in Theorem~\ref{thm:valid-clocks}, and we explicitly show how each paradox materialises from its canonical classification datum.

\begin{example}[\textbf{(a) A (different) GHZ paradox}]\label{ex:first}
    Let $N = 12$ and $v=1$, so that the Charlie clock is
    \begin{equation*}
        \Gamma = (12, -1, 12, 12, 0).
    \end{equation*}
    These data correspond to the Charlie measurements $C_0 = 0$, $C_1 = \frac{\pi}{12}$. The tick pairs are
    \begin{equation*}
        \Tpair_0 = (0, \pi),\quad \Tpair_1 = (-\tfrac{\pi}{12}, \tfrac{11\pi}{12}),
    \end{equation*}
    and the tick indices are
    \begin{gather*}
        t_{0,0} = 0,\quad t_{0,1} = 12,\\
        t_{1,0} = -1,\quad t_{1,1} = 11.
    \end{gather*}
    We derive the following values of $\HH_\zz$ and $\dd_\zz$ for each $\zz \in Q$:
    \begin{center}
        \begin{tabular}{c|c|c}
            $\zz$   & $\HH_\zz$ & $\dd_\zz$ \\\hline
            $(0,0)$ & $-1$      & $1$       \\
            $(0,1)$ & $11$      & $1$       \\
            $(1,0)$ & $-13$     & $1$       \\
            $(1,1)$ & $-1$      & $1$
        \end{tabular}
    \end{center}


    We have $D_\zz = \Z_{12}$ for all $\zz$. Therefore, there is one Alice--Bob measurement layer, and Alice's and Bob's measurement sets must be, up to equal and opposite rotations, $M_1 = M_2 = \frac{\pi}{12}\Z_{12}$. The Alice--Bob completion $P$ is
    \begin{gather*}
        S = Q,\quad \I = \{0\},\quad \iota(\zz) = 0\ \forall \zz \in S,\quad
        y_\zz = 0\ \forall \zz \in S,
    \end{gather*}
    with $\operatorname{Sol}_P(\zz) = \{(0,0)\}$ for all $\zz \in Q$, and the shift tuple is fixed as $\alpha(0) = 0$.
\end{example}

\begin{example}[\textbf{(b) Non-GHZ with $X$}]\label{ex:second}
    Let $N = 12$, $s = 5$, and $v=3$, so that the Charlie clock is
    \begin{equation*}
        \Gamma = (12, -3, 12, 5, 0).
    \end{equation*}
    The tick pairs are
    \begin{equation*}
        \Tpair_0 = (0, \pi),\quad \Tpair_1 = (-\tfrac{\pi}{4}, \tfrac{\pi}{6}),
    \end{equation*}
    and the tick indices are
    \begin{gather*}
        t_{0,0} = 0,\quad t_{0,1} = 12,\\
        t_{1,0} = -3,\quad t_{1,1} = 2.
    \end{gather*}
    We derive the following values of $\HH_\zz$ and $\dd_\zz$, and the parity quotient $\Pi_\zz$, for each $\zz \in Q$:
    \begin{center}
        \centering
        \begin{tabular}{c|c|c|c}
            $\zz$   & $\HH_\zz$ & $\dd_\zz$ & $\Pi_\zz$ \\\hline
            $(0,0)$ & $-3$      & $3$       & $\Z_2$    \\
            $(0,1)$ & $2$       & $2$       & $\Z_2$    \\
            $(1,0)$ & $-15$     & $3$       & $\Z_2$    \\
            $(1,1)$ & $-10$     & $2$       & $\Z_2$
        \end{tabular}
    \end{center}

    One possible canonical Alice--Bob completion $P$ is given by
    \begin{gather*}
        S = Q,\quad \I = \{0, 1\},\\
        \iota((0,0)) = \iota((1,0)) = 0,\quad \iota((0,1)) = \iota((1,1)) = 1,\\
        y_{(0,0)} = y_{(1,0)} = 1,\quad y_{(0,1)} = y_{(1,1)} = 0,
    \end{gather*}
    inducing the Alice and Bob measurement indices
    \begin{gather*}
        \J_0 = 1 + 3\Z_{12}, \quad \J_1 = 2\Z_{12},\\
        \K_0 = 0 - 1 + 3\Z_{12} = 2 + 3\Z_{12},\quad \K_1 = 0 + 2\Z_{12} = 2\Z_{12},
    \end{gather*}
    noting that $u_{(0,0)} = u_{(0,1)} = 0$. Indeed, it can be verified that the solution sets for $\zz \in Q$ are
    \begin{center}
        \centering
        \begin{tabular}{c|c}
            $\zz$   & $\operatorname{Sol}_P(\zz)$ \\\hline
            $(0,0)$ & $\{(0,1)\}$                 \\
            $(0,1)$ & $\{(1,0)\}$                 \\
            $(1,0)$ & $\{(0,1)\}$                 \\
            $(1,1)$ & $\{(1,0)\}$
        \end{tabular}
    \end{center}

    One possible choice of shift tuple is $\alpha(0) = 0$, $\alpha(1) = 1/3$.
\end{example}

\begin{example}[\textbf{(c) Non-GHZ with equal spread}]
    Let $N = 10$, $s = 9$, and $t=6$, so that the Charlie clock is
    \begin{equation*}
        \Gamma = (10, 6, 9, 9, 25/2).
    \end{equation*}
    The tick pairs are
    \begin{equation*}
        \Tpair_0 = (-\tfrac{3\pi}{4}, \tfrac{3\pi}{20}),\quad \Tpair_1 = (-\tfrac{3\pi}{20}, \tfrac{3\pi}{4}),
    \end{equation*}
    and the tick indices are
    \begin{gather*}
        t_{0,0} = 0,\quad t_{0,1} = 9,\\
        t_{1,0} = 6,\quad t_{1,1} = 15.
    \end{gather*}
    We derive the following values of $\HH_\zz$ and $\dd_\zz$ for each $\zz \in Q$:
    \begin{center}
        \centering
        \begin{tabular}{c|c|c}
            $\zz$   & $\HH_\zz$ & $\dd_\zz$ \\\hline
            $(0,0)$ & $6$       & $2$       \\
            $(0,1)$ & $15$      & $5$       \\
            $(1,0)$ & $-3$      & $1$       \\
            $(1,1)$ & $6$       & $2$
        \end{tabular}
    \end{center}

    Note that $D_{(1,0)} = \Z_{10}$, which immediately implies that there is one Alice--Bob measurement layer. Alice's and Bob's measurement sets must be, up to equal and opposite rotations, $M_1 = \frac{\pi}{10}\Z_{10}$ and $M_2 = \frac{\pi}{10}(\Z_{10} + \frac{1}{2})$ respectively. The Alice--Bob completion $P$ is
    \begin{gather*}
        S = \{(1,0)\},\quad \I = \{0\},\quad \iota((1,0)) = 0,\quad
        y_{(1,0)} = 0,
    \end{gather*}
    with $\operatorname{Sol}_P(\zz) = \{(0,r) \mid r \in \Z_{10}/D_\zz\}$ for all $\zz \in Q$, and the shift tuple is fixed as $\alpha(0) = 0$.
\end{example}

\begin{example}[\textbf{(d) Non-GHZ with unequal spread}]\label{ex:last}
    Let $N = 42$, $t = 2$, $s_0 = 4$, and $s_1 = 5$. It can be computationally verified that $\Theta = 0.981\ldots$, so the Charlie clock
    \begin{equation*}
        \Gamma = (42, 2, 4, 5, \mu)
    \end{equation*}
    is realisable for a unique $\mu \in [0, 84)$ (numerically, $\mu = 81.566\ldots$).
    The tick pairs are
    \begin{equation*}
        \Tpair_0 = \left(\tfrac{\pi}{42} \mu,\, \tfrac{\pi}{42}(\mu + 4)\right),\quad \Tpair_1 = \left(\tfrac{\pi}{42}(\mu + 2),\, \tfrac{\pi}{42}(\mu + 7)\right)
    \end{equation*}
    and the tick indices are
    \begin{gather*}
        t_{0,0} = 0,\quad t_{0,1} = 4,\\
        t_{1,0} = 2,\quad t_{1,1} = 7.
    \end{gather*}
    We derive the following values of $\HH_\zz$ and $\dd_\zz$, and the parity quotient $\Pi_\zz$, for each $\zz \in Q$:
    \begin{center}
        \centering
        \begin{tabular}{c|c|c|c}
            $\zz$   & $\HH_\zz$ & $\dd_\zz$ & $\Pi_\zz$ \\\hline
            $(0,0)$ & $2$       & $2$       & $0$       \\
            $(0,1)$ & $7$       & $7$       & $\Z_2$    \\
            $(1,0)$ & $-2$      & $2$       & $\Z_2$    \\
            $(1,1)$ & $3$       & $3$       & $\Z_2$
        \end{tabular}
    \end{center}

    One possible canonical Alice--Bob completion $P$ is given by
    \begin{gather*}
        S = Q,\quad \I = \{0, 1\},\\
        \iota((0,0)) = \iota((0,1)) = \iota((1,0)) = 0,\quad \iota((1,1)) = 1,\\
        y_{(0,0)} = y_{(1,0)} = 1,\quad y_{(0,1)} = 3,\quad y_{(1,1)} = 0,
    \end{gather*}
    inducing the Alice and Bob measurement indices
    \begin{gather*}
        \J_0 = (1 + 2\Z_{42}) \cup (3 + 7\Z_{42}), \quad \J_1 = 3\Z_{42},\\
        \K_0 = (1 + 2\Z_{42}) \cup (4 + 7\Z_{42}),\quad \K_1 = 1 + 3\Z_{42},
    \end{gather*}
    noting that $u_{(0,0)} = 0$, $u_{(0,1)} = 4$. Indeed, it can be verified that the solution sets for $\zz \in Q$ are
    \begin{center}
        \centering
        \begin{tabular}{c|c}
            $\zz$   & $\operatorname{Sol}_P(\zz)$ \\\hline
            $(0,0)$ & $\{(0,1)\}$                 \\
            $(0,1)$ & $\{(0,3)\}$                 \\
            $(1,0)$ & $\{(0,1)\}$                 \\
            $(1,1)$ & $\{(1,0)\}$
        \end{tabular}
    \end{center}

    One possible choice of shift tuple is $\alpha(0) = 0$, $\alpha(1) = 1/\sqrt{2}$.
\end{example}

We illustrate one witness path for each of these four paradoxes in Figure~\ref{fig:two-charlie-implication-diagrams} in Appendix~\ref{sec:two-charlie-witness-paths}.

\section{Interpolant-state paradoxes with more than two Charlie measurements}
\label{sec:6}

In this section, we investigate interpolant-state scenarios with at least three Charlie measurements and identify new classes of previously unknown, more exotic forms of nonlocality paradox.  These were discovered by using computer-aided searches applied to the novel graph-theoretic formalism developed in Section \ref{sec:3}.  Examples of minimal three-qubit nonlocality paradoxes with more than two Charlie measurements were not known to exist prior to this.

\subsection{A new GHZ-state paradox}
\label{sec:6.1}

For the GHZ state, the impossibility equation \eqref{eqn:betaABC} simplifies to
\begin{equation}\label{eq:beta_GHZ}
    A+B+C \equiv (a\oplus b\oplus c \oplus 1) \pi .
\end{equation}
Thus, if $A+B+C$ is an integer multiple of $\pi$, the outcomes given by a compatible global assignment must satisfy the possible parity constraint \begin{equation}\label{eq:GHZ-outcome}
    a \oplus b \oplus c \equiv \frac{A+B+C}{\pi} \mod{2},
\end{equation}
whereas if $A+B+C$ is not a multiple of $\pi$, the context has no impossible events, hence contributes no parity constraint.

Let us consider the quantum scenario $(\ket{\mathrm{GHZ}}, \M)$ with \begin{equation*}
    M_i = \{ \tfrac{\pi}{4},\, \tfrac{\pi}{2},\, \tfrac{3\pi}{4} \}
    \qquad \text{for each } i=1,2,3. \end{equation*}
Given a global assignment, we denote the corresponding outcomes by $a_0,a_1,a_2$ for Alice, and similarly by $b_k$ and $c_l$ for Bob and Charlie. By \eqref{eq:GHZ-outcome}, the associated system of $\mathbb{Z}_2$-linear equations for this scenario is as follows:
\begin{equation}\label{eq:GHZ-system2}
    \widetilde{\Psi}_{\GHZ} = \begin{cases}
        a_0 \oplus b_0 \oplus c_1 = 1, \quad a_0 \oplus b_1 \oplus c_0 = 1, \\
        a_1 \oplus b_0 \oplus c_0 = 1, \quad a_1 \oplus b_2 \oplus c_2 = 0, \\
        a_2 \oplus b_1 \oplus c_2 = 0, \quad
        a_2 \oplus b_2 \oplus c_1 = 0.\end{cases}
\end{equation}
Summing all six equations in \eqref{eq:GHZ-system2} yields $0 = 1$, establishing paradoxicality. Notably, the scenario $(\ket{\GHZ}, \M)$ is \emph{not} maximally impossible in the sense of \cite[Definition~3.5]{de_Silva_2025}, illustrating the rich structure of nonlocality paradoxes even for the simple GHZ state.

\subsection{Three Charlie measurements: symmetric Alice--Bob case}
\label{sec:6.2}

We establish the existence of new exotic families of nonlocality paradoxes involving interpolant states $\Bsimple$ with $\lambda \neq 0$. We refer to a scenario $(\Bsimple, \M)$ as an $(N_1, N_2, N_3)$-scenario if $|M_i| = N_i$ for $i=1,2,3$. In this subsection, we will deal with the case $N_1 = N_2$.

Fix an even integer $N:=2n \ge 4$. For an integer $0< m\le n-2$, define the interpolant-state parameter
\begin{equation*}
    \lambda_{N,m} := \frac{\pi}{2} - \frac{m\pi}{N} \in (0,\tfrac{\pi}{2}).
\end{equation*}
Note that different pairs $(N,m)$ may yield the same value of $\lambda_{N,m}$ (for example, $\lambda_{12,4}=\lambda_{24,8}=\frac{\pi}{6}$).

Define the set
\begin{math}
    J := \frac{2\pi}{N} \Z_n
\end{math}
and the measurement scenario
\begin{equation*}
    \begin{aligned}
        M_1 := J \cup (u + J), \quad
        M_2 := v - M_1 , \quad
        \text{and} \quad M_3 := \{ C_0, C_1, C_2 \} ;
    \end{aligned}
\end{equation*}
where all angles are taken modulo $\pi$ with \[\begin{aligned}
        u := \big(v - \beta(\lambda_{N,m},C_2)\big) \mod{\tfrac{2\pi}{N}},\quad   v := \beta\big(\lambda_{N,m},\tfrac{\pi}{2}\big) \mod{\pi} , \\
        C_0 (N, m) := \arcsin\!\left(\frac{\tan\lambda_{N,m+1}}{\tan\lambda_{N,m}}\right), \quad C_1 := \tfrac{\pi}{2}, \quad \text{and} \quad C_2 = \pi - C_0.
    \end{aligned}   \]
This choice of $u$ ensures that $|M_1|=N$, and by symmetry $|M_2|=N$.
Since $m\le n-2$, we have $\lambda_{N,m+1}\in(0,\frac{\pi}{2})$ and
$\tan\lambda_{N,m} > \tan\lambda_{N,m+1}$, so the argument of $\arcsin$ lies in $(0,1)$ and $C_0$ is well defined.

\begin{remark}
    For $(N,m)=(12,3)$ we have $u \equiv 0 \mod{\frac{\pi}{6}}$, and hence $|M_1|=|M_2|=6$; we therefore omit this degenerate case.
\end{remark}

\begin{theorem}\label{thm:NN3}
    Let $\nu_2(N)$ denote the $2$-adic valuation of $N$.
    For $N\ge 6$ and $1\le m \le n-2$, the above non-degenerate $(N,N,3)$-scenario $(\ket{\mathrm{B}(\lambda_{N,m})}, \M)$ is a nonlocality paradox if and only if
    \[ 2^{\nu_2(N)} \nmid m. \]
\end{theorem}
\begin{proof}
    See Appendix~\ref{sec:proof-NN3}.
\end{proof}

\subsection{Three Charlie measurements: asymmetric Alice–Bob case}
\label{sec:6.3}

Fix an integer $p\ge 1$, $N := 4p$, an odd integer $1\le m < 2p$, and set
\[ \theta := \frac{m\pi}{N} \in (0, \tfrac{\pi}{2} ) .\]
Define the interpolant parameter $\lambda$ and the angle $C$ by
\begin{equation}\label{eq:param_intp_3A}
    \begin{aligned}
        u             & := \sqrt{\tfrac{\tan{\theta}}{\tan{(\theta/2)}}}, \qquad v:= \sqrt{\tan{\theta} \tan{(\theta/2)}}, \\
        \lambda_{p,m} & := 2\arctan \big( \tfrac{u-1}{u+1} \big), \qquad C:= 2\arctan{v} .
    \end{aligned}
\end{equation}
Since $u>1$ and $v>0$, it follows that $\lambda \in \big( 0, \piovertwo)$ and $C\in (0,\pi)$.

Define the measurement sets (modulo $\pi$):
\begin{equation*}
    \begin{aligned}
        M_1 & := \big\{j\tfrac{\pi}{N} \mid j \in \Z_{N} \big\},       \\
        M_2 & := \big\{k\tfrac{\pi}{N} \mid k\in2\mathbb Z_{2p}\big\}, \\
        M_3 & := \{C_0:=0,\ C_1:=C,\ C_2:=\pi-C \},
    \end{aligned}
\end{equation*}
so that $|M_1| = 4p = N$ and $|M_2|= 2p = \frac{N}{2}$.

\begin{lemma}\label{lem:beta-delta-4N2N3}
    Let $T_{l,z} = \beta(\lambda_{p,m}, C_l+z\pi)$. Then $T_{0,z} \equiv z\pi$,  \begin{equation*}
        \begin{aligned}
            \qquad T_{1,0} & \equiv -m\tfrac{\pi}{N}, \quad       &  & T_{1,1} \equiv (N-2m)\tfrac{\pi}{N} , \\
            T_{2,0}        & \equiv (2m - N) \tfrac{\pi}{N},\quad &  & T_{2,1} \equiv m \tfrac{\pi}{N} .
        \end{aligned}
    \end{equation*}
\end{lemma}
\begin{proof}
    See Appendix~\ref{sec:proof-beta-delta-4N2N3}.
\end{proof}

\begin{theorem}\label{thm:4P2p3}
    For every choice of parameters $(N=4p, m)$ above, the quantum scenario $(\ket{\mathrm{B}(\lambda_{p,m})}, \M)$ is a nonlocality paradox.

    Moreover, it is minimal if and only if $\gcd (N, m)=1$.
\end{theorem}
\begin{proof}
    See Appendix~\ref{sec:proof-4P2p3}.
\end{proof}

\subsection{Four Charlie measurements}
\label{sec:6.4}

In contrast to the previous cases, we do not obtain a uniform family here, primarily because the parameter $\lambda$ and the angles $C_l$ do not admit a closed-form description in this setting.

We take $M_1, M_2 \subseteq \frac{\pi}{12}\Z_{12} \subset [0,\pi)$ and write $A_j :=j\frac{\pi}{12}$ and $B_k := k\frac{\pi}{12}$ for the Alice and Bob measurement angles, respectively.
Define \[ h_* = \frac{\sqrt{6}- \sqrt{2}}{2}, \qquad \lambda_* := \arcsin{h_*} , \qquad u_* = \frac{h_*}{1+\sqrt{1-h_*^2}} ;\]
and $C_*, C'_* \in (0,\pi)$ by \[ e^{iC_*} = \frac{1-u_* e^{i23\pi/12}}{e^{i23\pi/12} - u_*} , \qquad e^{iC'_*} = \frac{1-u_* e^{i21\pi/12}}{e^{i21\pi/12} - u_*} .  \]

We set the Charlie's measurement set as follows: \[ M_3 := \{ C_0 :=C_*,\ C_1 :=C_*',\ C_2 :=\pi-C_*,\ C_3 :=\pi-C_*' \}. \]

\begin{example}\label{ex:four-charlie}
    For this choice of $M_3$, the scenarios $\big(\ket{\mathrm{B}(\lambda_*)}, \M)\big)$ below, arising from different choices of $M_1$ and $M_2$, are nonlocality paradoxes. Here we only list the measurement cardinalities $(|M_1|, |M_2|, |M_3|)$.  Specifically, we obtain
    \begin{tabenum}
        \tabenumitem a $(4,4,4)$-scenario;
        \tabenumitem a $(6,6,4)$-scenario;

        \tabenumitem a $(6,2,4)$-scenario;
        \tabenumitem a $(7,5,4)$-scenario.
    \end{tabenum}
    Further details are provided in Appendix~\ref{sec:B}, including the Alice--Bob measurement settings.
\end{example}

\section{A nonlocality paradox beyond interpolant states}
\label{sec:counterexample}

In this section, we present the first example of a nonlocality paradox arising from a state outside the interpolant family; indeed, the existence of a non-interpolant-state paradox was previously conjectured to not hold. Since the measurement scenario includes only two Charlie measurements, the paradox remains biconditional. However, its proof does not admit a reformulation as a parity proof, and therefore lies outside the classification of biconditional parity proofs given in Section~\ref{sec:two-charlie}.

For this section, it is convenient to rewrite the additive impossibility condition \eqref{eqn:imposs_} in multiplicative form. Under this reformulation, the quantum condition that an event has probability zero becomes a product equation rather than an additive congruence. To this end, we identify each $\beta$ value with its corresponding point on the circle. To witness paradoxicality, we construct finite implication cycles that force an Alice literal and its complement to lie in the same strongly connected component.

\subsection{The unit circle reformulation}
\label{sec:4.1}

Recall that a literal is a pair $(\varphi, z)$ with $\varphi \in [0,\pi)$ an equatorial measurement angle and $z\in \Z_2$.

\begin{definition}
    Given a balanced-state parameter $\lambda \in \piovertwointerval$, and a literal $(\varphi, z)$, define the tick $T_{\lambda} (\varphi + z\pi)$ of the literal:
    \begin{equation}\label{eq:circle-tick}
        T_{\lambda} (\varphi + z\pi) := e^{i\beta(\lambda, \varphi+z\pi)} = \frac{\braket{\varphi+z\pi | w_{\lambda}}}{\braket{\varphi+z\pi | v_{\lambda}}} = \frac{\sin{\frac{\lambda}{2}}+ e^{-i(\varphi +z\pi)} \cos{\frac{\lambda}{2}}}{\cos{\frac{\lambda}{2}} + e^{-i(\varphi +z\pi)} \sin{\frac{\lambda}{2}}} \in \T,
    \end{equation}
    where $\T = \{\zeta \in \mathbb{C} \mid |\zeta| = 1 \}$ is the unit circle in the complex plane.
\end{definition}
It is straightforward to see that the quotient in \eqref{eq:circle-tick} has modulus $1$, since its numerator and denominator have the same modulus. Indeed,
\[ \left|\sin\tfrac{\lambda}{2} \pm e^{-i\varphi} \cos\tfrac{\lambda}{2} \right|^2 = \sin^2\tfrac{\lambda}{2}+\cos^2\tfrac{\lambda}{2} \pm 2\sin\tfrac{\lambda}{2} \cos\tfrac{\lambda}{2} \cos\varphi = \left|\cos\tfrac{\lambda}{2} \pm e^{-i\varphi} \sin\tfrac{\lambda}{2} \right|^2 . \]
Thus, each such quotient defines a point on the unit circle.

Moreover, each tick determines a unique literal, since the map \[ w \longmapsto \frac{\sin{\frac{\lambda}{2}}+ w \cos{\frac{\lambda}{2}}}{\cos{\frac{\lambda}{2}} + w \sin{\frac{\lambda}{2}} } \] is a fractional-linear automorphism of $\T$. We can now restate the impossibility condition in this formalism.

\begin{lemma}
    Consider the context $(A, B, C)$ in the quantum scenario $(\Bstate, \M)$. Then the event $(A,B,C) \to (a,b,c)$ is impossible if and only if \begin{equation}\label{eq:imposs-ticks}
        T_{\lambda_1} (A+a\pi)\, T_{\lambda_2} (B+b\pi)\, T_{\lambda_3} (C+c\pi) = e^{i(\pi-\Phi)} = -e^{-i\Phi} .
    \end{equation}
\end{lemma}

Note that $\beta(\lambda, \varphi)$, where $\beta$ is as in \eqref{eqn:beta}, is the argument of $T_{\lambda} (\varphi)$. Hence, \eqref{eqn:betaABC} is exactly the argument form of \eqref{eq:imposs-ticks}.

It will also be useful to relate the tick $T_\lambda(\varphi)$ to the tick of the complementary literal, namely $T_\lambda(\varphi+\pi)$.

\begin{lemma}
    \label{lem:comp-tick}
    For $\lambda \in \piovertwointerval$, let $F_\lambda : \T \to \T$ be the map defined by
    \begin{equation}
        \label{eq:tick-involution}
        F_{\lambda} (\zeta) := \frac{\sin{\lambda}-\zeta}{1-\zeta\sin{\lambda}}.
    \end{equation}
    Let $\zeta=T_\lambda(\varphi)$ be the tick corresponding to the literal $(\varphi,0)$. Then the complementary literal $(\varphi,1)$ has tick $F_{\lambda} (\zeta)$.

    Moreover, the map $F_{\lambda}$ is an involution of $\mathbb T$.
\end{lemma}
\begin{proof}
    See Appendix~\ref{sec:proof-comp-tick}.
\end{proof}

The complement map $F_{\lambda}$ on the unit circle corresponds to an additive involution $\mathsf{F}_{\lambda}: \R/2\pi \Z \to \R/2\pi \Z$ on phases defined by $e^{i\mathsf{F}_{\lambda} (\xi)} = F_{\lambda} (e^{i\xi})$. Equivalently, \[ \mathsf{F}_{\lambda} (\xi) = \xi + \hat{\delta} (\xi),  \]
where $\hat\delta (\xi)$ is the $\delta$-shift associated to the unique $\varphi$ such that $\beta(\lambda, \varphi) = \xi$.

Fix a Charlie conditioning $(C_l,z_l)$. An Alice tick $T_1$ and a Bob tick $T_2$ determine an impossible event precisely when \[ T_1\, T_2 = -e^{i\Phi} (T_{\lambda_3} (C_l + z_l \pi))^{-1} =: \gamma_{l,z}.\]
Since the example has two Charlie measurements, each total Charlie assignment $\zz \in \Z_2^2$ determines two associated target products: $\gamma_{0,z_0}$ and $\gamma_{1,z_l} \in \T$.

For $\gamma \in \T$, define the reflection $R_{\gamma} : \T \to \T $ by \begin{equation}
    R_{\gamma} (\zeta) := \overline{\zeta} \gamma.
\end{equation}
Note that $\zeta R_\gamma (\zeta) = \gamma$. Thus $R_\gamma (\zeta)$ is the unique tick paired with $\zeta$ to hit the target product $\gamma$.

We also assume that the balanced-state parameters for Alice and Bob coincide, as will be the case in the example below. Thus, set $\alpha := \lambda_1 = \lambda_2$. Fix a total Charlie assignment $\zz$, and write $\gamma_0 := \gamma_{0,z_0}$ and $\gamma_1 := \gamma_{1,z_1}$. We define the Alice return map $P_{\zz}: \T \to \T$ as follows:

\begin{equation}
    \label{eq:A-A-map} P_{\zz} := F_{\alpha} \circ R_{\gamma_1} \circ F_{\alpha} \circ R_{\gamma_0} .
\end{equation}
An application of $P$ corresponds to passing through a first impossible event, taking the complementary tick, then passing through a second impossible event and taking the complementary tick once more.

\begin{lemma}\label{lem:cycle-cert}
    Suppose $P_{\zz}$ has finite order $N$ on $\T$, and suppose that for some $\eta \in \T$ and some integer $m$ with $0<m<N$, \begin{equation}\label{eq:cycle-cert}
        P_{\zz}^m (\eta) = F_{\alpha}(\eta).
    \end{equation}
    Include the Alice measurements whose literal ticks are $\eta_j := P_{\zz}^j (\eta)$ and the Bob measurements whose literal ticks are $\vartheta_j := R_{\gamma_0} (\eta_j)$ for $0\le j < N$. Then the associated $2$-CNF formula is unsatisfiable.
\end{lemma}
\begin{proof}
    See Appendix~\ref{sec:proof-cycle-cert}.
\end{proof}

\subsection{The non-interpolant paradox example}
\label{sec:4.2}

\subsubsection{The balanced state and Charlie's measurements}
Let us set
\[ \rho:=\sqrt{2}-1, \qquad \alpha:=\arcsin(\rho). \]
We take Alice and Bob to have the same balanced-state parameter: $\lambda_1=\lambda_2:=\alpha$.
Thus, their common complement map is
\[ F(\zeta):= F_\alpha(\zeta) = \frac{\rho-\zeta}{1-\rho\zeta}. \]
For Charlie, we take the balanced-state parameter $\lambda_3:=\frac{\pi}{6}$, and the two equatorial measurements
\[ C_0:=0, \qquad C_1:=\frac{\pi}{2}. \]
Finally, take the balanced-state phase to be $\Phi:= \pi$. The state is therefore \begin{equation}
    \label{eq:bal-state-counterex} \ket{\mathrm{B} ((\alpha, \alpha, \tfrac{\pi}{6}), \pi)} = \frac{1}{\sqrt{2(1-\rho^2/2)}} \left( \ket{v_\alpha}\ket{v_\alpha}\ket{v_{\pi/6}} - \ket{w_\alpha}\ket{w_\alpha}\ket{w_{\pi/6}} \right) .
\end{equation}
Charlie's ticks are as follows:
\[ T_{\frac{\pi}{6}}(0) = 1,\ T_{\frac{\pi}{6}}(\pi) = -1 , \quad T_{\frac{\pi}{6}}(\tfrac{\pi}{2}) = e^{-i\frac{\pi}{3}},\ T_{\frac{\pi}{6}}(\tfrac{\pi}{2}+\pi) = e^{i\frac{\pi}{3}}.   \]
Since $\Phi=\pi$, the impossibility condition \eqref{eq:imposs-ticks} reduces to the requirement that the product of the three ticks be equal to $1$. Thus, after fixing a Charlie literal with tick $\chi$, the corresponding target product for the Alice and Bob ticks is $\chi^{-1}$. Therefore the target products for $C_0$ are $\gamma_{0,0} = 1, \ \gamma_{0,1} = -1$, and for $C_1$ are $\gamma_{1,0} = e^{i\frac{\pi}{3}} ,\ \gamma_{1,1} = e^{-i\frac{\pi}{3}}$.

\subsubsection{The four Alice return maps}
We will verify that, for each $\zz \in \Z_2^2$, the map $P_{\zz}$ has finite order and that there exists a tick $\zeta_{\zz}\in\T$ satisfying \eqref{eq:cycle-cert}. We also note that the reflection map $R_\gamma$ may be written as $\zeta \longmapsto \gamma/\zeta$, since, for $\zeta\in\mathbb T$, one has $\overline{\zeta}=\zeta^{-1}$. Thus $R_\gamma$ may be viewed as a Möbius transformation, represented by the linear fractional matrix \[ M_{\gamma} = \begin{pmatrix}
        0 & \gamma \\ 1 & 0
    \end{pmatrix} , \]
where matrices are understood projectively. Similarly, the complement map $F$ is represented by \[ M_{F} = \begin{pmatrix}
        -1 & \rho \\ -\rho & 1
    \end{pmatrix} .   \]
Consequently, the Alice return map $P_{\zz}$ is represented projectively by \begin{equation}\label{eq:matrix-return-map}
    M_{\zz} := M_F M_{\gamma_{1,z_1}} M_F M_{\gamma_{0,z_0}} .
\end{equation}

Let $\gamma_{0,z_0} = e^{i\theta_{0,z_0}}$ and $\gamma_{1,z_1} = e^{i\theta_{1,z_1}}$, using representatives \[ \theta_{0,z_0} \in \{0, \pi\}, \qquad \theta_{1,z_1} \in \{\tfrac{\pi}{3}, -\tfrac{\pi}{3} \}.  \]
Multiplying the matrices in \eqref{eq:matrix-return-map} and dividing the trace by $\sqrt{\det M_{\zz}}$ gives the normalised trace of $M_\zz$: \begin{equation}\label{eq:norm-trace}
    \tau_{z_0,z_1} := \frac{\operatorname{Tr}(M_\zz)}{\sqrt{\operatorname{det}(M_\zz)}} = 2\,\frac{\cos{\left( \frac{\theta_{0,z_0}-\theta_{1,z_1}}{2} \right)} - \rho^2 \cos{\left(\frac{\theta_{0,z_0}+\theta_{1,z_1}}{2}\right)}}{1-\rho^2} .
\end{equation}
Substituting $\rho^2 = 3 - 2\sqrt{2}$ and $\frac{1+\rho^2}{1-\rho^2} = \sqrt{2}$ into \eqref{eq:norm-trace} gives \begin{equation}
    \tau_{0,0} = \tau_{0,1} = 2\cos{\tfrac{\pi}{6}}, \qquad \tau_{1,0} = - 2\cos{\tfrac{\pi}{4}}, \qquad \tau_{1,1} = 2\cos{\tfrac{\pi}{4}}.
\end{equation}
Let $M$ be a determinant-one representative of a projective Möbius transformation. If $\operatorname{tr}(M)=2\cos\theta$, then the characteristic polynomial of $M$ is $\xi^2-2\cos\theta\,\xi+1$, so the eigenvalues are $e^{i\theta}, \ e^{-i\theta}$. The projective order is the smallest positive integer $n$ such that $M^n$ is a scalar multiple of the identity. Equivalently, $e^{in\theta}=e^{-in\theta}$ or $n\theta\in \pi\Z$. Hence, if $\theta/\pi=k/N$ in lowest terms, the projective order is $N$. Therefore \begin{equation}
    \operatorname{ord} (P_{0,0}) = \operatorname{ord} (P_{0,1}) = 6, \qquad \operatorname{ord} (P_{1,0}) = \operatorname{ord} (P_{1,1}) = 4.
\end{equation}

\subsubsection{The starting tick}

It remains to set the starting ticks $\eta_\zz \in \T$ that satisfy \eqref{eq:cycle-cert} for each $\zz\in \Z_2^2$. For $\zz \in \{ (0,0), (0,1)\}$, we take $\eta_{0,0} = \eta_{0,1} = 1$. Note that \[ M_{0,0} = M_{0,1} = 8\begin{pmatrix}
        10-7\sqrt{2} & 7-5\sqrt{2} \\ -7+5\sqrt{2} & -10+7\sqrt{2}.
    \end{pmatrix} \]
Hence, \[ P_{0,0}^3 (1) = P_{0,1}^3 (1) = \frac{(10-7\sqrt{2}) + (7-5\sqrt{2})}{(-7+5\sqrt{2}) + (-10+7\sqrt{2})} = -1  \]
Since $F(1) = \frac{\rho -1}{1- \rho} = -1$, we have \begin{equation}\label{eq:P00}
    P_{0,0}^3 (\eta_{0,0}) = F(\eta_{0,0}), \qquad P_{0,1}^3 (\eta_{0,1}) = F(\eta_{0,1}).
\end{equation}

For $\zz\in \{ (1,0), (1,1) \}$, i.e.\ the two order-four maps $P_{1,0}$ and $P_{1,1}$, we choose the starting ticks by solving \eqref{eq:cycle-cert} explicitly. Let
\[ h(t):=(3-2\sqrt2)t^2+\sqrt3(10-7\sqrt2)t+(7-5\sqrt2). \]
Since its discriminant $\Delta=430-304\sqrt2$ is positive, $h$ has real roots. We choose
\[ t_* := \frac{-\sqrt3(10-7\sqrt2)+\sqrt{430-304\sqrt2}}{2(3-2\sqrt2)}, \]
so that $h(t_*)=0$.

Let us set
\[ \eta(t):=\frac{1+it}{1-it}. \]
For $t\in \R$, $|1+it|=|1-it|$ and hence $|\eta(t)|=1$. We set
\[ \eta_{1,0}:=\eta(t_*), \qquad \eta_{1,1}:=\eta(-t_*) = \overline{\eta_{1,0}}. \]
It remains to verify \eqref{eq:cycle-cert}. Write
\[M_{1,0}^2= \begin{pmatrix}
        a & b \\ c & d
    \end{pmatrix}.\]
Since $F(\zeta)=(\rho-\zeta)/(1-\rho \zeta)$, the equation $P_{1,0}^2(\zeta)=F(\zeta)$
is equivalent, after clearing denominators, to
\[ (-\rho a+c)\zeta^2 + (a-\rho b-\rho c+d)\zeta + (b-\rho d)=0. \]
Substituting $\zeta=\eta(t)$ and multiplying by $(1-it)^2$ simplifies the left-hand side to
\[ 4(1+i\sqrt3)h(t).\]
Thus, $h(t_*)=0$ implies
\begin{equation}\label{eq:P10}
    P_{1,0}^2(\eta_{1,0})=F(\eta_{1,0}).
\end{equation}
Finally, $P_{1,1}$ is obtained from $P_{1,0}$ by complex conjugating the target products. Since $F$ has real coefficients, conjugating \eqref{eq:P10} gives
\begin{equation}\label{eq:P11}
    P_{1,1}^2(\eta_{1,1})=F(\eta_{1,1}).
\end{equation}
Collecting \eqref{eq:P00}--\eqref{eq:P11}, each total Charlie assignment $\zz$ has the data required by the conditions of Lemma~\ref{lem:cycle-cert} (see Table~\ref{tab:charlie-data-counterexample}).
\begin{table*}[ht]
    \centering
    \begin{tabular}{c|c|c|c}
        $\zz$   & $\operatorname{ord} (P_\zz)$ & $m_\zz$ & $\eta_\zz$   \\\hline
        $(0,0)$ & $6$                          & $3$     & $1$          \\
        $(0,1)$ & $6$                          & $3$     & $1$          \\
        $(1,0)$ & $4$                          & $2$     & $\eta(t_*)$  \\
        $(1,1)$ & $4$                          & $2$     & $\eta(-t_*)$
    \end{tabular}
    \caption{Finite-order data for the four total Charlie assignments. For each $\zz\in \Z_2^2$, the table records the order of the Alice return map $P_\zz$, the exponent $m_\zz$ witnessing the  relation $P_\zz^{m_\zz}(\eta_\zz)=F_\alpha(\eta_\zz)$, and the corresponding initial tick $\eta_\zz$.}
    \label{tab:charlie-data-counterexample}
\end{table*}

\subsubsection{Proof of paradoxicality}

For a tick $\zeta \in \T$, let $\mathsf{m}(\zeta) \in [0,\pi)$ denote the equatorial measurement whose two literal ticks are $\zeta$ and $F(\zeta)$. This is well-defined because $T_\alpha$ is a bijection of the unit circle. We now define the measurement scenario for Alice and Bob. As above, Charlie's measurement set is \begin{equation}
    \label{eq:M_3} M_3 = \left\{0,\, \frac{\pi}{2} \right\} .
\end{equation}
For Alice, we take all measurements appearing in the four Alice return orbits:
\begin{equation}\label{eq:M_1}
    M_1 := \big\{ \mathsf{m}(P_{\zz}^r(\eta_{\zz})) \mid \zz \in \Z_2^2,\ 0\le r < \operatorname{ord} (P_{\zz}) \big\} .
\end{equation}
For Bob, we take the corresponding reflected ticks for the first Charlie conditioning:
\begin{equation}\label{eq:M_2}
    M_2 := \big\{ \mathsf{m}(R_{\gamma_{0,z_0}} \circ P_{\zz}^r(\eta_{\zz})) \mid \zz \in \Z_2^2,\ 0\le r < \operatorname{ord} (P_{\zz}) \big\} .
\end{equation}
These sets are finite because each return map $P_{\zz}$ has finite order.

We now verify that the balanced state \eqref{eq:bal-state-counterex} and the above measurement scenario indeed witness a nonlocality paradox.
\begin{theorem}
    The quantum scenario $\big( \ket{\mathrm{B} ((\alpha, \alpha, \tfrac{\pi}{6}), \pi)}, (M_1, M_2, M_3) \big)$ with measurement sets defined in \eqref{eq:M_3}--\eqref{eq:M_2}, is a nonlocality paradox.
\end{theorem}
\begin{proof}
    Fix a total Charlie assignment $\zz\in\mathbb Z_2^2$. The corresponding target products for Alice and Bob are precisely $\gamma_{0,z_0}$ and $\gamma_{1,z_1}$ as computed in the previous subsection. Let $P_{\zz}$ denote the Alice return map associated with this total Charlie assignment.

    Table~\ref{tab:charlie-data-counterexample} provides a starting tick $\eta_{\zz}$, the order $N_\zz = \operatorname{ord} (P_\zz)$, and an exponent $m_{\zz}$ such that \[ 0 < m_\zz < N_\zz, \qquad P_\zz^{m_\zz} (\eta_\zz) = F_\alpha (\eta_\zz) . \]
    By the definitions of $M_1$ and $M_2$, all Alice and Bob measurements required by Lemma~\ref{lem:cycle-cert} for this choice of $\zz$ are included in the measurement scenario. Hence, the associated $2$-CNF formula $\Omega(\zz)$ contains the unsatisfiable subformula by Lemma~\ref{lem:cycle-cert}. Therefore, $\Omega(\zz)$ itself is unsatisfiable.

    Since this holds for every total Charlie assignment $\zz\in\mathbb Z_2^2$, Lemma~\ref{lem:2CNF-paradox} implies that the quantum scenario is a nonlocality paradox.
\end{proof}

We have shown that the above scenario is paradoxical. We now prove that the underlying state is not equivalent to an interpolant state.

\begin{proposition}
    \label{prop:not-interpolant}
    The state $\ket{\mathrm{B} ((\alpha, \alpha, \tfrac{\pi}{6}), \pi)}$ is not equivalent, under local unitaries and permutations of qubits, to any interpolant state.
\end{proposition}
\begin{proof}
    See Appendix~\ref{sec:proof-not-interpolant}.
\end{proof}

\section*{Statements and declarations}

\paragraph*{Funding.}
This work was supported by the Canada Research Chairs Program, the
Natural Sciences and Engineering Research Council of Canada (NSERC)
Discovery Grant RGPIN-2022-03103, and the NSERC--European Commission project FoQaCiA.

\paragraph*{Competing interests.}
The authors have no competing interests to declare that are relevant to the content of this article.

\paragraph*{Data availability.}
All data generated or analysed during this study are included in this article. No external datasets were analysed.

\paragraph*{Code availability.}
Computer-aided searches were performed with the help of artificial intelligence tools; the outputs were checked by the authors, and relevant cases were included in the article.  Code written by the authors earlier in the research process is not directly relevant to the final version of this work and therefore not included.

\paragraph*{Author contributions.}
All authors contributed to the conception and design of the work, the development of the results, and the preparation of the manuscript. All authors read and approved the final manuscript.

\paragraph*{Use of artificial intelligence.}
Commercially available large language models were used during the research process for discussion, drafting, creating diagrams, and editorial support. Specifically, mathematical discussions with an LLM were helpful for refining the technical details of some of our proposed ideas for the canonical triple construction---particularly the shadow test for Alice--Bob completions---and for formalising and streamlining the notation. Additionally, some of the parity proofs with more than two Charlie measurements and the non-interpolant-state counterexample were initially found by an LLM following our prompting, then checked and refined by us. Throughout the process, we carefully analysed, verified, simplified, and made corrections to all LLM outputs that were relevant to our results. All authors take full responsibility for the content, results, proofs, and exposition in this article.

\bibliographystyle{plainnat}
\bibliography{sn-bibliography}

\appendix

\section{Proofs of select results}
\label{sec:A}

\subsection{Proofs of results in Section~\ref{sec:two-charlie}}
\subsubsection{Proof of Lemma~\ref{lem:tick-pair-sin-lambda}}
\label{sec:proof-tick-pair-sin-lambda}
\begin{proof}
    Let $u := \tan(\lambda/2)$ and $z := e^{i C}$. We can write the Charlie tick value $T_{C,0} \equiv -\tau$ in the following compact form after setting it as the argument of a unit complex number, $e^{-i\tau} = e^{i\beta(\lambda, C)}$:
    \begin{equation}\label{eqn:e-i-beta}
        e^{i \beta(\lambda, C)} = \frac{\cos(\lambda/2) + \sin(\lambda/2) e^{iC}}{\sin(\lambda/2) + \cos(\lambda/2)e^{iC}} = \frac{1 + uz}{u + z}.
    \end{equation}
    Combining this equation with the following equation for $e^{i(\sigma - \tau)} = e^{i\beta(\lambda, C + \pi)}$,
    \begin{equation}\label{eqn:e-i-beta-pi}
        e^{i\beta(\lambda, C + \pi)} = \frac{1 - uz}{u - z},
    \end{equation}
    we can eliminate $z$ and obtain
    \begin{equation*}
        u^2 - 2\rho u + 1 = 0,\quad\text{where}\quad\rho := \frac{\cos(-\tau + \sigma/2)}{\cos(\sigma/2)}.
    \end{equation*}
    Rearranging this to $u + \frac{1}{u} = 2\rho$, we find
    \begin{equation*}
        \sin\lambda = \frac{2u}{1+u^2} = \frac{1}{\rho},
    \end{equation*}
    as required.
\end{proof}

\subsubsection{Remainder of proof of Proposition~\ref{prop:one-tick-pair}}
\label{sec:one-tick-pair}
\begin{proof}
    For the reverse direction of 3, let $(-\tau, \sigma - \tau)$, $0 < \tau < \sigma < \pi$, be a tick pair. We claim that $\lambda = \Lambda(\tau, \sigma)$. Indeed, let $u := \tan(\lambda/2)$, $E := e^{-i\tau}$, and
    \begin{align}
                      & z := \frac{1 - uE}{E - u}\qquad\qquad\label{eqn:z-def} \\
        \implies\quad & E = \frac{1 + uz}{u+z}.\nonumber
    \end{align}
    We have $|z| = 1$, so $z = e^{iC}$ for a unique $C \in [0,2\pi)$. Then, by \eqref{eqn:e-i-beta}, we have $\beta(\lambda, C) = -\tau$, verifying the first value of the tick pair and confirming that $C \in (0,\pi)$. For the second value, first substitute \eqref{eqn:z-def} into \eqref{eqn:e-i-beta-pi} to obtain
    \begin{equation*}
        e^{i\beta(\lambda, C+\pi)} = \frac{E - \sin\lambda}{E\sin\lambda - 1}.
    \end{equation*}
    Next, rewrite \eqref{eqn:sin-lambda} using exponentials to obtain
    \begin{equation*}
        \sin\lambda = \frac{e^{i(\sigma-\tau)} + E}{Ee^{i(\sigma-\tau)} + 1}.
    \end{equation*}
    Comparing these two equations, we see that $e^{i\beta(\lambda, C+\pi)} = e^{i(\sigma-\tau)}$, thus verifying the second value of the tick pair.

\end{proof}

\subsection{Proofs of results in Section~\ref{sec:6}}
\subsubsection{Proof of Theorem~\ref{thm:NN3}}
\label{sec:proof-NN3}

\begin{proof}
    With the choices of Charlie's measurements, one obtains the following tick matrix $(T_{l,z})_{l,z}$ and tick indices $(t_{l,z})_{l,z}$: \[ \mathsf{T} =  \begin{pmatrix}
            -T_{2,0} - \frac{2(m+1)\pi}{N} & - T_{2,0} \\ -\frac{m\pi}{N} & \frac{m\pi}{N} \\ T_{2,0} & T_{2,0} + \frac{2(m+1)\pi}{N}
        \end{pmatrix} \quad \text{and} \quad \tickinds = \begin{pmatrix}
            -2 & 2m \\ 0 & 2m \\ 0 & 2(m+1)
        \end{pmatrix} \]
    with $\mu_0 = \alpha - m,\ \mu_1 = -m,\ \mu_2 = -\alpha - m$, where $\alpha = - \beta(\lambda_{N,m}, C_2) - m$, so that $T_{l,z} = \frac{\pi}{N} (\mu_l + t_{l,z})$. Let us denote the Alice and Bob measurements as follows: \[ \begin{aligned}
            A_{0,j} := \pioverN j \in \A_0 & , \quad A_{1,j} := \pioverN (\alpha +j) \in \A_1 ; \\ B_{0,k} := \pioverN (-m+k) \in \B_0 &,\quad  B_{1,k} := \pioverN (-\alpha -m+k) \in \B_1;
        \end{aligned} \]
    so that $M_1 = \A_0 \cup \A_1$ and $M_2 = \B_0 \cup \B_1$. Note that, in both of these measurement layers, we have $j,k \in 2\Z_N \subseteq \Z_N$.

    Fix a total Charlie assignment $\zz = (z_0,z_1,z_2) \in \Z_2^3$ and consider the implication graph $I_\Psi(\zz)$. Each edge of this graph arises from one of the three Charlie conditionings $(C_l, z_l)$ for $l \in \Z_3$. We derive the Alice--Bob measurement matchings resulting from each Charlie conditioning, using the impossibility equation~\eqref{eqn:imposs_interpolant}.

    \begin{itemize}
        \item For $l = 1$, \eqref{eqn:imposs_interpolant} is satisfied only if Alice's and Bob's measurements $A_{i,j}, B_{i',k}$ belong to the same layer, i.e.\ $i = i'$, so that the impossibility equation becomes
              \begin{equation}\label{eqn:symmetric-l-1-imposs}
                  j + k \equiv 2mz_1 + (1 \oplus a \oplus b) N \mod{2N}.
              \end{equation}
              Now, the argument that proves Lemma~\ref{lem:edges-vs-eqns} can be readily generalised to show that any values of $j,k,a,b$ that satisfy~\eqref{eqn:symmetric-l-1-imposs} imply the existence of a directed edge $(A_{i,j}, a) \Rightarrow (B_{i,k}, b \oplus 1)$ in the implication graph, with $k \equiv 2mz_1 - j \mod{N}$.

        \item For $l = 0$, Alice's and Bob's measurements $A_{i,j}, B_{i',k}$ must now belong to different layers: specifically, $i = 1$ and $i' = 0$. The impossibility equation becomes
              \begin{equation}\label{eqn:symmetric-l-0-imposs}
                  j + k \equiv 2(m+1)z_0 - 2 + (1 \oplus a \oplus b) N \mod{2N}.
              \end{equation}
              Despite the different layers, the argument that proves Lemma~\ref{lem:edges-vs-eqns} can still be generalised in the same way as above to show that any values of $j,k,a,b$ that satisfy~\eqref{eqn:symmetric-l-0-imposs} imply the existence of a directed edge $(A_{1,j}, a) \Rightarrow (B_{0,k}, b \oplus 1)$, with $k \equiv 2(m+1)z_0 - 2 - j \mod{N}$.

        \item For $l = 2$, Alice's and Bob's measurements $A_{i,j}, B_{i',k}$ must also belong to different layers: this time, $i = 0$ and $i' = 1$. The impossibility equation becomes
              \begin{equation}\label{eqn:symmetric-l-2-imposs}
                  j + k \equiv 2(m+1)z_2 + (1 \oplus a \oplus b) N \mod{2N}.
              \end{equation}
              Hence, any values of $j,k,a,b$ that satisfy~\eqref{eqn:symmetric-l-2-imposs} imply the existence of a directed edge $(A_{0,j}, a) \Rightarrow (B_{1,k}, b \oplus 1)$, with $k \equiv 2(m+1)z_2 - j \mod{N}$.
    \end{itemize}

    Thus, every vertex in the implication graph has degree $2$: one incident edge comes from the $(C_1, z_1)$ conditioning, while the other comes from either the $(C_0, z_0)$ or the $(C_2, z_2)$ conditioning. More explicitly, the Alice--Bob measurement matchings occur in the cyclic order \[ \A_0 \xrightarrow{C_1} \B_0 \xrightarrow{C_0} \A_1 \xrightarrow{C_1} \B_1 \xrightarrow{C_2} \A_0 . \]

    By composing these four matchings using~\eqref{eqn:symmetric-l-1-imposs}--\eqref{eqn:symmetric-l-2-imposs}, we derive the measurement return map
    \begin{equation}
        \label{eq:return-map}
        j\longmapsto j+\HH_{\zz} \mod N,
    \end{equation}
    where
    \begin{equation}
        \label{eq:Hz}
        \HH_{\zz}:= 2(m+1)(z_0+z_2)-4mz_1-2.
    \end{equation}
    Although this return process now involves two Alice and Bob measurement layers rather than just one, we still have well-defined return maps given by $\HH_\zz$. Let $\dd_\zz := \gcd(N,\HH_\zz)$. Applying the measurement return map $N/\dd_\zz$ times, we return to the starting Alice measurement. Note that $\HH_\zz$ is unchanged if we consider the return map on $\A_1$ rather than $\A_0$.

    Suppose the implication graph $I_\Psi(\zz)$ contains a minimal directed path that begins and ends at the Alice measurement $A_{0,j_0}$ for some $j_0 \in 2\Z_N$. We may generalise the relabelling procedure in Remark~\ref{rem:clocks} in this case: this time, there are \emph{two} layers involved in any minimal directed path, and such a path (analogous to~\eqref{eqn:dir-path-relabelled}) may be written as
    \begin{equation}\label{eqn:dir-path-symmetric-a-b}
        \begin{gathered}
            [j_0 + a_0 N]_0 \longrightarrow [k_0 + b_0N]_0 \longrightarrow [j'_0 + a'_0 N]_1 \longrightarrow [k'_0 + b'_0 N]_1 \longrightarrow [j_1 + a_1 N]_0 \longrightarrow \cdots \\\longrightarrow [j_q + a_q N]_0 \longrightarrow \cdots \longrightarrow [j_0 + a_L N]_0,
        \end{gathered}
    \end{equation}
    where $L = N/\dd_\zz$, $j_q,j'_q,k_q,k'_q \in 2\Z_N$, $a_q,a'_q,b_q,b'_q \in \Z_2$, and the subscript on each vertex indicates whether the point is in layer $0$ or $1$.

    Now we make an important observation: the lifting concept detailed in Lemma~\ref{lem:lifting} readily generalises here. The argument is almost identical: taking successive pairs of edges and finding the difference between the corresponding equations from~\eqref{eqn:symmetric-l-1-imposs}--\eqref{eqn:symmetric-l-2-imposs}, we see that the Alice-literal path derived from~\eqref{eqn:dir-path-symmetric-a-b} can be written as
    \begin{equation*}
        [\tilde{j}_0]_0 \longrightarrow \cdots \longrightarrow [\tilde{j}_0 + \HH_\zz]_0 \longrightarrow \cdots \longrightarrow [\tilde{j}_0 + q\HH_\zz]_0 \longrightarrow \cdots \longrightarrow \left[\tilde{j}_0 + \frac{N}{\dd_\zz}\HH_\zz\right]_0,
    \end{equation*}
    where $\tilde{j}_0 := j_0 + a_0 N$.

    After repeating this argument with a starting Alice measurement $A_{1,j_0}$, we therefore deduce the following result, analogous to the last part of Lemma~\ref{lem:lifting}.

    \begin{proposition}
        The implication graph $I_\Psi(\zz)$ contains a witness path if and only if $\HH_\zz/\dd_\zz$ is odd. Thus, the non-degenerate $(N,N,3)$-scenario $(\ket{\mathrm{B}(\lambda_{N,m})}, \M)$ is a nonlocality paradox if and only if $\HH_\zz/\dd_\zz$ is odd for all $\zz \in \Z_2^3$.
    \end{proposition}

    Let $\widetilde{\HH}_\zz := \HH_\zz/2$, so that $\dd_{\zz} = 2\gcd(n,\widetilde{\HH}_\zz)$. Let $s := \nu_2 (N)$, so $\nu_2 (n) = s-1$. We have that ${\HH_\zz}/{\dd_\zz}$ is even exactly when $\nu_2(\HH_\zz)>\nu_2(N)$, i.e.\
    \begin{equation}\label{eq:cons-criterion}
        2^{\nu_2(N)+1}=2^{s+1} \mid \HH_\zz \iff 2^s \mid \widetilde{\HH}_\zz  .
    \end{equation}
    Now explicitly list the possible integer values of $\widetilde{\HH}_\zz$:
    \begin{equation*}
        \widetilde{\HH}_\zz = \begin{cases}
            -1-2z_1m    , & z_0=z_2=0,   \\
            m(1-2z_1)   , & z_0\neq z_2, \\
            1+2(1-z_1)m , & z_0=z_2=1.
        \end{cases}
    \end{equation*}
    In particular, in the first and third cases $\HH_\zz$ is \emph{odd}, so $2^s\nmid \widetilde{\HH}_\zz$ for every $s\ge 1$.
    The \emph{only} case in which $\widetilde{\HH}_\zz$ can be divisible by $2^s$ is when $z_0 \neq z_2$, where $\widetilde{\HH}_\zz = \pm m$.

    Therefore by \eqref{eq:cons-criterion},\begin{itemize}
        \item If $2^s\mid m$, then choose any $\zz$ with $z_0 \neq z_2$; for that $\zz$, one has $2^s\mid \widetilde{\HH}_\zz=\pm m$. Hence $\Psi(\zz)$ is consistent, so the scenario is \emph{not} a paradox.
        \item If $2^s\nmid m$, then for every $\zz$ we have $2^s\nmid \widetilde{\HH}_\zz$, since either $\widetilde{\HH}_\zz$ is odd, or $\widetilde{\HH}_\zz=\pm m$. Hence, by \eqref{eq:cons-criterion} every $\Psi(\zz)$ is inconsistent, so the scenario is a paradox.
    \end{itemize}
\end{proof}

\subsubsection{Proof of Lemma~\ref{lem:beta-delta-4N2N3}}
\label{sec:proof-beta-delta-4N2N3}
\begin{proof}
    Note that $T_{0,z} \equiv z\pi$ follows immediately from Lemma~\ref{lem:beta-facts}. We therefore focus on establishing the expressions for the $C_1$ tick values; the expressions for $C_2 = \pi - C_1$ then follow by symmetry.

    From \eqref{eq:param_intp_3A}, we have \begin{equation}\label{eq:uvrandomfact}
        \tan{\left(\tfrac{\lambda_{p,m}}{2}\right)} = \frac{u-1}{u+1}, \quad \text{and}\quad  \tan{\big(\tfrac{C_1}{2}\big)} = v ;
    \end{equation}
    and hence \[ \tan{\lambda_{p,m}} = \frac{u^2 -1}{2u} . \]
    Also, since $\sin{C_1} = \frac{2v}{1+v^2}$,
    \[ \sin{C_1}\tan{\lambda_{p,m}} = \frac{v(u^2 -1)}{u(1+v^2)} . \]
    Again by \eqref{eq:param_intp_3A}
    \[ \begin{aligned}
            \frac{u^2 -1}{1+v^2} & = \frac{\tan{\theta} - \tan{(\theta/2)}}{\tan{(\theta/2)} + \tan{\theta} \tan^2{(\theta/2)}} \\ &= \frac{1}{\tan{(\theta/2)}} \cdot \frac{\tan{\theta} - \tan{(\theta/2)}}{1 + \tan{\theta} \tan{(\theta/2)}} \\ &= 1,
        \end{aligned} \]
    where the last equality follows from the subtraction and half-angle formulae for the tangent. Therefore \[ \sin{C_1}\tan{\lambda_{p,m}} = \frac{v}{u} = \tan{\big(\tfrac{\theta}{2}\big)}. \]
    Plugging this into the formula for $\delta$ gives \begin{equation}\label{eq:delta-proof}
        \delta (\lambda_{p,m}, C_1) \equiv \pi - 2\arctan{\left(\tan{\big(\tfrac{\theta}{2} \big)}\right)} \equiv \pi - \theta = (N-m)\frac{\pi}{N}.
    \end{equation}

    Note that \eqref{eqn:beta} can be written as
    \begin{equation}\label{eq:beta-proof-1}
        \beta (\lambda_{p,m}, C_1) \equiv C_1 - 2\arctan{\left(\frac{\sin{C_1}}{\tan{(\lambda_{p,m}/2)} + \cos{C_1}}\right)} .
    \end{equation}
    Using $\sin{C_1} = \frac{2v}{1+v^2}$ and $\cos{C_1} = \frac{1-v^2}{1+v^2}$, and a bit of algebra, we can rewrite \eqref{eq:beta-proof-1} as
    \[\begin{aligned}
            \beta (\lambda_{p,m}, C_1) & \equiv  C_1 - 2\arctan{\left(v \cdot \frac{u+1}{u-v^2} \right)} \\ &= C_1 - 2\arctan{\left(\tan{\left( \tfrac{C_1 + \theta}{2} \right)}\right)} \\
                                       & \equiv -\theta = -m\tfrac{\pi}{N}.
        \end{aligned}\]
    The proof can now be completed using the identity \[ T_{C_1, z} = \beta (\lambda_{p,m}, C_1) + z\delta (\lambda_{p,m}, C_1).\]
    We leave the details to the reader.
\end{proof}

\subsubsection{Proof of Theorem~\ref{thm:4P2p3}}
\label{sec:proof-4P2p3}
\begin{proof}
    One obtains the following Charlie tick indices $(t_{l,z})_{l,z}$:
    \begin{equation*}
        \tickinds = \begin{pmatrix}
            0    & N    \\
            -m   & N-2m \\
            2m-N & m
        \end{pmatrix},
    \end{equation*}
    with $T_{l,z} \equiv \pioverN t_{l,z}$. We have $M_1 = \pioverN \Z_N$ and $M_2 = \pioverN (2\Z_N)$, and we denote the Alice and Bob measurements as follows:
    \begin{gather*}
        A_j := j \pioverN,\quad j \in \Z_N;\\
        B_k := k \pioverN,\quad k \in 2\Z_N.
    \end{gather*}

    As in the proof of Theorem~\ref{thm:NN3}, given a total Charlie assignment $\zz \in \Z_2^3$, we find all possible Alice--Bob measurement matchings involved in edges of the implication graph $I_\Psi(\zz)$, by considering each Charlie conditioning $(C_l, z_l)$.

    \begin{itemize}
        \item For $l = 0$, the impossibility equation~\eqref{eqn:imposs_interpolant} becomes
              \begin{equation}\label{eqn:asymmetric-l-0-imposs}
                  j + k \equiv Nz_0 + (1 \oplus a \oplus b)N \mod{2N}.
              \end{equation}
              Since $N$ is even, there exists a directed edge $(A_j, a) \Rightarrow (B_k, b \oplus 1)$ due to the $l = 0$ Charlie conditioning if and only if $j$ is even.

        \item For $l = 1$, the impossibility equation is
              \begin{equation}\label{eqn:asymmetric-l-1-imposs}
                  j + k \equiv -m + (N-m)z_1 + (1 \oplus a \oplus b)N \mod{2N}.
              \end{equation}
              Since $N$ is even and $m$ is odd, there exists a directed edge $(A_j, a) \Rightarrow (B_k, b \oplus 1)$ due to the $l = 1$ Charlie conditioning if and only if $j$ has the same parity as $-m + (N-m)z_1$: in other words, $j$ is odd if $z_1 = 0$, and $j$ is even if $z_1 = 1$.

        \item For $l = 2$, the impossibility equation is
              \begin{equation}\label{eqn:asymmetric-l-2-imposs}
                  j + k \equiv 2m - N + (N-m)z_2 + (1 \oplus a \oplus b)N \mod{2N}.
              \end{equation}
              Since $N$ is even and $m$ is odd, there exists a directed edge $(A_j, a) \Rightarrow (B_k, b \oplus 1)$ due to the $l = 2$ Charlie conditioning if and only if $j$ has the same parity as $2m - N + (N-m)z_2$: in other words, $j$ is even if $z_2 = 0$, and $j$ is odd if $z_2 = 1$.
    \end{itemize}

    Next, we determine the possible minimal directed paths present in $I_\Psi(\zz)$ by determining the possible measurement return maps. We consider two cases depending on $\zz$: when $z_1 = 1$ or $z_2 = 0$, and when $(z_1, z_2) = (0, 1)$.

    \begin{itemize}
        \item When $z_1 = 1$ or $z_2 = 0$, we have the following pairs of Charlie conditionings together with the resulting return maps:
              \begin{alignat*}{3}
                  \text{if } z_1 = 1\colon\quad            & M_1 \xrightarrow{C_0} M_2 \xrightarrow{(C_1, 1)} M_1, \qquad      &  & j \text{ even} \longmapsto j + \HH_1 \mod{N}; \\
                  \text{if } z_2 = 0\colon\quad            & M_1 \xrightarrow{C_0} M_2 \xrightarrow{(C_2, 0)} M_1, \qquad      &  & j \text{ even} \longmapsto j + \HH_2 \mod{N}; \\
                  \text{if } (z_1, z_2) = (1,0)\colon\quad & M_1 \xrightarrow{(C_1, 1)} M_2 \xrightarrow{(C_2, 0)} M_1, \qquad &  & j \text{ even} \longmapsto j + \HH_3 \mod{N};
              \end{alignat*}
              where
              \begin{align*}
                  \HH_1 & := N(1-z_0) - 2m,   \\
                  \HH_2 & := - N(1+z_0) + 2m, \\
                  \HH_3 & := - 2N + 4m.
              \end{align*}
              Note that, in this case, there are no measurement return maps on odd $j$.

        \item When $(z_1, z_2) = (0, 1)$, we have the following pair of Charlie conditionings together with the resulting return map:
              \begin{equation*}
                  M_1 \xrightarrow{(C_1, 0)} M_2 \xrightarrow{(C_2, 1)} M_1, \qquad j \text{ odd} \longmapsto j + \HH_4 \mod{N},
              \end{equation*}
              where $\HH_4 := 2m$. Note that, in this case, there are no measurement return maps on even $j$.
    \end{itemize}

    Each of these four different return maps induces minimal directed paths. Furthermore, just as in the proof of Theorem~\ref{thm:NN3}, by finding the difference between the two relevant equations out of~\eqref{eqn:asymmetric-l-0-imposs}--\eqref{eqn:asymmetric-l-2-imposs}, we may readily generalise the lifting concept detailed in Lemma~\ref{lem:lifting}. For any choice of rotation $\HH_w$, $w \in \{1, 2, 3, 4\}$, there exist Alice-literal paths of the form
    \begin{equation*}
        \tilde{j} \longrightarrow \tilde{j} + \HH_w \longrightarrow \cdots \longrightarrow \tilde{j} + q \HH_w \longrightarrow \cdots \longrightarrow \tilde{j} + \frac{N}{\dd_w}\HH_w,
    \end{equation*}
    where $\dd_w := \gcd(N, \HH_w)$, for all $\tilde{j} := j + aN$, $a \in \Z_2$, such that $j$ is even if $w \in \{1,2,3\}$ and $j$ is odd if $w=4$. Therefore, we deduce the following result, analogous to the last part of Lemma~\ref{lem:lifting}.

    \begin{proposition}
        The implication graph $I_\Psi(\zz)$ contains a witness path if and only if the relevant condition below holds.
        \begin{itemize}
            \item $(z_1, z_2) = (1,1)$: $\HH_1/\dd_1$ is odd;
            \item $(z_1, z_2) = (0,0)$: $\HH_2/\dd_2$ is odd;
            \item $(z_1, z_2) = (1,0)$: at least one of $\HH_1/\dd_1,\, \HH_2/\dd_2,\, \HH_3/\dd_3$ is odd;
            \item $(z_1, z_2) = (0,1)$: $\HH_4/\dd_4$ is odd.
        \end{itemize}
        Thus, the quantum scenario $(\ket{\mathrm{B}(\lambda_{p,m})}, \M)$ is a nonlocality paradox if and only if $\HH_1/\dd_1$, $\HH_2/\dd_2$, and $\HH_4/\dd_4$ are all odd.
    \end{proposition}

    Now, we in fact have $\dd_1 = \dd_2 = \dd_4 = \gcd(N, 2m) =: d$. Since $m$ is odd, we have that $2 \mid d$ but $4 \nmid d$. Since $N$ is a multiple of $4$, it is therefore indeed the case that $\HH_1/\dd_1$, $\HH_2/\dd_2$, and $\HH_4/\dd_4$ are all odd. This proves the paradoxicality part of Theorem~\ref{thm:4P2p3}.

    For the minimality part of the theorem, first note that all three Charlie measurements are required for the existence of witness paths for all total Charlie assignments. Second, note that $\gcd(N, m) = 1 \Leftrightarrow d = \gcd(N, 2m) = 2$. Suppose that $\gcd(N,m) = 1$. Considering any total Charlie assignment $\zz$ with $(z_1, z_2) = (1,1)$, we see that any minimal directed path involves Alice measurements whose indices form the coset $2\Z_N$, so removing any of these even-indexed measurements will destroy all witness paths for these $\zz$. Additionally, any such minimal directed path involves all $N/2$ Bob measurements. Now considering any total Charlie assignment $\zz$ with $(z_1, z_2) = (0,1)$, we see that any minimal directed path involves Alice measurements whose indices form the coset $1 + 2\Z_N$, so removing any of these odd-indexed measurements will destroy all witness paths for these $\zz$. Again, all $N/2$ Bob measurements are involved in any such path. Thus, removing any of Alice's, Bob's, or Charlie's measurements destroys the paradox.

    Conversely, suppose that $\gcd(N,m) > 1$, so that $d$ is an even number greater than $2$. Take any total Charlie assignment with $z_1 = 1$ or $z_2 = 0$; then we have a witness path with the following labelling of Alice literals, for some $w \in \{1,2,3\}$:
    \begin{equation*}
        0 \longrightarrow \HH_w \longrightarrow \cdots \longrightarrow q\HH_w \longrightarrow \cdots \longrightarrow N.
    \end{equation*}
    The Alice measurements present in this witness path have indices that form the coset $d\Z_N$. But since $2 \notin d\Z_N$, we may remove the Alice measurement $A_2$ without destroying this witness path. Furthermore, all of the witness paths for the remaining total Charlie assignments involve odd-indexed Alice measurements, so removing $A_2$ does not affect these witness paths either. Hence we can remove an Alice measurement without destroying the paradox, so the paradox is not minimal.

\end{proof}

\subsection{Proofs of results in Section~\ref{sec:counterexample}}

\subsubsection{Proof of Lemma~\ref{lem:comp-tick}}
\label{sec:proof-comp-tick}
\begin{proof}
    From \eqref{eq:circle-tick}, we have \[ \zeta =T_\lambda(\varphi) = \frac{\sin{\frac{\lambda}{2}} + e^{-i\varphi} \cos{\frac{\lambda}{2}}}{\cos{\frac{\lambda}{2}} + e^{-i\varphi} \sin{\frac{\lambda}{2}}} . \]
    By definition the complementary tick of $\zeta$ is given by
    \[ \zeta' = T_\lambda (\varphi+\pi) = \frac{\sin{\frac{\lambda}{2}}- e^{-i\varphi} \cos{\frac{\lambda}{2}}}{\cos{\frac{\lambda}{2}} - e^{-i\varphi} \sin{\frac{\lambda}{2}}} . \]
    Solving for $e^{-i\varphi}$ in the first equation and substituting it into the second gives \[ \zeta' = \frac{2\sin{\frac{\lambda}{2}}\cos{\frac{\lambda}{2}}- \zeta (\cos^2{\frac{\lambda}{2}}+\sin^2{\frac{\lambda}{2}})}{(\cos^2{\frac{\lambda}{2}}+\sin^2{\frac{\lambda}{2}}) - 2\zeta \sin{\frac{\lambda}{2}}\cos{\frac{\lambda}{2}}} = \frac{\sin{\lambda}-\zeta}{1-\zeta\sin{\lambda}} . \]
    Substituting $F_\lambda (\zeta)$ into the formula for $F_\lambda$ gives $F_\lambda (F_\lambda (\zeta)) = \zeta$.
\end{proof}

\subsubsection{Proof of Lemma~\ref{lem:cycle-cert}}
\label{sec:proof-cycle-cert}
\begin{proof}
    Since each tick determines a unique literal, for the rest of this proof, we will denote the tick $\eta$ as the literal, and $F_{\alpha} (\eta)$ as the complementary literal.
    Since $\vartheta_j := R_{\Gamma_0} (\eta_j)$, we have $\eta_j \vartheta_j = \Gamma_0$. Therefore, under the first Charlie conditioning, the Alice--Bob literal pair $(\eta_j, \vartheta_j)$ is impossible. So, the implication graph contains the implication \begin{equation}\label{eq:impl_1}
        \eta_j \Longrightarrow F(\vartheta_j).
    \end{equation}
    Using the definition of $P_\zz$, \[ \eta_{j+1} = P_\zz (\eta_j) = F_\alpha \big(R_{\Gamma_1} (F_\alpha (R_{\Gamma_0} (\eta_j)))\big) = F_\alpha \big(R_{\Gamma_1} (F_\alpha (\vartheta_j))\big) . \]
    Applying $F_\alpha$ to both sides and using the fact that $F_\alpha$ is an involution, we have \[ F_\alpha (\eta_{j+1}) = R_{\Gamma_1} (F_\alpha (\vartheta_j)) \iff F_\alpha (\vartheta_j) F_\alpha (\eta_{j+1}) = \Gamma_1 .\]
    Thus, under the second Charlie conditioning, the Alice--Bob literal pair $(F_\alpha (\eta_{j+1}), F_\alpha (\vartheta_j))$ is impossible. The associated clause gives the following implication in the implication graph: \begin{equation}\label{eq:impl_2}
        F_\alpha (\vartheta_j) \Longrightarrow \eta_{j+1} .
    \end{equation}
    Combining \eqref{eq:impl_1} and \eqref{eq:impl_2}, we obtain, for each $j$, the following directed path in the implication graph: \begin{equation}\label{eq:impl_3}
        \eta_j \Longrightarrow \cdots \Longrightarrow \eta_{j+1} .
    \end{equation}
    Iterating \eqref{eq:impl_3} $m$ times and using \eqref{eq:cycle-cert}, we obtain the directed implication path \[ \eta_0 \Longrightarrow \cdots \Longrightarrow P_\zz^m (\eta_0) = F_\alpha (\eta_0)   \]
    in the implication graph. Since $P_{\zz}$ has order $N$, we also have \[ P_\zz^{N-m} (F_\alpha (\eta_0)) = P_\zz^{N-m} (P_\zz^m (\eta_0)) = P_\zz^N (\eta_0) = \eta_0 . \]
    Iterating \eqref{eq:impl_3} a further $N-m$ times therefore gives the reverse implication \[ F_\alpha (\eta_0) \Longrightarrow \cdots \Longrightarrow \eta_0 . \]
    Thus, the literal represented by $\eta_0$ and its complement, represented by $F_\alpha(\eta_0)$, lie in the same strongly connected component of the implication graph. By Lemma~\ref{lem:impl-graph}, the associated $2$-CNF formula is unsatisfiable.
\end{proof}

\subsubsection{Proof of Proposition~\ref{prop:not-interpolant}}
\label{sec:proof-not-interpolant}
\begin{proof}
    Local unitaries preserve the eigenvalues of the one-qubit reduced density matrices, and permutations of qubits only permute these eigenvalues. Hence, the number of maximally mixed one-qubit marginals is an invariant of the equivalence relation.

    Every interpolant state has at least two maximally mixed one-qubit marginals. Indeed, up to phase conventions, it has the form
    \[ \mathcal{N} \bigl(\ket{0}\ket{0}\ket{v_\lambda} + \ket{1}\ket{1}\ket{w_\lambda} \bigr), \]
    So, tracing out the other two qubits gives $I/2$ on each of the first two qubits.

    We now compute the one-qubit marginals of $\ket{\mathrm{B} ((\alpha, \alpha, \tfrac{\pi}{6}), \pi)}$. Write
    \[ r_i:=\braket{v_{\lambda_i}\,|\,{w_{\lambda_i}}}=\sin\lambda_i, \qquad R_i:=\prod_{k\neq i} r_k, \qquad R:=r_1r_2r_3 . \]
    Let $\varrho_i$ denote the reduced density matrix obtained by tracing out all qubits except qubit $i$. Since $\Phi=\pi$, a direct trace calculation gives\[
        \varrho_i = \frac{1}{2} I + \frac{r_i-R_i}{2(1-R)}X . \]
    Thus $\varrho_i$ is maximally mixed if and only if $r_i=R_i$.

    For $\ket{\mathrm{B} ((\alpha, \alpha, \tfrac{\pi}{6}), \pi)}$ we have \[
        (r_1,r_2,r_3) = \big(\rho,\rho,\tfrac{1}{2} \big), \qquad \rho=\sqrt{2}-1. \]
    The conditions $r_i=R_i$ would respectively give \[
        \rho=\frac{\rho}{2}, \qquad \rho=\frac{\rho}{2}, \qquad \frac{1}{2}=\rho^2. \]
    All three are false. Hence $\ket{\mathrm{B} ((\alpha, \alpha, \tfrac{\pi}{6}), \pi)}$ has no maximally mixed one-qubit marginal, whereas every interpolant state has at least two. Therefore, $\ket{\mathrm{B} ((\alpha, \alpha, \tfrac{\pi}{6}), \pi)}$ is not equivalent to any interpolant state under local unitaries and permutations of qubits.
\end{proof}

\section[Witness paths for BiCPP examples]{Witness paths for $\BPP$ examples}\label{sec:two-charlie-witness-paths}

\tikzset{
lit/.style={
        draw,
        circle,
        inner sep=1pt,
        minimum size=8mm,
        font=\small
    },
edge/.style={-{Stealth[length=2mm]}, line width=0.55pt},
dashededge/.style={edge, dashed}
}

\newcommand{\Apair}[2]{%
    \node[lit] (A#1z) at (#2,1.25) {$A_{#1}^{0}$};
    \node[lit] (A#1o) at (#2+0.95,1.25) {$A_{#1}^{1}$};
}

\newcommand{\Bpair}[2]{%
    \node[lit] (B#1z) at (#2,-3) {$B_{#1}^{0}$};
    \node[lit] (B#1o) at (#2+0.95,-3) {$B_{#1}^{1}$};
}

\newcommand{\ABpair}[3]{%
    \Apair{#1}{#3}
    \Bpair{#2}{#3}
}

\begin{figure}[H]
    \centering

    \resizebox{0.8\textwidth}{!}{%
        \begin{tikzpicture}

            \ABpair{0}{0}{0}
            \ABpair{1}{1}{2.6}
            \ABpair{2}{2}{5.2}
            \ABpair{10}{10}{10.8}
            \ABpair{11}{11}{13.4}

            \node at (8.0,1.25) {$\cdots$};
            \node at (8.0,-3) {$\cdots$};

            \draw[edge] (A0z) -- (B0z);
            \draw[edge] (B0z) -- (A11o);
            \draw[edge] (A11o) -- (B1z);
            \draw[edge] (B1z) -- (A10o);
            \draw[edge] (A10o) -- (B2z);
            \draw[edge] (A2o) -- (B10z);
            \draw[edge] (B10z) -- (A1o);
            \draw[edge] (A1o) -- (B11z);
            \draw[edge] (B11z) -- (A0o);

            \draw[dashededge] (B2z) -- (8.0,0.95);
            \draw[dashededge] (8.0,0.95) -- (8.0,-2.7);
            \draw[dashededge] (8.0,-2.7) -- (A2o);
        \end{tikzpicture}%
    }

    \vspace{1.1em}

    \resizebox{0.7\textwidth}{!}{%
        \begin{tikzpicture}

            \ABpair{1}{2}{0}
            \ABpair{4}{5}{3.2}
            \ABpair{7}{8}{6.4}
            \ABpair{10}{11}{10.0}

            \draw[edge] (A1z) -- (B11o);
            \draw[edge] (B11o) -- (A10o);
            \draw[edge] (A10o) -- (B2z);
            \draw[edge] (B2z) -- (A7o);
            \draw[edge] (A7o) -- (B5z);
            \draw[edge] (B5z) -- (A4o);
            \draw[edge] (A4o) -- (B8z);
            \draw[edge] (B8z) -- (A1o);
        \end{tikzpicture}%
    }

    \vspace{1.1em}

    \resizebox{0.7\textwidth}{!}{%
        \begin{tikzpicture}

            \ABpair{0}{0}{0}
            \ABpair{2}{2}{2.55}
            \ABpair{4}{4}{5.1}
            \ABpair{6}{6}{7.65}
            \ABpair{8}{8}{10.2}

            \draw[edge] (A0z) -- (B0z);
            \draw[edge] (B0z) -- (A6z);
            \draw[edge] (A6z) -- (B4o);
            \draw[edge] (B4o) -- (A2o);
            \draw[edge] (A2o) -- (B8z);
            \draw[edge] (B8z) -- (A8o);
            \draw[edge] (A8o) -- (B2z);
            \draw[edge] (B2z) -- (A4z);
            \draw[edge] (A4z) -- (B6o);
            \draw[edge] (B6o) -- (A0o);
        \end{tikzpicture}%
    }

    \vspace{1.1em}

    \resizebox{0.8\textwidth}{!}{%
        \begin{tikzpicture}

            \ABpair{1}{1}{0}
            \ABpair{3}{3}{2.6}
            \ABpair{5}{5}{5.2}
            \ABpair{39}{39}{10.8}
            \ABpair{41}{41}{13.4}

            \node at (8.0,1.25) {$\cdots$};
            \node at (8.0,-3) {$\cdots$};

            \draw[edge] (A1z) -- (B41o);
            \draw[edge] (B41o) -- (A3z);
            \draw[edge] (A3z) -- (B39o);
            \draw[edge] (B39o) -- (A5z);

            \draw[edge] (B5o) -- (A39z);
            \draw[edge] (A39z) -- (B3o);
            \draw[edge] (B3o) -- (A41z);
            \draw[edge] (A41z) -- (B1o);
            \draw[edge] (B1o) -- (A1o);

            \draw[dashededge] (A5z) -- (8.0,-2.7);
            \draw[dashededge] (8.0,-2.7) -- (8,0.95);
            \draw[dashededge] (8,0.95) -- (B5o);
        \end{tikzpicture}%
    }

    \caption{From top to bottom, a witness path for each of Examples~\ref{ex:first}--\ref{ex:last}, under the total Charlie assignment $\zz = (0,0)$. We write $A_j^a$ for the literal $(A_j,a)$, and similarly write $B_k^b$ for the literal $(B_k,b)$.}
    \label{fig:two-charlie-implication-diagrams}
\end{figure}
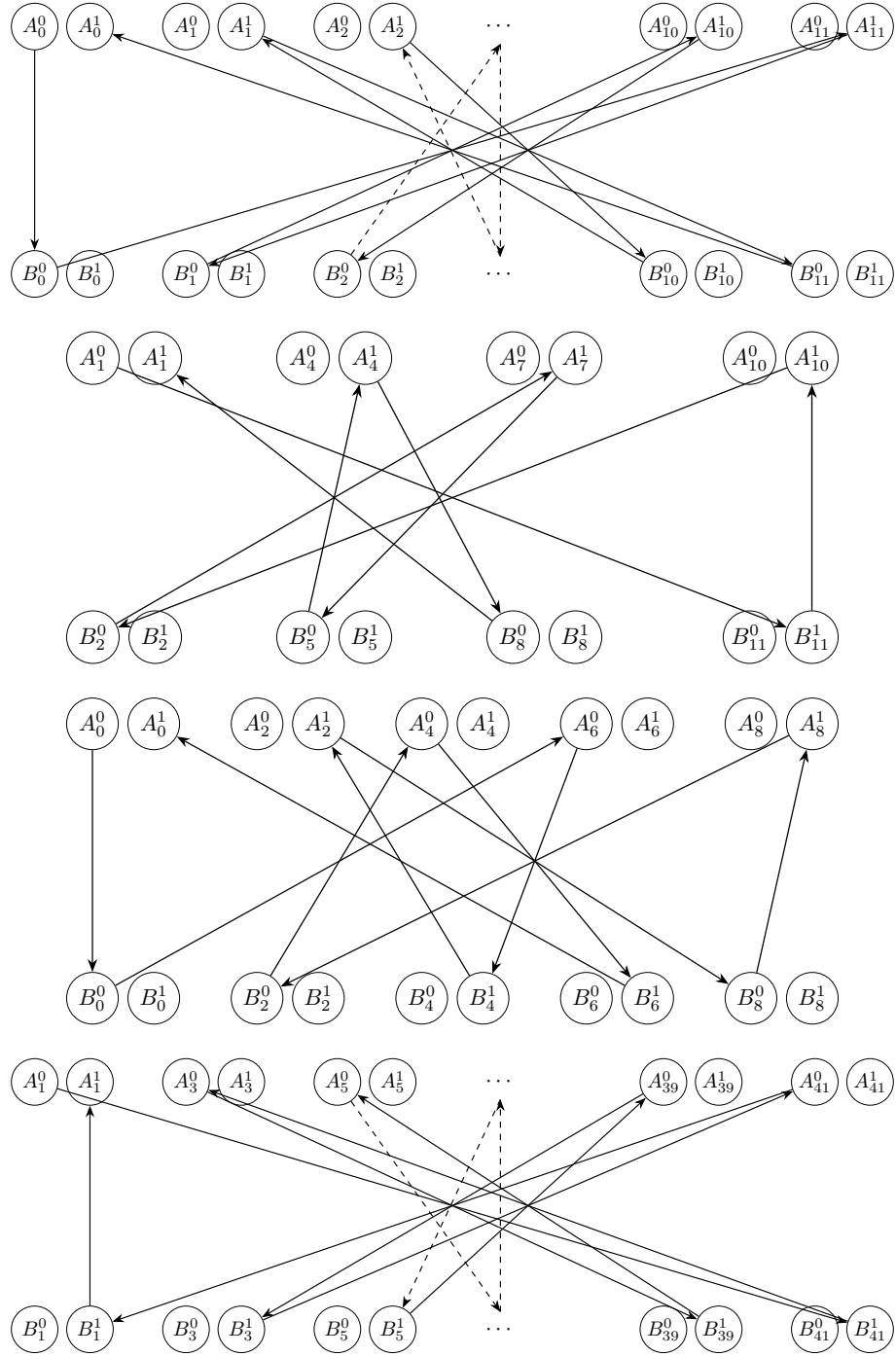

\section{Lists of witness paths for Example~\ref{ex:four-charlie}}
\label{sec:B}

We provide further details for the four scenarios in Example~\ref{ex:four-charlie}. Note that in each scenario, the Charlie tick indices $\tickinds = (t_{l,z})_{l \in \Z_4,\, z \in \Z_2}$ are
\begin{equation*}
    \tickinds = \begin{pmatrix}
        23 & 9 \\
        21 & 5 \\
        15 & 1 \\
        19 & 3
    \end{pmatrix}.
\end{equation*}

For each scenario, we first specify the Alice and Bob measurement sets. We then exhibit, for every total Charlie assignment $\zz\in \Z_2^4$, a witness path in the corresponding implication graph $I_{\Psi}(\zz)$. This proves that the conditioned formula is inconsistent in each case, and hence establishes the desired nonlocality paradox. By Lemma~\ref{lem:interpolant-impl}, it suffices to exhibit a path from some literal to its complement. Throughout this section, we write $A^a$ for the literal $(A,a)$, and similarly write $B^b$ for the literal $(B,b)$.

\begin{remark}\label{rem:inc-systems}
    Observe that $t_{0,1} \equiv t_{1,0} + N \mod{2N}$ and $t_{2,0} \equiv t_{3,1} + N \mod{2N}$. Therefore, if $\zz = (1,0,z_2, z_3)$, then the system $\Psi(\zz)$ is necessarily inconsistent, i.e.\ $I_\Psi(\zz)$ contains a witness path, since any edge $A^a \Rightarrow B^b$ that arises from the $l=0$ Charlie conditioning implies the existence of an edge $B^b \Rightarrow A^{a \oplus 1}$ from the $l=1$ Charlie conditioning (see Lemma~\ref{lem:edges-vs-eqns}). Similarly, if $\zz = (z_0, z_1, 0, 1)$, then $\Psi(\zz)$ is also inconsistent. Thus, we already have that $\Psi(\zz)$ is inconsistent for each of these four total Charlie assignments.
\end{remark}




\subsection{Example~\ref{ex:four-charlie}(1)}

For the $(4,4,4)$-scenario, we have the following sets of Alice and Bob measurements.
\[ \begin{aligned}
        M_1^{(4,4)}                   & := \{ A_0,\ A_2,\ A_4,\ A_6\}, \\
        \text{and} \quad  M_2^{(4,4)} & :=\{B_1,\ B_3,\ B_5,\ B_9 \}.
    \end{aligned}  \]

We give an example of the implication graph for $\zz=(0,0,0,0)$:

\begin{figure}[H]
    \centering
    \begin{tikzpicture}[
        every node/.style={font=\small},
        lit/.style={draw, circle, minimum size=7mm, inner sep=1pt},
        edge/.style={{Latex[length=2mm]}-{Latex[length=2mm]}, very thin, black}
        ]

        \node[lit] (A00) at (0,0)    {$A_0^0$};
        \node[lit] (A01) at (1,0)    {$A_0^1$};

        \node[lit] (A20) at (3,0)    {$A_2^0$};
        \node[lit] (A21) at (4,0)    {$A_2^1$};

        \node[lit] (A40) at (6,0)    {$A_4^0$};
        \node[lit] (A41) at (7,0)    {$A_4^1$};

        \node[lit] (A60) at (9,0)    {$A_6^0$};
        \node[lit] (A61) at (10,0)   {$A_6^1$};

        \node[above=2mm of A21] {\hspace{2cm}\textbf{Alice literals}};

        \node[lit] (B10) at (0,-5)   {$B_1^0$};
        \node[lit] (B11) at (1,-5)   {$B_1^1$};

        \node[lit] (B30) at (3,-5)   {$B_3^0$};
        \node[lit] (B31) at (4,-5)   {$B_3^1$};

        \node[lit] (B50) at (6,-5)   {$B_5^0$};
        \node[lit] (B51) at (7,-5)   {$B_5^1$};

        \node[lit] (B90) at (9,-5)   {$B_9^0$};
        \node[lit] (B91) at (10,-5)  {$B_9^1$};

        \node[below=2mm of B31] {\hspace{2cm} \textbf{Bob literals}};

        \newcommand{\xorEdges}[3]{%
            \ifnum#3=0
                \draw[edge] (A#10) -- (B#20);
                \draw[edge] (A#11) -- (B#21);
            \else
                \draw[edge] (A#10) -- (B#21);
                \draw[edge] (A#11) -- (B#20);
            \fi
        }


        \xorEdges{6}{5}{1}
        \xorEdges{2}{9}{1}

        \xorEdges{0}{9}{1}
        \xorEdges{4}{5}{1}
        \xorEdges{6}{3}{1}

        \xorEdges{0}{3}{1}
        \xorEdges{2}{1}{1}
        \xorEdges{6}{9}{0}

        \xorEdges{2}{5}{1}
        \xorEdges{4}{3}{1}
        \xorEdges{6}{1}{1}

    \end{tikzpicture}
    \caption{Implication graph for the total Charlie assignment $\zz = (0,0,0,0)$ of Example~\ref{ex:four-charlie}(1).}
    \label{fig:4-charlie}
\end{figure}
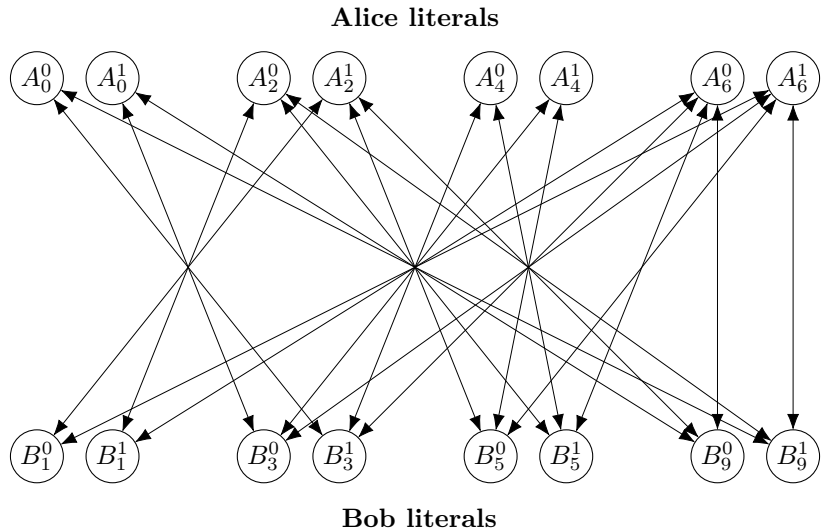

For each of the remaining $\zz$ values not covered by Remark~\ref{rem:inc-systems}, we list a witness path in the corresponding implication graph.
Here, in the last line, $z\in\Z_2$; thus the displayed witness path applies to both $\zz=(1,1,1,0)$ and $\zz=(1,1,1,1)$.

\begin{itemize}
    \item $\zz = (0,0,0,0):\ A_0^0 \Longrightarrow B_3^1 \Longrightarrow A_6^0 \Longrightarrow B_9^0 \Longrightarrow A_0^1$.
    \item $\zz = (0,0,1,0):\ A_0^0 \Longrightarrow B_1^0 \Longrightarrow A_6^1 \Longrightarrow B_5^0 \Longrightarrow A_4^1 \Longrightarrow B_9^0 \Longrightarrow A_0^1$.
    \item $\zz = (0,0,1,1):\ A_0^0 \Longrightarrow B_3^0 \Longrightarrow A_6^1 \Longrightarrow B_9^0 \Longrightarrow A_0^1$.
    \item $\zz = (0,1,0,0):\ A_2^0 \Longrightarrow B_9^1 \Longrightarrow A_6^1 \Longrightarrow B_5^0 \Longrightarrow A_2^1$.
    \item $\zz = (0,1,1,0):\ A_2^0 \Longrightarrow B_9^1 \Longrightarrow A_4^0 \Longrightarrow B_1^0 \Longrightarrow A_6^1 \Longrightarrow B_5^0 \Longrightarrow A_2^1$.
    \item $\zz = (0,1,1,1):\ A_2^0 \Longrightarrow B_9^1 \Longrightarrow A_6^0 \Longrightarrow B_5^1 \Longrightarrow A_0^1 \Longrightarrow B_3^1 \Longrightarrow A_2^1$.
    \item $\zz = (1,1,0,0):\ A_0^0 \Longrightarrow B_9^0 \Longrightarrow A_6^0 \Longrightarrow B_1^1 \Longrightarrow A_4^1 \Longrightarrow B_5^1 \Longrightarrow A_0^1$.
    \item $\zz = (1,1,1,z):\ A_0^0 \Longrightarrow B_9^0 \Longrightarrow A_4^1 \Longrightarrow B_5^1 \Longrightarrow A_0^1$.
\end{itemize}

\subsection{Example~\ref{ex:four-charlie}(2)}

For the $(6,6,4)$-scenario, we have
\[ \begin{aligned}
        M_1^{(6,6)}                  & := \{ A_0,\ A_1,\ A_2,\ A_3,\ A_6,\ A_7\}, \\
        \text{and} \quad M_2^{(6,6)} & :=\{B_0,\ B_1,\ B_2,\ B_5,\ B_6,\ B_9 \}.
    \end{aligned}  \]

Again, we list witness paths from the remaining $\zz$ values not covered in Remark~\ref{rem:inc-systems}.

\begin{itemize}
    \item $\zz = (0,0,0,0):\  A_2^0 \Longrightarrow B_9^1 \Longrightarrow A_6^1 \Longrightarrow B_5^0 \Longrightarrow A_2^1$.

    \item $\zz = (0,0,1,0):\ A_0^0 \Longrightarrow B_1^0 \Longrightarrow A_6^1 \Longrightarrow B_5^0 \Longrightarrow A_2^1 \Longrightarrow B_9^0 \Longrightarrow A_0^1$.

    \item $\zz = (0,0,1,1):\ A_1^0 \Longrightarrow B_2^0 \Longrightarrow A_7^1 \Longrightarrow B_6^0 \Longrightarrow A_3^1 \Longrightarrow B_0^1 \Longrightarrow A_1^1$.

    \item $\zz = (0,1,0,0):\ A_2^0 \Longrightarrow B_9^1 \Longrightarrow A_6^1 \Longrightarrow B_5^0 \Longrightarrow A_2^1$.

    \item $\zz = (0,1,1,0):\ A_1^0 \Longrightarrow B_0^0 \Longrightarrow A_7^1 \Longrightarrow B_6^0 \Longrightarrow A_1^1$.

    \item $\zz = (0,1,1,1):\ A_0^0 \Longrightarrow B_1^0 \Longrightarrow A_2^0 \Longrightarrow B_9^1 \Longrightarrow A_6^0 \Longrightarrow B_5^1 \Longrightarrow A_0^1$.

    \item $\zz = (1,1,0,0):\ A_0^0 \Longrightarrow B_5^0 \Longrightarrow A_2^1 \Longrightarrow B_1^0 \Longrightarrow A_6^1 \Longrightarrow B_9^1 \Longrightarrow A_0^1$.

    \item $\zz = (1,1,1,z):\ A_3^0 \Longrightarrow B_2^0 \Longrightarrow A_7^0 \Longrightarrow B_6^1 \Longrightarrow A_3^1$.
\end{itemize}

\subsection{Example~\ref{ex:four-charlie}(3)}

For the $(6,2,4)$-scenario, we have

\[ \begin{aligned}
        M_1^{(6,2)}                  & := \{ A_0,\ A_2,\ A_4,\ A_6,\ A_8,\ A_{10}\}, \\
        \text{and} \quad M_2^{(6,2)} & :=\{ B_1,\ B_7 \}.
    \end{aligned}  \]

Again, we list witness paths from the remaining $\zz$ values not covered in Remark~\ref{rem:inc-systems}.

\begin{itemize}
    \item $\zz = (0,0,0,0):\ A_2^0 \Longrightarrow B_7^1 \Longrightarrow A_8^1 \Longrightarrow B_1^0 \Longrightarrow A_2^1$.

    \item $\zz = (0,0,1,0):\ A_0^0 \Longrightarrow B_7^1 \Longrightarrow A_6^0 \Longrightarrow B_1^1 \Longrightarrow A_0^1$.

    \item $\zz = (0,0,1,1):\ A_2^0 \Longrightarrow B_7^1 \Longrightarrow A_8^0 \Longrightarrow B_1^1 \Longrightarrow A_2^1$.

    \item $\zz = (0,1,0,0):\ A_4^0 \Longrightarrow B_1^0 \Longrightarrow A_{10}^1 \Longrightarrow B_7^0 \Longrightarrow A_4^1$.

    \item $\zz = (0,1,1,0):\ A_4^0 \Longrightarrow B_1^0 \Longrightarrow A_{10}^1 \Longrightarrow B_7^0 \Longrightarrow A_4^1$.

    \item $\zz = (0,1,1,1):\ A_4^0 \Longrightarrow B_1^0 \Longrightarrow A_{10}^1 \Longrightarrow B_7^0 \Longrightarrow A_4^1$.

    \item $\zz = (1,1,0,0):\ A_2^0 \Longrightarrow B_7^0 \Longrightarrow A_8^0 \Longrightarrow B_1^0 \Longrightarrow A_2^1$.

    \item $\zz = (1,1,1,0):\ A_0^0 \Longrightarrow B_7^1 \Longrightarrow A_6^0 \Longrightarrow B_1^1 \Longrightarrow A_0^1$.

    \item $\zz = (1,1,1,1):\ A_2^0 \Longrightarrow B_7^0 \Longrightarrow A_8^1 \Longrightarrow B_1^1 \Longrightarrow A_2^1$.
\end{itemize}

\subsection{Example~\ref{ex:four-charlie}(4)}

For the $(7,5,4)$-scenario, we have

\[ \begin{aligned}
        M_1^{(7,5)}                  & := \{ A_0,\ A_1,\ A_2,\ A_3,\ A_6,\ A_7,\ A_8 \}, \\
        \text{and} \quad M_2^{(7,5)} & :=\{B_0,\ B_1,\ B_4,\ B_7,\ B_8 \}.
    \end{aligned}  \]

Again, we list witness paths from the remaining $\zz$ values not covered in Remark~\ref{rem:inc-systems}.

\begin{itemize}
    \item $\zz = (0,0,0,0):\ A_3^0 \Longrightarrow B_4^1 \Longrightarrow A_7^0 \Longrightarrow B_8^0 \Longrightarrow A_3^1$.

    \item $\zz = (0,0,1,0):\ A_1^0 \Longrightarrow B_8^1 \Longrightarrow A_3^0 \Longrightarrow B_4^1 \Longrightarrow A_7^0 \Longrightarrow B_0^1 \Longrightarrow A_1^1$.

    \item $\zz = (0,0,1,1):\ A_2^0 \Longrightarrow B_7^1 \Longrightarrow A_8^0 \Longrightarrow B_1^1 \Longrightarrow A_2^1$.

    \item $\zz = (0,1,0,0):\ A_3^0 \Longrightarrow B_4^1 \Longrightarrow A_7^0 \Longrightarrow B_8^0 \Longrightarrow A_3^1$.

    \item $\zz = (0,1,1,0):\ A_0^0 \Longrightarrow B_1^0 \Longrightarrow A_6^1 \Longrightarrow B_7^0 \Longrightarrow A_0^1$.

    \item $\zz = (0,1,1,1):\ A_1^0 \Longrightarrow B_4^0 \Longrightarrow A_7^1 \Longrightarrow B_8^0 \Longrightarrow A_3^1 \Longrightarrow B_0^1 \Longrightarrow A_1^1$.

    \item $\zz = (1,1,0,0):\ A_1^0 \Longrightarrow B_4^0 \Longrightarrow A_3^1 \Longrightarrow B_0^0 \Longrightarrow A_7^1 \Longrightarrow B_8^1 \Longrightarrow A_1^1$.

    \item $\zz = (1,1,1,0):\ A_0^0 \Longrightarrow B_1^0 \Longrightarrow A_6^1 \Longrightarrow B_7^0 \Longrightarrow A_0^1$.

    \item $\zz = (1,1,1,1):\ A_2^0 \Longrightarrow B_7^0 \Longrightarrow A_8^1 \Longrightarrow B_1^1 \Longrightarrow A_2^1$.
\end{itemize}

\end{document}